\documentclass[twocolumn]{svjour3}

\newcommand{\AckGiorgio}{Giorgio Orsi has also been supported by the Oxford Martin School's grant no. LC0910-019.}
\newcommand{\AckStandard}{The research leading to these results has received
    funding from the European Research Council under the European Community's
    Seventh Framework Programme (FP7/2007--2013) / ERC grant agreement DIADEM,
    no.~246858. }
\newcommand{\ACKNOWLEDGEMENTS}{\AckStandard\AckGiorgio}

\makeatletter

\RequirePackage{svn-multi}
\RequirePackage[utf8]{inputenc}
\RequirePackage{amsmath}
\RequirePackage{amssymb}
\RequirePackage{stmaryrd}
\RequirePackage{subfig}
\RequirePackage{graphicx}
\graphicspath{{}}
\RequirePackage{color}
\RequirePackage{xcolor}
\definecolor{purple}{rgb}{0.65, 0.12, 0.82}
\definecolor{lightblue}{rgb}{0,0,.7}
\definecolor{orange}{rgb}{1,.7,0}
\definecolor{darkorange}{rgb}{1,.4,0}
\definecolor{darkgreen}{rgb}{0,.8,0}
\definecolor{darkblue}{rgb}{0,0,.4}
\definecolor{darkred}{rgb}{.4,0,0}
\definecolor{gray}{rgb}{.2,.2,.2}
\definecolor{darkgray}{rgb}{.4,.4,.4}
\definecolor{shadecolor}{gray}{0.975}
\RequirePackage{xspace}
\RequirePackage{url}
\RequirePackage{booktabs}
\RequirePackage{paralist}
\newif\if@restonecol

\RequirePackage[algoruled,vlined,linesnumbered,shortend]{algorithm2e}
\SetAlgoSkip{meddeskip}
\SetNlSty{sffamily}{}{}
\RequirePackage{listings}

\lstdefinelanguage{XQuery}{%
  morekeywords=[1]{for,in,let,where,return,if,then,else,case,satisfies,default,typeswitch},%
  morekeywords=[1]{child,descendant,attribute,self,descendant-or-self,following-sibling,following,parent,ancestor,preceding-sibling,preceding,ancestor-or-self,},%
  morekeywords=[2]{cast,as,castable,instance,of,some,every,in,unordered,ordered,stable,order,by,descending,empty,collation,derives-from,nilled,nillable,node,schema-element,document-node},%
  morekeywords=[3]{element,function,variable,comment,processing-instruction,text,document,attribute,to},%
  morekeywords=[4]{xquery,declare,namespace,module,option,ordering,version,encoding,order,last,copy-namespaces,preserve,no-inherit,import,schema,at,external,construction,strip,base-uri,boundary-space},%
  morekeywords=[3]{fn:count,fn:sum,fn:empty,fn:max,fn:doc,fn:avg,fn:node-name,fn:last,fn:position,fn:error,position,last,size,count,substring,fn:substring,string,fn:string,contains,fn:contains,concat,fn:concat,substring-after,fn:substring-after,fn:string-join,string-join},%
  morekeywords=[5]{mod,eq,or,and,idiv,true,false,is,div,union,intersect,except,gt,lt,},%
  morestring=[b]",%
  morecomment=[s]{(:}{:)},%
  alsoletter={-},
  sensitive%
}[keywords,comments,strings]

\lstdefinelanguage{XPath1}{%
  morekeywords=[1]{child,descendant,attribute,self,descendant-or-self,following-sibling,following,parent,ancestor,preceding-sibling,preceding,ancestor-or-self,desc,foll,prec},%
  morekeywords=[2]{comment,processing-instruction,text,attribute},%
  morekeywords=[3]{fn:count,fn:sum,fn:empty,fn:max,fn:doc,fn:avg,fn:node-name,fn:last,fn:position,fn:error,position,last,size,count,substring,fn:substring,string,fn:string,contains,fn:contains,concat,fn:concat,boolean},%
  morekeywords=[5]{or,and,true,false,union,intersect,except},%
  morestring=[b]",%
  morestring=[b]',%
  morecomment=[s]{(:}{:)},%
  alsoletter={-},
  literate={<----}{$\longleftarrow\;\,\,$}{3},
  sensitive%
}[keywords,comments,strings]

\lstdefinelanguage{OXPath}[]{XPath1}{%
  morekeywords=[1]{style},%
  morekeywords=[2]{field,any-field,location},%
  morekeywords=[3]{},%
  morekeywords=[4]{},%
  morekeywords=[5]{},%
  morecomment=[s]{\{}{\}},%
  commentstyle=\bfseries\color{red},
}[keywords,comments,strings]

\lstdefinelanguage{CSS}{%
  morekeywords=[1]{>,+,*,\#,|,~,not,nth-child,attribute,first-of-type,first-line,last-of-type,last-child,empty,nth-of-type},%
  morekeywords=[2]{@import,@media,@charset},
  morekeywords=[3]{color,font-weight,margin,margin-left,padding-left,border-left,background-color,font-family,text-transform,font-size,border,text-align,background,font-style},%
  morekeywords=[4]{sans-serif,rgb,solid,none,normal,uppercase,bold,print},%
  morestring=[b]",%
  alsoletter={>,+,\#,|,~,-,@},
  alsoother={:,::},
  sensitive%
}[keywords,comments,strings]

\lstdefinelanguage{XMLSchema}[]{XML}{
  morekeywords=[1]{xs:complexType,xs:element,xs:sequence,xs:choice,xs:simpleType,xs:key,xs:selector,xs:unique,xs:keyref,xs:field,xs:restriction,xs:minExclusive,xs:attribute,xs:extension,xs:complexContent,xs:schema},
  morekeywords=[3]{base,name,minOccurs,maxOccurs,type,xpath,refer,use}
}

\lstdefinelanguage{FancyXML}[]{XML}{
  morecomment=[s]{<}{>},%
  commentstyle=\bfseries\color{darkblue},
}

\lstdefinelanguage{SPARQL}{%
  morekeywords=[1]{CONSTRUCT,WHERE,SELECT},
  morekeywords=[2]{AND,FILTER,UNION,OPT,OPTIONAL,MINUS},%
  morekeywords=[3]{sameTerm,isBLANK,isLITERAL,isIRI,BOUND},
  morekeywords=[4]{},
  morekeywords=[5]{},
  morestring=[b]",%
  alsodigit={-},%
}[keywords,strings]

\lstdefinelanguage{TurtleTwo}[]{SPARQL}{%
  morekeywords=[1]{@prefix},
  morekeywords=[2]{rdf,rdfs,fb},%
  morekeywords=[5]{type,subClassOf,domain,range,subPropertyOf},
  morekeywords=[6]{hasPart,linksTo,title,genre,author},
}[keywords,strings]

\definecolor{purple}{rgb}{0.65, 0.12, 0.82}
\definecolor{flexred}{rgb}{0.65, 0.01, 0.01}
\definecolor{flexgreen}{rgb}{0, 0.6, 0}
\definecolor{flexgray}{rgb}{0.25, 0.37, 0.75}
\definecolor{flexblue}{rgb}{0, 0.2, 1}
\definecolor{flexfunction}{rgb}{0.2, 0.6, 0.4}
\definecolor{flexvar}{rgb}{0.4, 0.6, 0.8}

\lstdefinelanguage{JavaScript} {
  sensitive=true,
  morecomment=[l][\color{flexgreen}\ttfamily]{//},
  morecomment=[s][\color{flexgreen}\ttfamily]{/*}{*/},
  morecomment=[s][\color{flexgray}\ttfamily]{/**}{*/},
  morestring=[b]",
  stringstyle=\color{flexred}\textbf,
  commentstyle=\color{flexgreen},
  showstringspaces=false,
  numberstyle=\scriptsize,
  numberblanklines=true,
  showspaces=false,
  breaklines=true,
  showtabs=false,
  emph ={[1]class, package, interface, prototype},
  emphstyle={[1]\color{purple}\textbf},
  emph ={[2]internal, public, protected, private,
    super, this, import, new, extends, implements,
    void, true, false, as
  },
  emphstyle={[2]\color{flexblue}\textbf},
  emph =
  {[3]
    function
  },
  emphstyle={[3]\color{flexfunction}\textbf},
  emph =
  {[4]
    var
  },
  emphstyle={[4]\color{flexvar}\textbf}
}

    \lstdefinelanguage{ASN1}{%
      morekeywords=[1]{SEQUENCE,INTEGER,OPTIONAL,NumericString,IA5String,FALSE,ENUMERATED},
      morekeywords=[3]{},%
      morekeywords=[4]{},%
      morestring=[b]",%
     alsoletter={},
     alsoother={},
     sensitive%
    }[keywords,comments,strings]

\lstdefinelanguage{XHTML5}[]{HTML}{%
      morekeywords=[1]{mark,header,hgroup,nav,section,article},
      deletekeywords=[1]{width,height},
      morekeywords=[2]{width,height,svg:svg,svg:circle,cx,r,stroke,fill,xmlns:svg},
      alsoletter={:},
	morestring=[b]',
    }[keywords,comments,strings]

\lstdefinelanguage{pprolog}{%
  morestring=[b]",%
  literate=*{:-}{$\Leftarrow\;\,\,$}{2} {,\ }{$\ \land\ $}{2}
    {not}{$\neg$}{3} {!=}{$\neq$}{2}, 
}[keywords,strings]

\lstdefinelanguage{OPAL}[]{Prolog}{%
  morekeywords=[1]{TEMPLATE,INSTANTIATE,using},
  morekeywords=[2]{child,descendant,adjacent,following,follows,next},
  morekeywords=[3]{form,link},
  morekeywords=[4]{concept,segment},
  morestring=[b]",%
  alsoother={@},
  sensitive,
  literate=* {:-}{$\Leftarrow\;\,\,$}{2} {not\ }{$\neg$}{4} {and}{$\ \land\ $}{3}  
  {all}{$\forall\:$}{3} {\ or\ }{$\lor$}{4}  {!}{$\neg$}{1} {<<}{\!\!\!\!\!$\prec$\;}{2}
  {!=}{\!\!\!$\neq$\;}{2} {...}{$\ldots$}{3} 
}[keywords,strings]

\lstdefinestyle{nonumbers}{numbers=none}
\lstdefinestyle{smaller}{basicstyle=\normalfont\ttfamily\footnotesize}
\lstdefinestyle{OPAL}{language=OPAL, numbers=none, mathescape=true,breaklines=false,
  emphstyle=[1]{\normalfont\sffamily\scshape},
  emphstyle=[2]{\normalfont\sffamily\itshape}}

\usepackage[nounderscore]{syntax}
\newcommand{\mathlit}[1]{\:\textquoteleft\ensuremath{\color{blue}#1}\textquoteright\:}
\lstdefinestyle{nonumbers}{numbers=none}
\lstdefinestyle{smaller}{basicstyle=\normalfont\ttfamily\footnotesize}
\lstdefinestyle{OPAL}{language=OPAL, numbers=none, mathescape=true,breaklines=false,
  emphstyle=[1]{\normalfont\sffamily\scshape},
  emphstyle=[2]{\normalfont\sffamily\itshape}}
\lstdefinestyle{OXPath}{language=OXPath,numbers=none,mathescape=false,breaklines=false}

\@ifundefined{begintheorem}{}{\newtheorem{theorem}{Theorem}}

\newcommand{\deskip}[1]{\vspace*{-#1\baselineskip}}

\newcommand{\smalldeskip}{\deskip{.4}}

\newenvironment{compactdef}{\smalldeskip\begin{definition}}%
        {\end{definition}\smalldeskip}
        {\end{theorem}\smalldeskip}
        {\end{proposition}\smalldeskip}

\newcommand{\fctn}[1]{\textsf{#1}\xspace}
\newcommand{\VALUES}{\fctn{values}}
\newcommand{\TRUE}{\textbf{true}\xspace}
\newcommand{\htmltag}[1]{\texttt{#1}}

\newcommand{\type}[1]{\textsf{\textit{#1}}}

\newcommand{\TYPE}[1]{\textsf{\textsc{#1}}}

\newcommand{\tool}[1]{\textsf{#1}\xspace}

\newcommand{\systemName}[1]{\textsc{#1}}

\newcommand{\OPAL}{\systemName{opal}\xspace}
\newcommand{\OPALlong}{\textbf{o}ntology based web
    \textbf{p}attern \textbf{a}nalysis with \textbf{l}ogic\xspace}

\newcommand{\PARTOF}{\ensuremath{\textrm{\raisebox{.04em}{\textendash}}\hspace*{-.02em}\diamond}}

\makeatother

\svnid{$Id: opal-vldbj13.tex 4101 2012-09-25 20:09:13Z timfu $}

\usepackage{amsmath}
\usepackage{color}
\usepackage{graphicx}
\usepackage{stmaryrd}
\usepackage{multirow}
\usepackage{amssymb}
\usepackage{xspace}
\usepackage{url}
\usepackage{paralist}
\usepackage[utf8]{inputenc}
\usepackage[T1]{fontenc}

\newcommand{\ALLOWEDNODES}{\textsf{Allowed}\xspace}
\newcommand{\FIELDS}{\textsf{Fields}\xspace}
\newcommand{\MATCHINGLABELS}{\textsf{M}\xspace}
\newcommand{\BLOCKEDLABELS}{\textsf{Block}\xspace}

\newcommand{\sem}[1]{\ensuremath{\left\llbracket\,
      #1\,\right\rrbracket}}

\newcommand{\unary}{\ensuremath{\textsf{Unary}}\xspace}
\newcommand{\FLAB}{\ensuremath{\mathfrak{La}}\xspace}
\newcommand{\FREP}{\ensuremath{\mathfrak{Re}}\xspace}
\newcommand{\TAG}{\fctn{tag}\xspace}
\newcommand{\TEXT}{\fctn{text}\xspace}
\newcommand{\BOX}{\fctn{box}\xspace}
\newcommand{\PROPER}{\fctn{isLabel}\xspace}
\newcommand{\VALUE}{\fctn{isValue}\xspace}
\newcommand{\PRECEDENCE}{\ensuremath{\prec}\xspace}
\newcommand{\CONTAINS}{\ensuremath{\lhd}\xspace}
\newcommand{\ALIGNED}{\fctn{aligned}\xspace}
\newcommand{\SUBCLASS}{\ensuremath{\sqsubset}\xspace}

\newcommand{\annotationTypes}{\mathcal{A}}
\newcommand{\childTypes}{\textit{child-}\mathcal{T}}

\graphicspath{{figures/}}

\newcommand{\OPALlang}{\textsc{opal-tl}\xspace}

\usepackage[T1]{fontenc}
\usepackage{graphicx}
\usepackage{mathptmx}
\usepackage[scaled=.8]{helvet}
\usepackage[scaled=.8]{beramono}
\usepackage{balance}
\DeclareMathAlphabet{\mathcal}{OMS}{cmsy}{m}{n}
\begin{document}



\title{The Ontological Key: Automatically Understanding and \\Integrating Forms
  to Access the Deep Web}


\author{Tim Furche \and Georg Gottlob \and Giovanni Grasso \and\\
  Xiaonan Guo \and Giorgio Orsi \and Christian Schallhart}
  
\institute{Department of Computer Science, Oxford University, Wolfson
  Building, Parks Road, Oxford OX1 3QD\\
  \email{firstname.lastname@cs.ox.ac.uk} }

\journalname{}
\date{25 Sep 2012}

\maketitle

\begin{abstract}
  Forms are our gates to the web. They enable us to access the deep content of
  web sites. Automatic form understanding provides applications, ranging from
  crawlers over meta-search engines to service integrators, with a key to this
  content. Yet, it has received little attention other than as component in
  specific applications such as crawlers or meta-search engines.  No
  comprehensive approach to form understanding exists, let alone one that
  produces rich models for semantic services or integration with linked open
  data.
  
  In this paper, we present \OPAL, the first comprehensive approach to form
  understanding and integration. We identify form labeling and form
  interpretation as the two main tasks involved in form understanding. On both
  problems \OPAL pushes the state of the art: For form labeling, it combines
  features from the text, structure, and visual rendering of a web page. In
  extensive experiments on the ICQ and TEL-8 benchmarks and a set of $200$
  modern web forms \OPAL outperforms previous approaches for form labeling by a
  significant margin. For form interpretation, \OPAL uses a schema (or ontology)
  of forms in a given domain. Thanks to this domain schema, it is able to
  produce nearly perfect ($>97\%$ accuracy in the evaluation domains) form
  interpretations. Yet, the effort to produce a domain schema is very low, as we
  provide a Datalog-based template language that eases the specification of such
  schemata and a methodology for deriving a domain schema largely automatically
  from an existing domain ontology.  We demonstrate the value of \OPAL's form
  interpretations through a light-weight form integration system that
  successfully translates and distributes master queries to hundreds of forms
  with no error, yet is implemented with only a handful translation rules.
\end{abstract}

\section{Introduction}
\label{sec:introduction}

Unlocking the vast amount of data in the deep web for automatic processing has
been a central part of ``web as a database'' visions in the past. The web offers
unprecedented choice and variety of products, but we lack tools to make these
wealth of offers easily manageable. Say you are looking for a
flat. \emph{Aren't you tired of filling registration forms with your search
  criteria on the websites of hundreds of local agencies?  You fear to miss the
  site with the very best offer? Wouldn't you wish to automatize these tiresome
  tasks?} To unlock this data for automatic processing requires two keys: a key
that allows us through the human-centric, scripted form interfaces of the web
and a key to identify offers among all the other data on the web. In this paper,
we focus on the former: A key to web forms, the gates to the deep web. Since
these gates are designed for human admission, they pose a plethora of challenges
for automatic processing: Even web forms within a single domain denote search
criteria differently, e.g., ``address'', ``city'', ``town'', and
``neighborhood'' all refer to locations, while other terms denote different
criteria ambiguously, e.g., ``tenure'' might refer to the choice either between
``freehold'' vs.~``leasehold'' or between ``buy'' vs.~``rent''.  Moreover, web
forms present their criteria in different manners, e.g., for a choice among
several options, a form may contain either a drop-down lists or a set of check
boxes. Automatically understanding these variants to pass through forms is
needed by a broad range of applications:
\begin{inparaenum}[]
\item crawling and surfacing the deep web~\cite{raghavan01:_crawl_hidden_web,%
    madhavan08:_googl_deep_web_crawl,%
    DBLP:journals/sigmod/CafarellaCFHHLMM08},
\item interface and service integration~\cite{zhen04:_under_web_query_inter},
\item matching interfaces across domains~\cite{DBLP:conf/icde/BilkeN05,%
    wu09:_model_and_extrac_deep_web_query_inter},
\item classifying the domain of web databases~\cite{DBLP:conf/www/BarbosaF07}
  for web site classification,
\item sampling the contents of web databases~\cite{DBLP:conf/sigmod/MaitiDZD09,%
    DBLP:journals/jacm/Bar-YossefG08}, 
\item ontology enrichment and knowledge-base construction~\cite{re-knowledge-base-construction},
\item question answering for the deep web~\cite{lehmann12:_deqa}.
\end{inparaenum}
In web engineering, automated form understanding contributes, e.g., to web
accessibility and usability~\cite{kalijuvee01:_effic_web_form_entry_pdas}, web
source integration~\cite{dragut09:_hierar_approac_to_model_web}, automated
testing on form-related web applications.

The form understanding problem has attracted a number of
approaches~\cite{zhen04:_under_web_query_inter,%
  wu09:_model_and_extrac_deep_web_query_inter,%
  dragut09:_hierar_approac_to_model_web,%
  nguyen08:_learn_to_extrac_from_label,%
  DBLP:conf/cikm/KhareA09}, for a recent survey see
\cite{khare10:_under_deep_web_searc_inter,DBLP:series/synthesis/2012Dragut}. These approaches turn observations
on common features of web forms (in general, across domains) into specifically
tailored algorithms and heuristics, but generally suffer from three major
limitations:
\begin{asparaenum}[\bfseries (1)]
\item Most approaches are \emph{domain independent} and thus limited to
  observations that hold for forms across all domains. This limitation is
  acknowledged in
  \cite{zhen04:_under_web_query_inter,nguyen08:_learn_to_extrac_from_label,DBLP:conf/cikm/KhareA09},
  but addressed only through domain specific training data, if at all.
  Our evaluation supports \cite{DBLP:conf/cikm/KhareA09} in that a set of
  generic design rules underlies all domains, but that specific domains
  parameterise or adapt to these rules in ways uncommon to other domains.
\item Most approaches are limited in the \emph{classes of features} they use in
  their heuristics and often based on a single sophisticated heuristics using
  one class of features, e.g., only visual features
  \cite{dragut09:_hierar_approac_to_model_web} or textual and field type
  features in \cite{DBLP:conf/cikm/KhareA09}.
\item Heuristics are translated into monolithic algorithms limiting
  maintainability and adaptability. For
  example,~\cite{wu09:_model_and_extrac_deep_web_query_inter}
  and~\cite{nguyen08:_learn_to_extrac_from_label} encode specific assumptions on
  the spatial distance and alignment of fields and
  labels,~\cite{DBLP:conf/cikm/KhareA09} employs hard-coded token classes for
  certain concepts such as ``min'', ``from'' vs.~``max'', ``to''.
\end{asparaenum}

To overcome these limitations, we present \OPAL (\OPALlong), a domain-aware form
understanding system that combines visual, textual, and structural features with
a thin layer of domain knowledge. The visual, textual, and structural features
are combined in a domain-independent analysis to produce a highly accurate form
labeling. However, for most applications what is actually needed is a form model
consistent with a (reference) schema of the forms in the given domain, where all
the fields are associated with given types. In \OPAL, the domain schema is not
only used to classify the fields and segments of the form model, but also to
improve the form model based on a set of structural constraints that describe
typical fields and their arrangement in forms of the domain, e.g., how price
ranges are presented in forms. To ease the development of these domain
ontologies, \OPAL extends Datalog with templates to enable reuse of common form
patterns in forms, e.g., how ranges (of any type) are presented in forms. With
this approach, \OPAL achieves nearly perfect analysis results ($>97\%$
accuracy).

In contrast to previous approaches, \OPAL produces rich form models, typed to
the given domain schema: The models contain not only types (and individual)
constraints for form fields, but group those fields into semantic segments,
possibly with inter-field constraints. These rich models ease the development of
applications that interact with these forms. To demonstrate this, we have
developed a light-weight form integration system on top of \OPAL that fully
automatically translates queries to the domain schema into queries to the
concrete forms. 

\begin{figure*}[tbp]
  \centering
  \subfloat[Full search
  form (highlighting added)]{\label{fig:cmea-full}\includegraphics[width=0.38\textwidth]{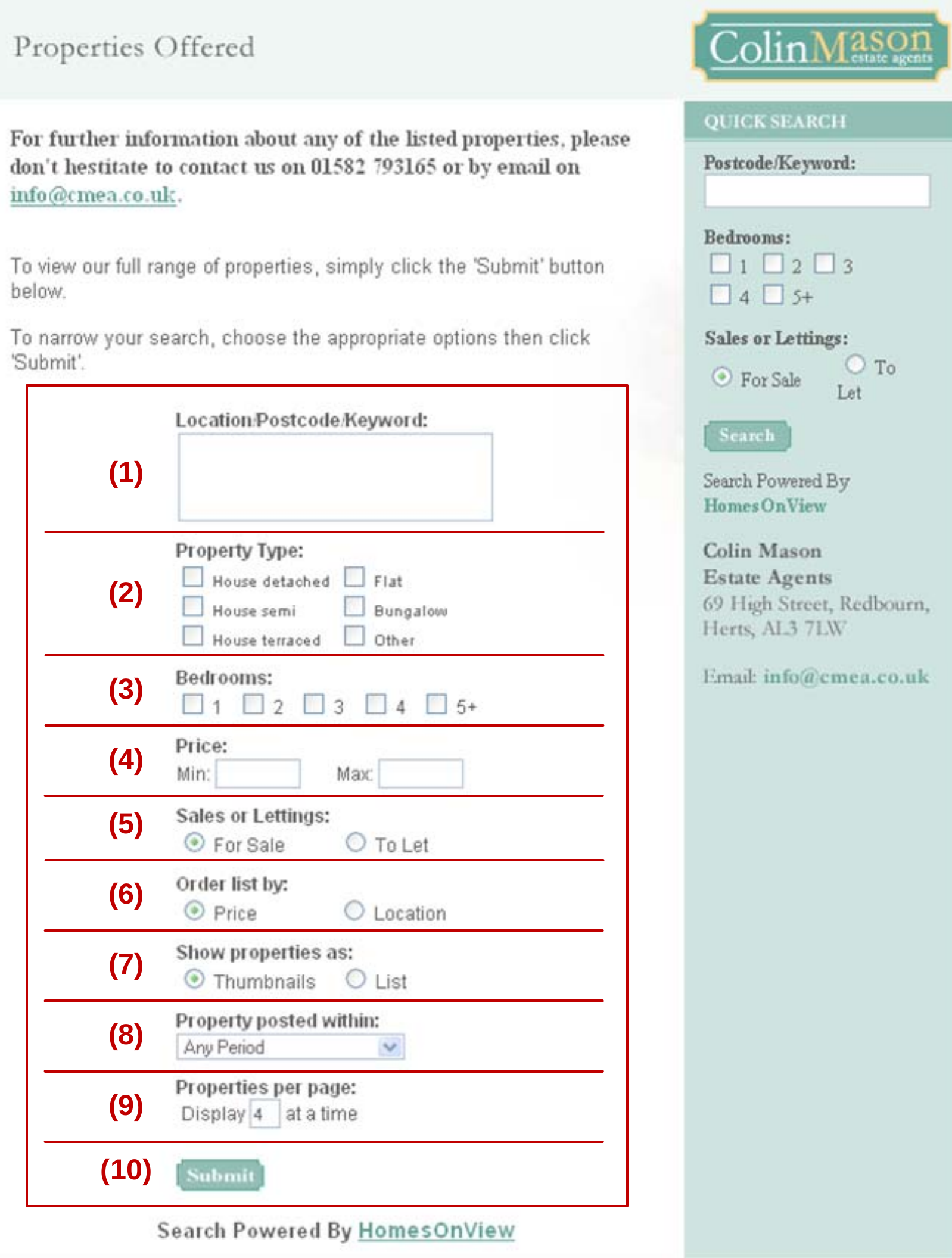}}
  \quad\quad\quad
  \subfloat[Labeling 1: Field
  labels]{\label{fig:cmea-field}\includegraphics[width=0.216\textwidth]{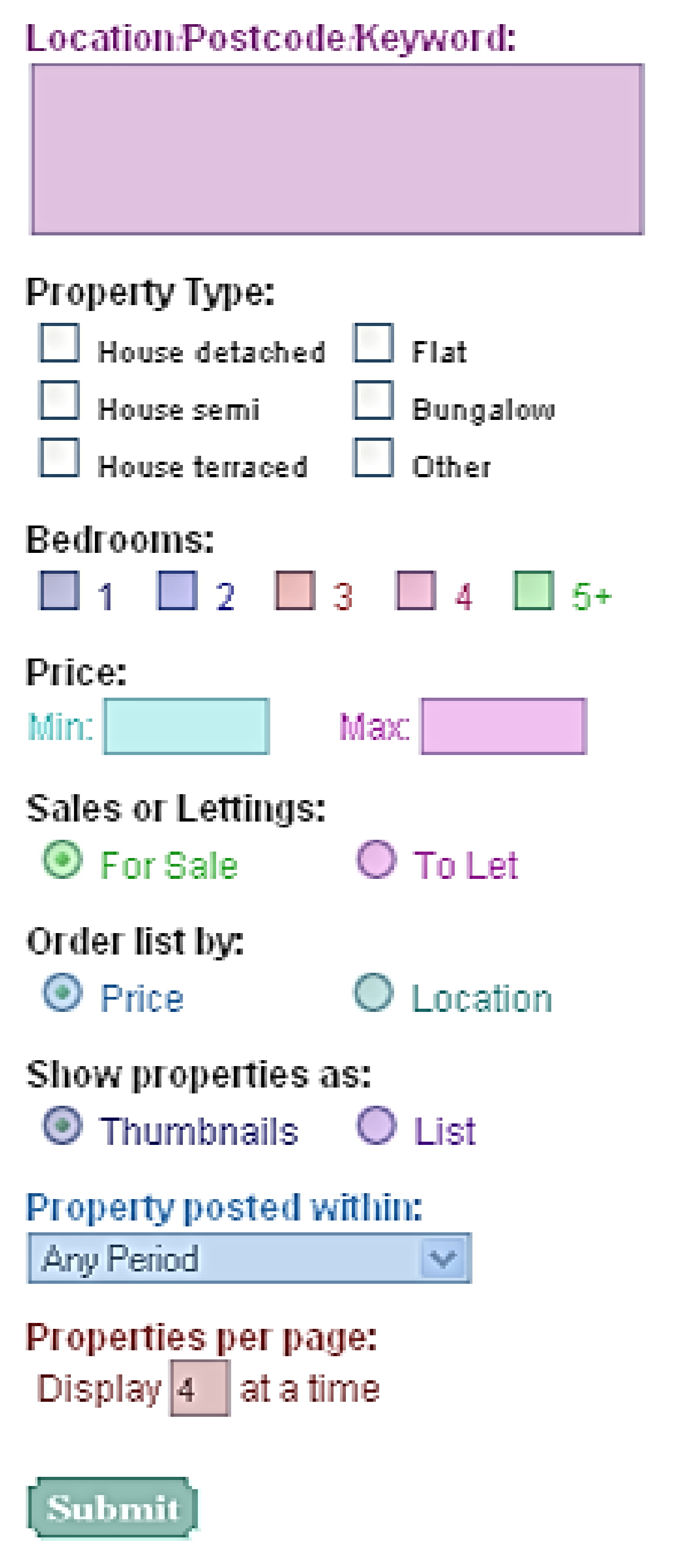}}
  \quad\quad\quad
  \subfloat[Interpretation]{\label{fig:cmea-domain}\includegraphics[width=0.275\textwidth]{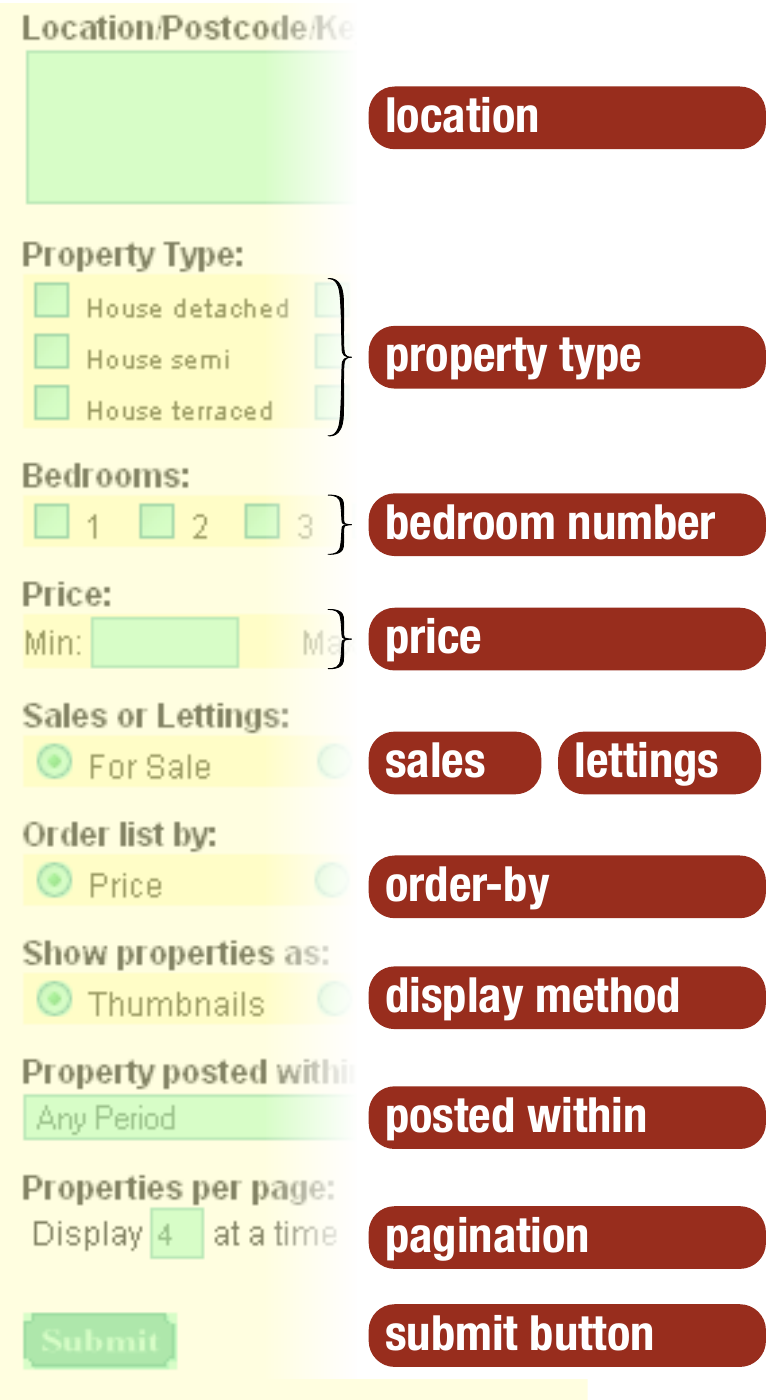}}
  \caption{\emph{Colin Mason} with \OPAL (see Figure~\ref{fig:cmea-segment-result} for
    segment scope)}
  \label{fig:cmea}
\end{figure*}

\subsection{Contributions} \OPAL's main contributions are:
\begin{asparaenum}[\bfseries (1)]
\item \emph{Multi-scope domain-independent analysis}
  (Section~\ref{sec:domain-independent}) that combines structural, textual, and
  visual features to associate labels with fields into a form labeling using
  three sequential ``scopes'' increasing the size of the neighbourhood from a
  subtree to everything visually to the left and top of a field.
  \begin{inparaenum}[\sffamily (i)] 
  \item At \emph{field} scope, we exploit the structure of the page between fields and
    labels;
  \item at \emph{segment} scope, observations on fields in groups of
    similar fields, and
  \item at \emph{layout} scope, the relative position of fields and texts in
    the visual rendering of the page.
  \end{inparaenum}
  We impose a strict preference on these scopes to disambiguate competing
  labelings and to reduce the number of fields considered in later scopes.
\item \emph{Domain awareness.} (Section~\ref{sec:domain-dependent}) \OPAL is
  domain-aware while being as domain-independent as possible
  without sacrificing accuracy. This is based on the observation that generic
  rules contribute significantly to form understanding, but nearly perfect
  accuracy is only achievable through an additional layer of domain knowledge.  To this
  end, we add an optional, domain-dependent classification and form model repair
  stage after the domain-independent analysis. Driven by a domain schema, \OPAL
  classifies form fields based on textual annotations of their labels and values
  assigned in the domain-independent form labeling, as well as the structure of
  that form labeling. This classification is often imperfect due to missing or
  misunderstood labels. \OPAL addresses this in a repair step, where structural
  constraints are used to disambiguate and
  complete the classification and reshape the form segmentation.
\item \emph{Template Language \OPALlang.}
  (Section~\ref{sec:schema-design:-opall}) To specify a domain schema, we
  introduce \OPALlang. It extends Datalog to express common patterns as
  parameterizable templates, e.g., describing a group consisting of a minimum
  and maximum field for some domain type. Together with some convenience
  features for querying the field labeling and its annotations, \OPALlang
  allows for very compact, declarative specification of domain schemata.  We
  also provide a template library of common phenomena, such that the adaption to
  new domains often requires only instantiating these templates with domain
  specific types. \OPALlang preserves the complexity of Datalog.
\item \emph{Methodology for Deriving Domain Schemata.}
  (Section~\ref{sec:domain:methodology-example}) To ease the derivation of an
  \OPAL domain schema, we present a simple, step-by-step methodology how to
  derive such a schema from a standard domain ontology. It is
  based on the observation that often the types of the properties (such
  as price or mileage of a car) in the domain ontology determine the
  configuration of form fields for that type.
\item \emph{Light-weight Form Integration.} (Section~\ref{sec:filling}) To
  demonstrate the value of \OPAL's rich form models, we implement a form
  integration system on top of \OPAL that automatically translates a master
  query to hundreds of concrete forms. As shown in the evaluation, even with rather
  simple translation rules, we achieve accurate form filling.
\item \emph{Extensive Evaluation.} (Section~\ref{sec:evaluation}) In an
  evaluation on over $700$ forms of four different datasets, we show
  that \OPAL achieves highly accurate ($>95\%$) form labelings and,
  with a suitable domain schema, near perfect accuracy in form classification
  ($>97\%$).  To compare with existing approaches (which only perform form
  labeling), we show that \OPAL's domain-independent analysis achieves $94-100\%$
  accuracy on the ICQ benchmark and $92-97\%$ on TEL-8. Thus,
  even without domain knowledge \OPAL outperforms existing approaches by at
  least $5\%$. We also show that the form integration system developed on top of
  \OPAL is able to fill forms correctly in nearly all cases ($>93\%$)
\end{asparaenum}

We believe that \OPAL offers a comprehensive solution to form understanding for
most applications, but also discuss, in Section~\ref{sec:conclusion}, the two
major remaining challenges for \OPAL (and form understanding, in general): highly
scripted, interactive forms, increasingly also using customised 
form widgets, as well as richer integrity constraints and access restrictions, in
particular for applications that aim to extract all of the data behind a form.

This paper is based on \cite{Furche:2012:OAF:2187836.2187948}, but has been
significantly extended in every part, in particular in the following three
aspects: First, \OPALlang is only sketched in
\cite{Furche:2012:OAF:2187836.2187948}. Section~\ref{sec:domain-dependent} is
the first formal definition of \OPALlang, including a full rewriting
semantics. It has also been extended significantly, most importantly in the
supported template features (predicate variables and template groups). Second,
we have added a more detailed description of an \OPAL domain schema and form model to
better illustrate how \OPAL operates and what the output of form understanding
looks like. Finally, we have implemented a full, though light-weight, form
integration and filling system on top of \OPAL (Section~\ref{sec:filling}) to
demonstrate the value of \OPAL's rich models. We have also significantly
extended the evaluation to show the results of the form integration, as well as
to discuss where and why a small portion of forms still pose a challenge to
\OPAL.


\begin{figure*}[tbp]
	\centering 
	\subfloat[Segmentation]{\label{fig:cmea-group}\includegraphics[width=0.37\textwidth]{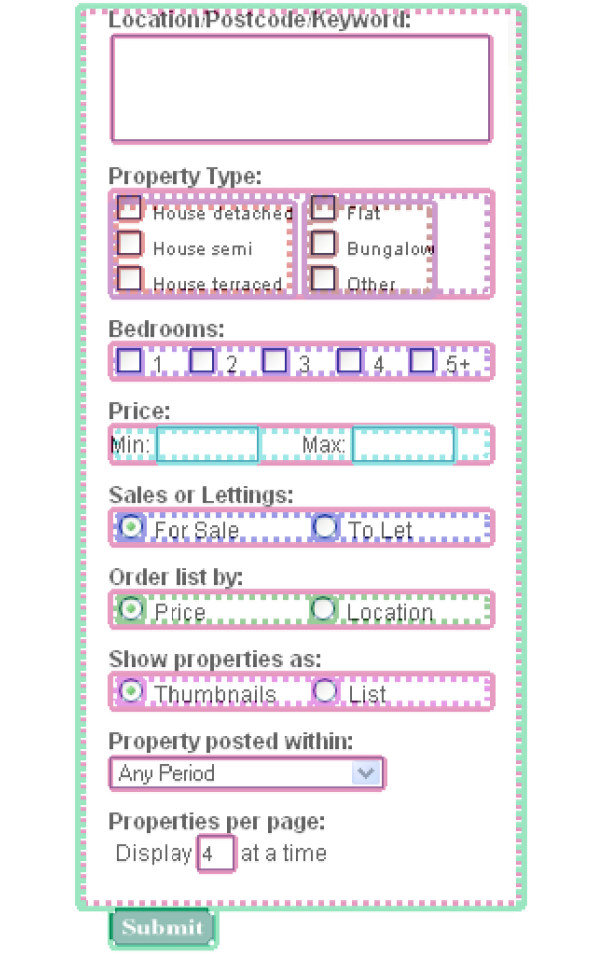}}
	\subfloat[Labeling 2a: Segments]{\label{fig:cmea-grouplabel}\includegraphics[width=0.37\textwidth]{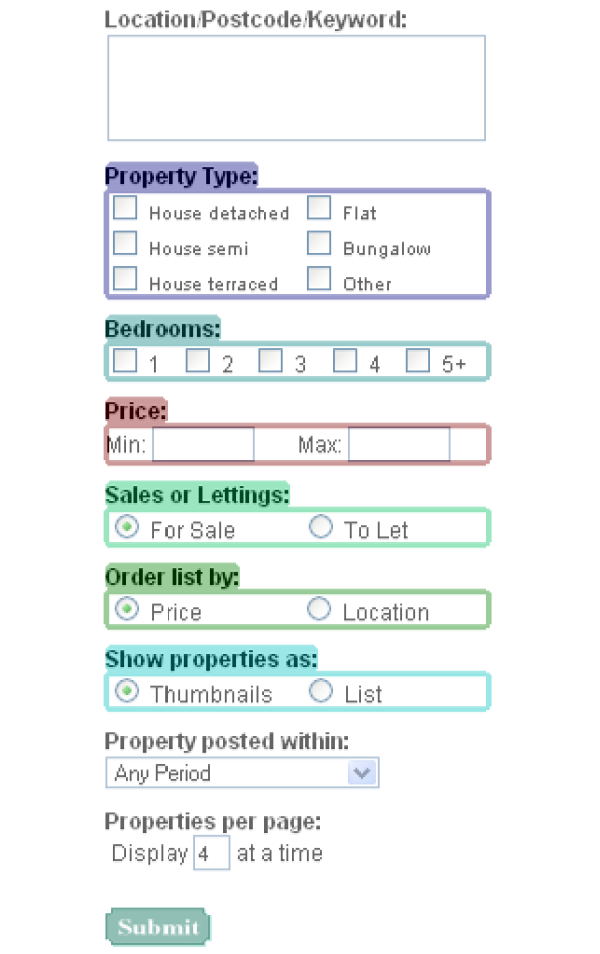}}
	\subfloat[Labeling 2b: Fields by Segment]{\label{fig:cmea-segment}\includegraphics[width=0.254\textwidth]{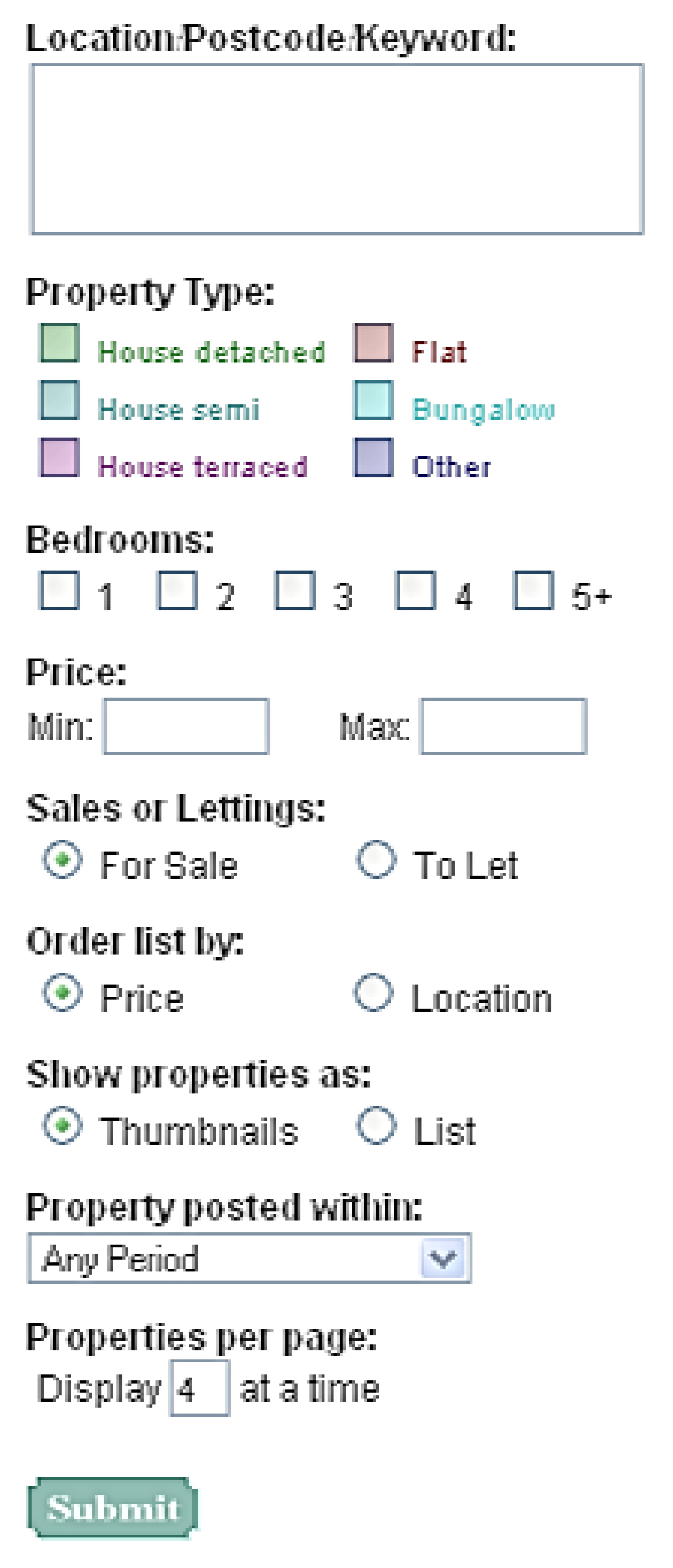}}
	\caption{\emph{Colin Mason} segment scope}
	\label{fig:cmea-segment-result}
\end{figure*}

\subsection{\OPAL: A Walkthrough}
\label{sec:example}


We present the \OPAL approach to form understanding using the form from the UK
real estate agency Colin Mason
(\url{cmea.co.uk/properties.asp}). Figure~\ref{fig:cmea-full} shows the web page
with its simplified CSS box model. The page contains two forms (center and
left): one for detailed search and the other for quick search. \OPAL is able to
identify, separate, label, and classify both forms correctly yielding two
(real-estate) form models. The following discussion focuses on the search form in
the center of Figure~\ref{fig:cmea-full}, in which each of the components
(1)-(10), each of the fields (3)-(7) and the two columns of checkboxes in (2) are
enclosed in a \htmltag{table}, \htmltag{tr}, or \htmltag{td} element. Labels for
each of the components such as ``Bedrooms:'' appear in separate \htmltag{tr}'s.


\OPAL's form understanding operates in two parts: Form labeling and form
interpretation. In the form labeling phase fields and groups of fields (called
segments) are assigned text labels. In the form interpretation phase
those text labels are used to classify the fields and segments on the page,
eventually verifying and repairing the label assignment and producing a form
model in line with the given domain schema. Form labeling itself is split into
field, segment, and layout scope, each assigning successively labels to more 
fields and segments of a form.

\textbf{Field scope.} (Section~\ref{subsec:field-scope}) \OPAL starts by
analysing individual fields assigning labels in two ways: First, we add labels
that explicit reference the field (using the \htmltag{for} attribute).  Second,
we add labels where the common ancestor with a field has no other fields as
descendant. In our example from Figure~\ref{fig:cmea-full}, no explicit
references occur, but the second approach correctly labels all fields except the
checkboxes in (2). In Figure~\ref{fig:cmea-field} we show this initial form
labeling using same color for fields and their labels.

\textbf{Segment scope.} (Section~\ref{subsec:segment-scope}) In segment scope,
we increase the scope of the analysis from form fields to groups of similar
fields (called \emph{segments}). \OPAL constructs these segments from the HTML
structure, but eliminates segments that likely have no semantic relevance and
are only introduced, e.g., for formatting reasons. This elimination is primarily
based on semantic similarity between contained fields approximated via semantic
attributes such as \texttt{class} and visual similarity.  In our example,
components (2)-(7) become segments, with (2) further divided into two segments
for each of the vertical checkbox groups, as shown in
Figure~\ref{fig:cmea-group}. This rough, approximate segmentation may later be
corrected in the form interpretation.

For each \emph{segment as a whole}, \OPAL associates text nodes to create
segment labels. Segment labels can be useful to repair the form model and to
classify fields that have no labels otherwise. In this example, \OPAL assigns
the text in bold face appearing atop each segment as the label, e.g., ``Price:''
becomes the label for (4), see Figure~\ref{fig:cmea-grouplabel}. Furthermore,
\emph{within} each segment, \OPAL identifies repeated groups of interleaving
fields and texts. In the example, each check box in (2) is labeled with the text
appearing after it, as shown in Figure~\ref{fig:cmea-segment}. 

\textbf{Layout scope.} (Section~\ref{sec:layout-scope}) In the layout scope,
\OPAL further enlarges the scope of the analysis to all fields visually to the
left and above a field. The primary challenge in this scope is
``overshadowing'', i.e., if other fields appear in the quadrants to the left
and above a field. In this example the layout scope is not needed.

The result of the layout scope is the form labeling. Notice, that the form
labeling is entirely domain independent.

\textbf{Domain scope.}  If a form model is required, the final step in \OPAL
produces a \emph{form model} that is consistent with a given domain schema.  How
to derive such a domain schema and the necessary annotators is discussed in
Section~\ref{sec:domain:methodology-example}.  It uses domain knowledge to
classify and repair the labeling and segmentation from the form labeling. In the
classification step, \OPAL annotates fields and segments with types based on
annotations of the text labels. The verification step repairs and verifies the
domain model if needed. For both steps, \OPAL uses constraints specified in
\OPALlang. These constraints model typical representations of types in a
domain. E.g., the first field in (4) is classified as \TYPE{min\_price} as we
recognise this segment as an instance of a price range template. These
constraints also disambiguate between multiple, conflicting annotations, e.g.,
fields in (6) are annotated with \type{order\_by} and \type{price}, but the
\type{price} annotation is disregarded due to the group label. Even without the
group label, \type{price} would be disregarded as the domain schema gives
precedence to \type{order\_by} over \type{price} due to the observation that if
they occur together, the field is likely about ``order by price'' and not about
actual prices. Finally, a single repair is performed in this case: We collapse
the two checkbox segments in (2) as they are the only children of their parent
segment and both of the same type. Figure~\ref{fig:cmea-domain} shows the final
field classification as produced by \OPAL.

%
\textbf{Form integration and filling.} Using the form interpretation constructed
in the preceding stages, \OPAL is able to map a master query formulated on the
domain schema into both of the concrete forms on this page (see
Figure~\ref{fig:cmea-full}).  For location, the values are typed in
directly. For price, the range in the master query can also be directly entered,
as the concrete forms use text inputs for prices and \OPAL's form interpretation
identifies the min and max price field successfully. For the bedroom number, the
value from the master query is compared with the members of the check box list
and the most similar is selected.


\section{Approach}
\label{sec:approach}

\OPAL constructs a model of a form consistent with a \emph{domain schema.}
A domain schema describes how forms in a given conceptual
domain, such as the UK real estate domain, are structured.
\OPAL divides this problem (``form understanding'') into \emph{form labeling}
and \emph{form interpretation.}
The form labeling identifies forms and their fields, arranges the fields into a
tree, and labels the found fields, segments, and forms with text nodes from the
page. The form interpretation aligns a form labeling with the given domain
schema and thereby classifies the form fields based on their labels.


\begin{figure*}[tbp]
  \centering
  \includegraphics[width=.8\linewidth]{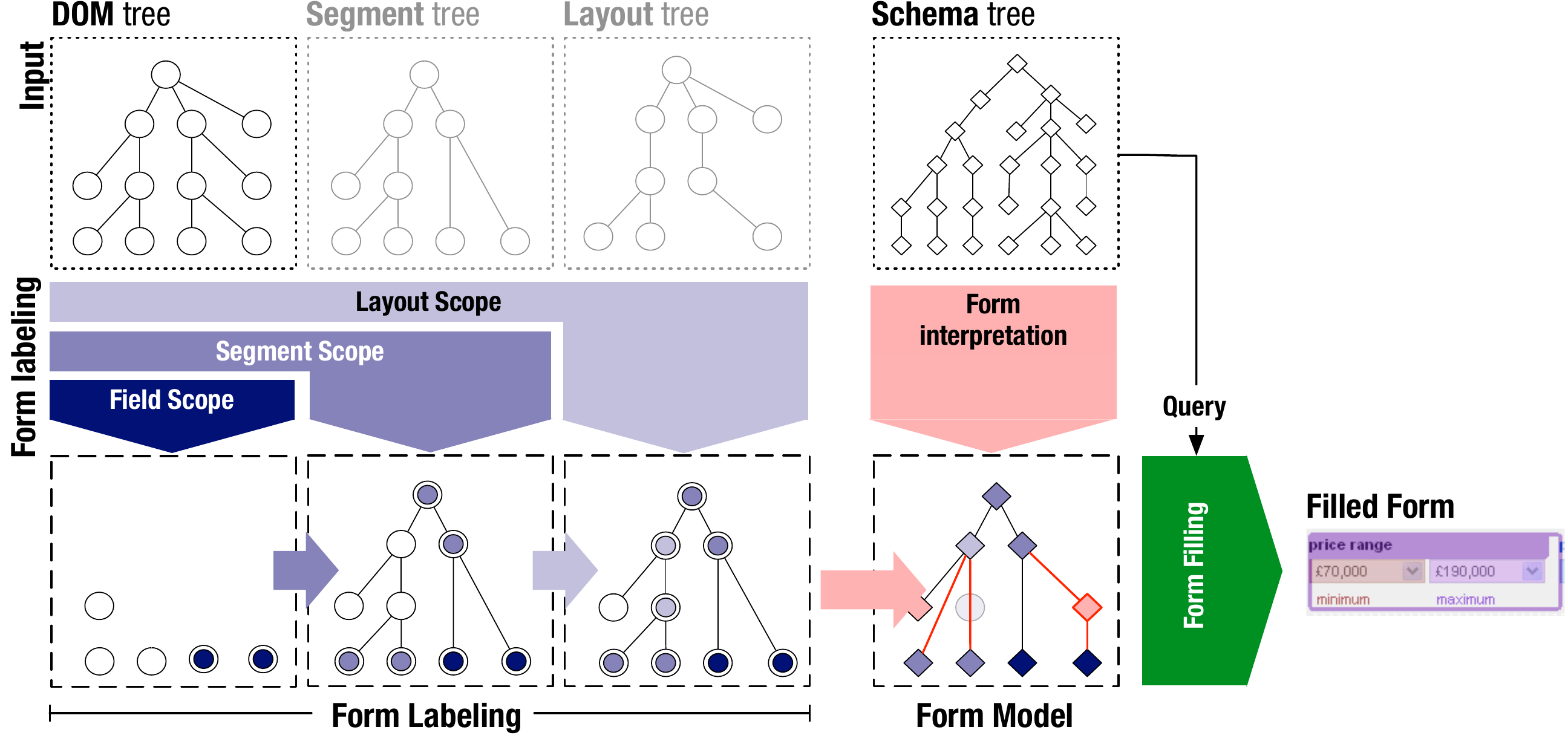} 
  \caption{\OPAL Overview}
  \label{fig:system-overview}
\end{figure*}

\subsection{Problem Definition}
\label{sec:sett-probl-defin}
\paragraph{Form Labeling.}
A \textbf{web page} is a DOM tree $P = \bigl((U)_{U\in\unary},$
$R_{\textsf{child}}, R_{\textsf{next-sibl}},
R_{\textsf{attribute}}\bigr)$ where $(U)_{U \in \unary}$ are unary
type and label relations, $R_{\textsf{child}}$ is the parent-child,
$R_{\textsf{next-sibl}}$ the direct next sibling, and
$R_{\textsf{attribute}}$ the attribute relation.
Further XPath relations (such as \textsf{descendant}) are derived from these
basic relations as usual \cite{Benedikt2007XPath-Leashed}.
$U$ contains relations for types as in XPath (element, text, attribute, etc.)
and three kinds of label relations, namely $\TAG^t$ for tags of elements and
attributes, $\TEXT^l$ for text nodes containing string $l$, and $\BOX^b$ for
elements with bounding box $b$ in a canonical rendering of the page. For
consistency with elements, we represent the value of an attribute as text child
node of the attribute.

\begin{compactdef}
  A \textbf{form labeling} of a web page $P$ is a tree $F$ with functions $\FREP$
  (representative) and $\FLAB$ (label), such that $\FREP$ maps the nodes of $F$
  into $P$. Leafs in $F$ are mapped to form fields and inner nodes to \emph{form
    segments}, that is a DOM element grouping a set of fields. Each node $n$ in
  $F$ is also mapped to a set $\FLAB(n)$ of text nodes, the \emph{labels} of
  $n$.
\end{compactdef}

A node can be labeled with no, one, or many labels via $\FLAB$. The form
labeling contains a representative (via $\FREP$) for each form. A representative
contains all fields (and segments) of that form. This allows \OPAL to distinguish
multiple forms on a single page, even if no form element is present or multiple
forms occur in a single form element.

\begin{compactdef}
  Given a web page $P$, the \textbf{form labeling problem} (or
  schema-less form understanding problem) asks for a form
  labeling $F$ where for each form $f$ in $P$
  \begin{compactenum}[\bfseries(1)]
  \item there is a node $r \in F$ such that $\FREP(r)$ is a suitable
    representative of $f$ and 
  \item for each field $e$ in $f$, there exists a leaf node $n_e \in F$
    such that $n_e$ is a descendant of $r$ and $\FREP(n_e) = e$ where $\FLAB(n_e)$
    is a suitable label set for $e$.
  \item for each inner node $s$ in $F$ (form segment), $\FLAB(s)$ is a suitable
    set of labels for the set of fields contained in $s$.
  \end{compactenum}
\end{compactdef}

The suitability of a form representative $\FREP(r)$ and a label set $\FLAB(n_e)$
is not defined formally, but needs to be evaluated by human annotators (which
this, after all, aims to simulate). Our evaluation
(Section~\ref{sec:evaluation}) shows that \OPAL produces form labelings $F_f$
that match the gold standard in nearly all cases ($>95\%$ without using any
domain knowledge).

We call a form labeling \emph{complete} for a web page, if, for all $e$,
$\FLAB(n_e)$ contains \emph{all} text nodes suitable as labels for $e$. Finding
such a form labeling is correspondingly called the \emph{complete form labeling}
problem.

\paragraph{Form Interpretation.}
To define the form interpretation problem, we formalize the notion of domain
schema and introduce a form model as a form labeling extended with type
information consistent with a given domain schema.  First, we define the part of
a domain schema that provides the necessary knowledge to interpret text nodes
(``annotation schema''): 

\begin{compactdef}
  An \textbf{annotation schema} $\Lambda = \bigl(\mathcal{A},
  \SUBCLASS,$\linebreak $\PRECEDENCE, (\PROPER_a,$ $\VALUE_a: a \in
  \mathcal{A})\bigr)$ defines a set $\mathcal{A}$ of annotation types,
  a transitive, reflexive subclass relation $\SUBCLASS$, a
  transitive, irreflexive, antisymmetric precedence relation
  $\PRECEDENCE$, and two characteristic functions $\PROPER_a$ and
  $\VALUE_a$ on text nodes for each $a \in \mathcal{A}$.
\end{compactdef}

For each annotation type $a\in\mathcal{A}$, we distinguish \emph{proper
labels} and \emph{values}, with $\PROPER_a$ and $\VALUE_a$ as 
corresponding characteristic functions. 
Proper labels are text nodes, such as \emph{``Price:''}, describing
the field type; values, such as \emph{``more
  than \textsterling 500''}, contain possible values of the field. Hence
$\PROPER_{\type{price}}(\mbox{``Price:''})$ and
$\VALUE_{\type{price}}(\mbox{``more than \textsterling 500''})$ hold.

The $\SUBCLASS$ relation holds for subtypes, e.g.,
$\type{postcode}\SUBCLASS\type{location}$, and the $\PRECEDENCE$ relation
defines precedence on annotation types used to disambiguate competing
annotations.
For example, an unlabeled select box with options \emph{``Choose
  sorting order''}, \emph{``By price''}, and \emph{``By postcode''} may be annotated
with \type{order-by}, \type{price}, and 
\type{postcode}. If $\type{order-by} \PRECEDENCE
\type{price}$ and $\type{order-by} \PRECEDENCE \type{postcode}$, we
pick \type{order-by}.

\begin{compactdef}
  A \textbf{domain schema} $\Sigma = (\Lambda, \mathcal{T}, \PARTOF,
  \mathcal{C_T}, \mathcal{C}_\Lambda)$ defines an annotation schema $\Lambda$, a
  set of domain types $\mathcal{T}$ with (transitive, reflexive) part-of
  relation $\PARTOF$, and $\mathcal{C_T}$ and $\mathcal{C}_\Lambda$ that map
  domain types to classification and structural constraints.
\end{compactdef}

For example, $\mathcal{C}_\Lambda(\TYPE{price})$ requires an annotation
\type{price} and prohibit any annotation of a type with precedence over
\type{price} (such as \type{order-by} above).
The set of structural constraints $\mathcal{C_T}(\TYPE{price-range})$ for a
\TYPE{price-range} segment requires a \TYPE{min-price} and 
\TYPE{max-price} field or a \TYPE{price-range} field.
We write $S \models C$, if a constraint set $C$ is satisfied by a set $S$ of
annotation or domain types. The empty constraint set is always
satisfied. $\PARTOF$ plays an important role in the definition of the
constraints, as it prescribes the structure of the types in the domain. For
details on constraints and how to define them, see
Section~\ref{sec:domain-dependent}.

Formally, a \emph{form interpretation} $(F,\tau)$ is a form labeling
$F$ with a partial type-of relation $\tau$, relating nodes in $F$ with
the types $\mathcal{T}$ of $\Sigma$.
Given a node $n$ in $F$, we denote with $\annotationTypes(n) =
\{a \in \mathcal{A}_\Lambda: \exists l\in\FLAB(n)\mbox{ with } 
  \VALUE_a(l) \mbox{ or } \PROPER_a(l)\}$ the set of annotation types
associated with $n$ via its labels, and with $\childTypes(n) =
\bigcup_{(n,n') \in F} \tau(n')$ the set of domain types of the
children of $n$.

\begin{compactdef}
  A form interpretation $(F,\tau)$ is a \textbf{form model} for
  $\Sigma$, iff $\annotationTypes(n) \models \mathcal{C}_\Lambda(t)$ and
  $\childTypes(n) \models \mathcal{C}_{\mathcal{T}}(t)$ for all $n \in F$, $t \in \tau(n)$.
\end{compactdef}
\begin{compactdef}
  Given a domain schema $\Sigma$ and a form labeling $F$, the
  \textbf{form interpretation problem} asks for a form model $(F',
  \tau)$ for $\Sigma$ such that $F'$ differs from $F$ only in inner
  nodes.
\end{compactdef}
Thus, form representatives, fields, and labels are shared between $F$ and
$F'$, but the form segments may be rearranged to conform with the
structural constraints of $\Sigma$.

\begin{compactdef}
  Given a domain schema $\Sigma$ and a web page $P$, the (schema-based) \textbf{form
    understanding} \textbf{problem} asks
  for a form model $(F, \tau)$ of $P$ under $\Sigma$, such that $F$ is a
  solution of the complete form labeling problem for $P$.
\end{compactdef}

\paragraph{Form Integration and Filling.}
In web interface integration a query against a global domain schema is
translated and executed on concrete forms. The returned data is translated into
the domain schema and returned. We focus here on the first part of the
integration problem, the query translation or \emph{form integration} problem,
and more specifically on its optimistic variant:

Let $\Sigma$ be a domain schema. Then a query $Q$ on $\Sigma$ is a set of unary
constraints on $\mathcal{T}$, the domain types in $\Sigma$.  We consider three
types of constraints:
\begin{inparaenum}[\bfseries(1)]
\item \emph{Equality constraints} such as $\TYPE{postcode} = $ OX1;
\item \emph{range constraints} such as $\TYPE{price} \in [700,1250]$;
\item \emph{inclusion constraints} such as $\TYPE{colour} \in \{$red, green,
  black$\}$. 
\end{inparaenum}

\begin{compactdef}
  Given a domain schema $\Sigma$, a query $Q$ on $\Sigma$, and a concrete form
  $F$, the \textbf{form integration problem} is the problem to translate $Q$ into
  a (single) query $Q'$ on $F$ such that $Q'$ returns all results that match $Q$
  and can be retrieved by $F$ and that there is no other query on $F$ with that
  property that returns less results.
\end{compactdef}

Note, that we do not require that $Q'$ returns only results that match $Q$, but
that its result set is minimal among all queries on $F$ that return all matches
for $Q$ that can be retrieved by $F$. This is necessary, as there may be no
query on $F$ that is able to exactly express $Q$.

\subsection{\OPAL Architecture}
\label{sec:system-overview}

\OPAL is divided into three parts. Of those, two form \OPAL's form
understanding: a domain-independent part to address the form labeling problem
and a domain-dependent part for form interpretation according to a domain
schema. The remaining part is devoted to form integration and translates queries
against a domain schema into queries on concrete forms.

\OPAL produces form labelings in a novel multi-scope approach that incrementally
constructs a form labeling combining textual, structural, and visual features
(Figure~\ref{fig:system-overview}). Each of the three labeling scopes considers
features not considered in prior scopes:

\begin{asparaenum}[\bfseries(1)]
\item In \emph{field scope}, we consider only fields and their immediate
  neighbourhood and thus use only the DOM tree as input.

\item In \emph{segment scope}, we detect and arrange form segments into a
  segment tree to interleave the contained text nodes and fields. 

\item In \emph{layout scope}, we broaden the potential labels of a field by
  searching in the layout tree, i.e., the visual rendering of the page, and
  assign text nodes to fields, given a strong visual relation.
\end{asparaenum}

Each scope builds on the partial form labeling of the previous scope and uses the
information from the additional input to find labels for previously unlabeled
fields (or segments). Only the segment scope adds nodes, namely form
segments, whereas field and layout scope only add labels.

Finally, in the \textbf{(4)} \emph{form interpretation}
(Section~\ref{sec:domain-dependent}) we turn the form labeling produced by the
first three scopes into a form model consistent with a given domain schema.
\begin{inparaenum}[\sffamily (i)]
\item The labeling model is extended with (domain-specific) annotations on the
  textual content of proper labels and values.
\item Fields and segments of the form labeling are classified according to
  classification constraints in the domain schema.
\item Finally, violations of structural schema constraints are repaired in a
  top-down fashion.
\end{inparaenum}

Types and constraints of the domain schema are specified using \OPALlang, an
extension of Datalog that combines easy querying of the form labeling and of
annotations with a rich template system. Datalog rules already ease the reuse of
common types and their constraints, but the template extension enables the
formulation of generic templates for such types and constraints that are
instantiated for concrete types of a domain. An example of a type template is
the range template, that describes typical ways for specifying range values
in forms. In the real estate domain it is instantiated, e.g., for price and
various room ranges. In the used car domain, we also find ranges for engine
size, mileage, etc.  Thus, creating a domain schema is in many cases
as easy as importing common types and instantiating templates, see in
Section~\ref{sec:domain-dependent}. 

The form understanding part of \OPAL is complemented with a form integration
part, where we translate a given query on the domain schema into queries on
concrete forms. To do so, we construct an \OPAL form model as above and then use
that form model to map the constraints of the given query to fields on the
concrete form. The form is then filled according to the constraints. Where a
constraint can not be mapped precisely, we use standard similarity techniques to
find the closest, inclusive option (in case of numerical types) or just the
closest option (in case of categorial types), see Section~\ref{sec:filling}. 

\begin{algorithm}[tbp]\small
\SetKwFunction{FieldScopeLabelling}{FieldScopeLabelling}
\SetKwFunction{id}{id}
\SetKwFunction{MarkMultiFieldAncestors}{MarkMultiFieldAncestors}
\SetKw{Break}{break}
\ForEach{field $f$ in $P$}{
  $n \gets f$\;
  \While{$n$ has a parent}{
    \lIf{$n$ is already coloured}{
      colour $n$ \textcolor{red}{red};
      \Break
    }\;
    colour $n$ \textcolor{orange}{orange}\;
    $n \gets$ parent of $n$\;
  }
}
$F \gets $ empty form labeling \;
\ForEach{field $f$ in $P$}{
  $n \gets$ new leaf node in $F$\;
  $\FREP(n) \gets f$\;
  \If{$\exists l \in P$ with \texttt{for} attribute referencing $f$}{
    assign all text node descendants of $l$ as labels to $n$ \;
  }
  $p \gets$ parent of $f$\;
  \While{$p$ not coloured \textcolor{red}{red}}{
    $f \gets p$; $p \gets$ parent of $f$\;
  }
  assign all text node descendants of $f$ as labels to $n$ \;
}
\caption{$\texttt{FieldScopeLabelling}(\textit{DOM } P)$}
\label{algo:label-field-scope}
\end{algorithm}

\section{Form Labeling}
\label{sec:domain-independent}

In \OPAL, form labeling is split into three scopes. Each scope is focused on a
particular class of features (e.g., visual, structural, textual). The form
labeling scopes, \emph{field}, \emph{segment}, and \emph{layout} scope, use
\textbf{domain-independent} labeling techniques to associate form fields or
segments with textual labels, building a form labeling $F$. If a domain schema
is available, the form labeling is extended to a form model in the
domain-dependent analysis (Section~\ref{sec:domain-dependent}).

The form labeling $F$ is constructed bottom-up, applying each scope's technique
in sequence to yet unlabelled fields.  Whenever a field is labelled at a certain
scope level, further scopes do not consider this field again. This precedence
order reflects higher confidence in earlier scopes and addresses
competing label assignments.

\subsection{Field Scope}
\label{subsec:field-scope}

Based on the DOM tree of the input page, the \textbf{field scope} assigns text
nodes in a unique structural relation to individual fields as labels to these
fields (see Algorithm~\ref{algo:label-field-scope}). It relies on the
observation that, if a text node shares a sub-tree of the DOM with a single
field only, then that text node is most likely related to that field. This
simple observation produces a significant portion of form labels, as shown in
Section~\ref{sec:evaluation}, and is designed to produce nearly no false
positives, as also verified in Section~\ref{sec:evaluation}. 

Specifically, Algorithm~\ref{algo:label-field-scope}
\begin{inparaenum}[\bfseries(1)]
\item colours (lines 1--6) all nodes in $P$ that are ancestors of a field
  and do not have other form fields as descendants in
  orange. The least ancestor that violates that condition is
  coloured red.
\item It identifies (line 7--10) all form fields and initialises the form labeling $F$
  with one leaf node for each such field.
\item It considers (lines 11--12) explicit HTML \texttt{label} elements with \emph{direct reference} to a
  form field.
\item It labels (lines 13--16) each field $f$ with all text nodes $t$ whose
  \emph{least common ancestor} with $f$ has no other form field as
  descendant. This includes all text nodes $t$ in the content of $f$ such as its
  values (in case of \texttt{select}, \texttt{input}, or \texttt{textarea}
  elements), since the least common ancestor of $t$ and $f$ is $f$ itself.  We
  find these text nodes in linear time due to the tree colouring.
\end{inparaenum}



\begin{algorithm}[tbp]\small
\SetKwFunction{MarkMultiFieldAncestors}{MarkMultiFieldAncestors}
\SetKwFunction{LookupCounter}{LookupCounter}
\SetKwData{REPRESENTATIVE}{Representative}
$P' \gets P$\;

\While{$\exists n\in P':$ $n$ not a field $\land \bigl(\not\exists$
  field  $d: R_{\textsf{descendant}}(d,n) \in P' \bigr)$}{
  delete $n$ and all incident edges from $P'$\;
}

\While{$\exists n \in P': |\{c\in P': R_{\textsf{child}}(c,n) \in P'\}| = 1$}{
  delete $n$ from $P'$ and move its child to the parent of $n$\;
}

\ForEach{inner node $n$ in $P'$ in bottom-up order}{
  $C \gets \{f: R_{\textsf{child}}(f,n) \in P' \land f \text{ is a field}\}$\;
  $C \gets C \cup \{\REPRESENTATIVE(n'): R_{\textsf{child}}(n',n) \in P'\}$\} \;
  choose $r \in C$ arbitrarily \;
  \If{$\forall r' \in C: r$ style-equivalent to $r'$}{
    $\REPRESENTATIVE(n) \gets r$\;
    delete all non-field children of $n$ and move their children to $n$\;
  }
  \lElse{
    $\REPRESENTATIVE(n) \gets \bot$
  }\;
}

\Return $P'$\;
\caption{$\texttt{SegmentTree}(\textit{DOM } P)$} 
\label{algo:segment-tree}
\end{algorithm}

\subsection{Segment Scope}
\label{subsec:segment-scope}\label{sec:segment-scope}
At \textbf{segment scope}, the labeling analysis expands from
individual fields to form segments, i.e., groups of consecutive fields
with a common parent, forming the segment tree
(Algorithm~\ref{algo:segment-tree}).
These segments are then used to distribute text nodes to unlabeled
fields in that segment (Algorithm~\ref{algo:label-segment-scope}). 
At this scope, we approximate form segments through the DOM structure
and the style of the contained fields. This segmentation is later
adjusted to yield only form segments with a clear semantic. It is
worth noting, that on many forms only very few adjustments are
necessary, supporting the veracity of the approximation of semantic
segments through structure and style.

\paragraph{Segmentation tree.}
We observe that the DOM is often a fair, but noisy approximation of the semantic
form structure, as it reflects the way the form
author grouped fields into segments. Therefore, we start from the DOM structure
to find the form segments, but we eliminate all nodes that can be safely identified
as superfluous: nodes without field descendants, nodes with only one child, and
nodes $n$ where all fields in $n$ are style-equivalent to the fields in the
siblings of $n$. Two fields are \textbf{style-equivalent} if they carry
the same \texttt{class} attribute (indicating a formatting or semantic
class) or the same \texttt{type} attribute and CSS \texttt{style}
information.

If all field descendants of the parent of an inner node $n$ are
style-equivalent, then $n$ should be eliminated from the segment tree, as it
artificially breaks up the sequence of style-equivalent fields and is thus
\emph{equivalence breaking}.

\begin{compactdef}
  The \textbf{segment tree} $P'$ of a form page $P$ is the maximal DOM tree
  included in $P$ (i.e., obtained by collapsing nodes) such that the leaves of
  $P'$ are all fields and, for all its inner nodes $n$,
  \begin{compactenum}[\bfseries(1)]
  \item $\bigl|\{c \in P': R_{\textsf{child}}(c,n)\}\bigr| > 1$,
  \item $n$ is not equivalence breaking. 
  \end{compactenum}
\end{compactdef}

\begin{figure}[tbp]
  \centering 
  \includegraphics[width=1\linewidth]{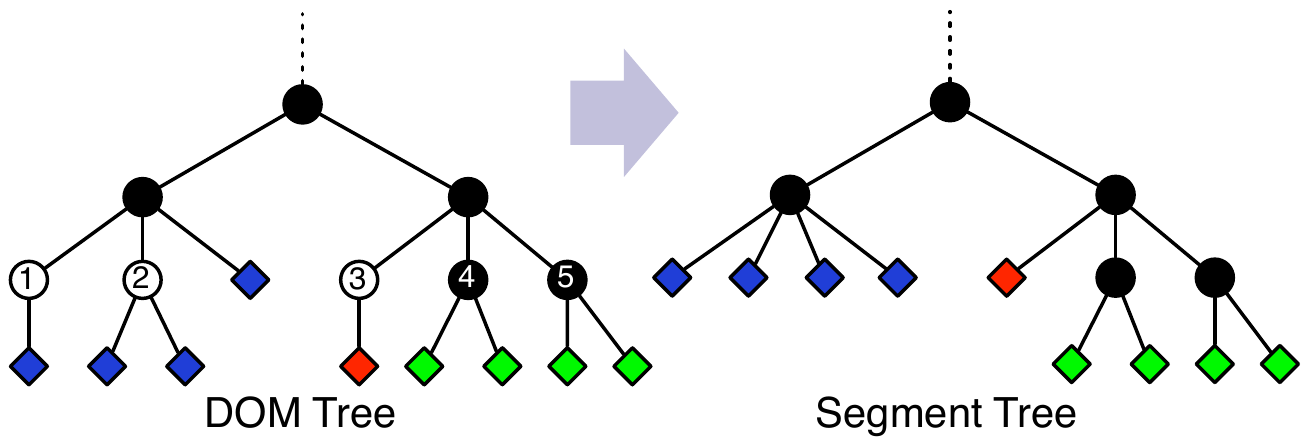}
  \caption{Example DOM and Segment Tree}
  \label{fig:seg-scope-example}
\end{figure}

As an example, consider the DOM tree on the left of
Figure~\ref{fig:seg-scope-example}, where diamonds represent fields and
style-equivalent fields carry the same colour. On the right hand side, we show
\OPAL's segment tree for that DOM. Nodes $1$ and $3$ from the original DOM are
eliminated as they have only one child, and node $2$ as it is equivalence
breaking. Nodes $4$ and $5$ are retained due to the red field.

\begin{theorem}
  The \emph{segment tree} $P'$ of a web page $P$ can be computed in $O(n\times
  d)$ with $n$ size and $d$ depth of $P$. 
\end{theorem}

\begin{proof}
  Algorithm~\ref{algo:segment-tree} computes the segment tree $P'$ for any DOM
  tree $P$. Its leafs are fields (as any non field leafs are eliminated in line
  2--3) and any inner node must have more than $1$ child (due to line 4--5), a
  field descendant (due to line 2--3), and not be equivalence breaking (due to
  lines 6--13). In lines 6--13, we compute a \fctn{Representative},
  bearing the style prevalent among the inner node's fields, for each
  inner node in a bottom-up fashion: If all field children (line 7) and the
  representatives of all inner children (line 8) are style-equivalent (line
  9--10), we choose an arbitrary representative and collapse all inner children
  of that node. Note, that it suffices to compare any of the representatives
  with the fields in $C$ as style-equivalence is transitive.  Otherwise, we
  assign $\bot$ as representative, which is style-equivalent neither
  to any node nor
  to itself. Thus it prevents this node (and its ancestors) from ever being
  collapsed.  By construction, these nodes respect (1) and (2) and this property
  is retained in all later steps, as their subtrees are never touched again. 

  $P'$ is maximal: Any tree $P''$ that includes $P'$ but is
  included in $P$ must contain at least one node from $P$ that has been deleted
  by one of the above conditions. Such a node, however, violates at least one of
  the conditions for a segment tree and thus $P''$ is not a segment tree. This
  holds because the order of the node deletions does not affect the nodes
  deleted. 

  Algorithm~\ref{algo:segment-tree} runs in $O(n\times d)$: Lines 2--3
  are in $O(n)$. Lines 4--5 and lines 6--13 are both in $O(n \times d)$ as they
  are dominated by the collapsing of the nodes. At most, we collapse
  $d-2$ inner nodes and move $O(n)$ leaves $d-2$ times.
\end{proof}

\begin{algorithm}[tbp]\small
\SetKwFunction{SegmentTree}{SegmentTree}
\SetKwFunction{FilterLabelledElements}{FilterLabelledElements}
\SetKwFunction{MergeConsecutiveTextNodes}{MergeConsecutiveTextNodes}
\SetKwFunction{OneToOneAlignmentLabeling}{OneToOneAlignmentLabeling}
\SetKwData{textNodeGroup}{textGrp}
\SetKwData{Fields}{Nodes}
\SetKwData{Labels}{Labels}
\SetKw{new}{new}
\SetKwFunction{add}{add}
\SetKwFunction{size}{size}

$S \gets \SegmentTree(P)$ \;
\ForEach{inner node $s$ in $S$ in bottom-up order}{
  create a new segment $n_s$ in $F$\;
  $\FREP(n_s) \gets s$\;
  create an edge $(n_s,c_s)$ in $F$ for every $\FREP(c_s)$ child of $s$\;
}
\ForEach{segment $n$ in $F$}{
  $\Fields, \Labels \gets \new\ List()$\;
  $\textNodeGroup \gets \emptyset$ \;
  \ForEach{$c: R_{\textsf{descendant}}(c,\FREP(n)) \in P$ in document order}{
    \If{$\exists f \in F: \FREP(f) = c \land \FLAB(f)=\emptyset$}{
      \lIf{$\textNodeGroup \neq \emptyset$}{
        $\Labels.\add(\textNodeGroup)$; $\textNodeGroup \gets \emptyset$\;
      }
      $\Fields.\add(c)$\;
      skip all descendants of $c$ in the iteration \;
    }
    \ElseIf{$c$ is a text node $\land \not\exists d \in F: c \in \FLAB(d)$ }{
      $\textNodeGroup \gets \textNodeGroup \cup \{c\}$\;
    }
  }
  \lIf{$\textNodeGroup \neq \emptyset$}{
    $\Labels.\add(\textNodeGroup)$; $\textNodeGroup \gets \emptyset$\;
  }
  \If{$\Labels.\size() = \Fields.\size() + 1$}{
    add $\Labels[0]$ to $\FLAB(n)$\;
    delete $\Labels[0]$ from $\Labels$\;
  }
  \If{$\Labels.\size() = \Fields.\size()$}{
    \lForEach{$i$}{add $\Labels[i]$ to $\FLAB(\Fields[i])$}\;
  }
}
\caption{$\texttt{SegmentLabeling}(\textit{DOM } P, \textit{Form Labeling }
  F)$}
\label{algo:label-segment-scope}
\end{algorithm}


\paragraph{Segment Labeling.}
We extend the existing form labeling $F$ of the field scope with form
segments according to the structure of the segment tree and distribute labels in
regular groups, see Algorithm~\ref{algo:label-segment-scope}.  First (lines
2--5), we create a form segment node $s$ in the form labeling for each inner
node $n_s$ in the segment tree and choose $n_s$ as representative for $s$ ($\FREP(s)
= n_s$). For each segment with regular interleaving of text nodes with
field or
segment nodes, we use those text nodes as labels for these nodes, preserving any
already assigned labels and fields (from field scope). In detail, we iterate
over all descendants $c$ of each segment in document order, skipping any nodes
that are descendants of another segment or field itself contained in $n$ (line
13). In the iteration, we collect all field or segment nodes in \textsf{Nodes},
and all sets of text nodes between field or segment nodes in \textsf{Labels},
except those already assigned in field scope (line 14), as
we assume that these are outliers in the regular structure of the
segment. We assign the $i$-th text node group to the $i$-th field, if
the two lists have the same size (possibly using the first group
as labels of the segment, line 17--19).

\begin{figure}[tbp]
  \centering 
  \includegraphics[width=\columnwidth]{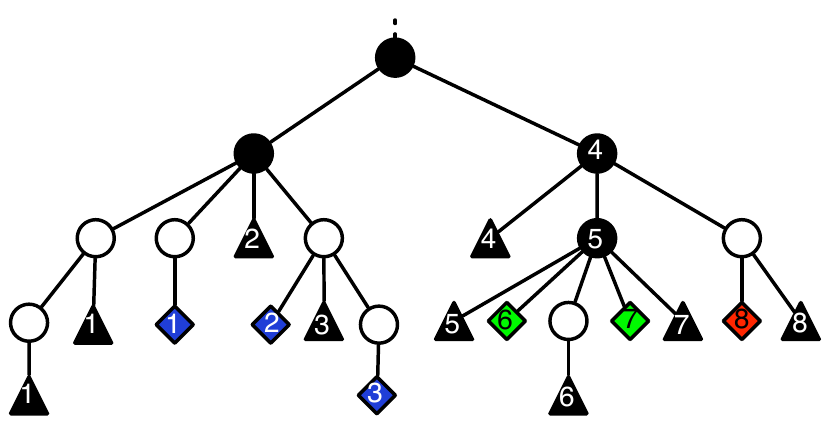}
  \caption{Example for Segment Scope Labeling}
  \label{fig:seg-scope-lab-cases}
\end{figure}

Figure~\ref{fig:seg-scope-lab-cases} illustrates the segment scope labeling with
triangles standing for text nodes, diamonds for fields, black circles for segments, and
white circles for DOM nodes not in the segment tree. The numbers indicate which text
nodes are assigned as labels to which segments or fields. E.g., for the left
hand segment, we observe a regular structure of $($text node$+$, field$)+$ and
thus we assign the $i$-th group of text nodes to the $i$-th field. For the right
hand segment ($4$), we find a subsegment ($5$) and field $8$ that is already
labeled with text node $8$ in the field scope. Thus $8$ is ignored and only
one text node remains directly in $4$, which becomes the segment label. In $5$,
we find one more text node group than fields and thus consider the first text
node group as a segment label. The remaining nodes have a regular structure
$($field, text node$+$$)+$ and get assigned accordingly.

\subsection{Layout Scope}
\label{sec:layout-scope}

At \textbf{layout scope}, we further refine the form labeling for each form
field not yet labelled in field or segment scope, by exploring the visible text
nodes in the west, north-west, or north quadrant, if they are not overshadowed
by any other field. To avoid false positives, we limit this search to the
boundaries of the enclosing form. First, \OPAL constructs a
layout tree from the CSS $\BOX$ labels of the DOM nodes:

\begin{compactdef}
  The \textbf{layout tree} of a given DOM $P$ is a tuple $(N_P, \CONTAINS,
  \fctn{w},\fctn{nw}, \fctn{n}, \fctn{ne}, \fctn{e}, \fctn{se}, \fctn{s},
  \fctn{sw}, \ALIGNED)$
  where $N_P$ is the set of DOM nodes from $P$, $\CONTAINS, \fctn{w},
  \fctn{nw}, \fctn{n}, \ldots$ the ``belongs to'' (containment), west, north-west,
  north, \ldots relations from RCR \cite{rcr-2006}, and $\ALIGNED(x,y)$ holds if
  $x$ and $y$ have the same height and are horizontally aligned. 
\end{compactdef}

We call $\fctn{w}, \fctn{nw}, \ldots$ the neighbour relations. The layout tree is at most
quadratic in size of a given DOM $P$ and can be computed in $O(|P|^2)$. 
For convenience, we write, e.g., $\fctn{w-nw-n}$ to denote the union of the
relations \fctn{w}, \fctn{nw}, and \fctn{n}.

In cultures with left-to-right reading direction, we observe a strong preference
for placing labels in the \fctn{w-nw-n} region from a field. However, forms
often have many fields interspersed with field labels and segment labels. Thus
we have to carefully consider overshadowing. Intuitively, for a field $f$, a
visible text node $t$ is overshadowed by another field $f'$ if $t$ is above $f'$
or also visible from, but closer to $f'$. In the particular case of aligned
fields, the former would prevent any labeling for these fields and thus we relax
the condition. 

\begin{compactdef}
  For a given text node $t$, a field $f'$ \textbf{overshadows} another
  field $f$ if
  \begin{compactenum}[\bfseries(1)]
  \item $f$ and $f'$ are unaligned, $\fctn{w-nw-n}(f',f)$, and
    $\fctn{w-nw-n-ne-e}(t,f')$ or 
  \item $f$ and $f'$ are aligned and 
    \begin{inparaenum}[\sffamily(i)]
    \item $\fctn{w}(t,f')$ or 
    \item $\fctn{nw-n}(t,f')$ and there is a text node $t'$ not overshadowed by
      another field with $\fctn{ne-e}(t',f')$ and $\fctn{w-nw-n}(t',
      f)$.
      %
    \end{inparaenum}
  \end{compactenum}
\end{compactdef}




To illustrate this overshadowing, consider the example in
Figure~\ref{fig:vis-overlap}. For field $F_1$, $T_2$ and $T_4$ are overshadowed
by $F_2$ and $T_3$ by $F_3$, only $T_1$ is not overshadowed, as there is no
other text node that is west, north-west, or north from $F_1$ and not overshadowed by
another field.

\begin{figure}[tbp]
  \centering 
  \includegraphics[width=.6\columnwidth]{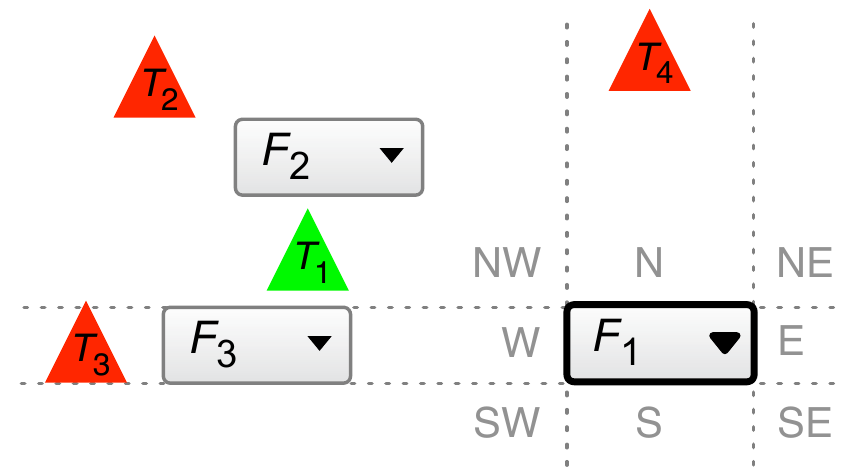}
  \caption{Layout Scope Labeling}
  \label{fig:vis-overlap}
\end{figure}

The layout scope labeling is then produced as follows: For each field $f$, we
collect all text nodes $t$ with $\fctn{w-nw-n}(t,f)$ and add them as labels to
$f$ if they are not overshadowed by another field and not contained in a segment
that is no ancestor of $f$. The latter prevents assignment of labels from
unrelated form segments.

\section{Form Interpretation}
\label{sec:domain-dependent}

\lstset{style=OPAL}
\lstMakeShortInline[style=OPAL]\#
\lstset{emph=[1]{radius,bedroom,room,price,bathroom,reception,buy_rent,buy,rent,geography,area_branch,location,button,link_button,submit,order_by,display_method,property_type,currency,mileage,engine_size,pagination,make,model,door,colour,transmission,used_new,road_tax_bracket,fuel_type,car_type,seat,emissions,seller_type,address,postcode,engine},
  emph=[2]{range_connector,min,max,submit_button}}

There is no straightforward relationship between form fields for domain
concepts, such as location or price, and their structure within a form. Even
seemingly domain-independent concepts, such as price, often exhibit domain
specific peculiarities, such as ``guide price'', ``current offers in excess'',
or payment periods in real estate. \OPAL's domain schemata
allow us to cover these specifics. We recall from Section~\ref{sec:approach}
that a form model $(F',\tau)$ for a schema $\Sigma$ is derived from a form
labeling $F$ by extending $F$ with types and restructuring its inner nodes to
fit the structural constraints of $\Sigma$. 

\OPAL performs form interpretation of a form labeling $F$ in two steps: 
\begin{inparaenum}[\bfseries(1)]
\item the \emph{classification} of nodes in $F$ according to the
  domain types $\mathcal{T}$ to obtain a (partial) typing $\tau_P$. This step
  relies on the annotation schema $\Lambda$ and its typing of 
  labels in $F$;  
\item the \emph{model repair} where the segmentation structure derived in
  the segmentation scope (Section~\ref{subsec:segment-scope}) is aligned with
  the structure constraints of $\Sigma$ to complete the typing. 
\end{inparaenum}

The effort for creating an \OPAL domain schema may, at the first glance, appear
considerable. However, not only do we provide \OPALlang (Section~\ref{sec:schema-design:-opall}) to ease the
specification of a domain schema, we also discuss in
Section~\ref{sec:domain:methodology-example} how all the artefacts needed by
\OPAL for a new domain can be nearly automatically derived from a standard
ontology of a domain (including concept labels) and a set of entity recognisers
(or annotators) for instances of the concepts. We illustrate this methodology
for \emph{domain instantiation} along the example of the used car domain.

\subsection{Schema Design: OPAL-TL}
\label{sec:schema-design:-opall}
\OPAL provides a \textbf{t}emplate \textbf{l}anguage, \OPALlang, for easily
specifying domain schemata reusing common concepts and their constraints
as well as concept templates. To implement a new domain, we only need to
provide 
\begin{inparaenum}[\bfseries(1)]
\item for each annoation type $a$ an annotator implementing
  $\PROPER_a$ and $\VALUE_a$ and
\item an \OPALlang specification of the domain types with their
  classification and structural constraints. The latter can be derived
  almost mechanically from the domain types as discussed in
  Section~\ref{sec:domain:methodology-example}.
\end{inparaenum}

\OPALlang extends Datalog with template capabilities and predefined predicates for
convenient querying of annotations and DOM nodes. An \OPALlang program is
executed against a form labeling $F$ and a DOM $P$. Relations from $F$ and $P$
are mapped in the obvious way to \OPALlang. We only use
\texttt{child} (\texttt{descendant}, resp.) for the child (descendant, resp.)
relation in $F$. We extend document and sibling order from $P$ to $F$:
$\texttt{follows}(X,Y)$ for $X, Y \in F$, if $R_{\textsf{following}}(\FREP(X),
\FREP(Y)) \in P$ and no other node in $F$ occurs between $X$ and $Y$ in document
order; $\texttt{adjacent}(X,Y)$, if $R_{\textsf{next-sibling}}(\FREP(X),
\FREP(Y))\in P$ or vice versa. Finally, we abbreviate $\TEXT^{l}(\FREP(X))$ and
$\TAG^t(\FREP(X))$ as $\texttt{"}l\texttt{"}(X)$ and $t(X)$.

\paragraph{Annotation types and their queries.}
Annotations (instances of annotation types) are characterised by an external
specification of the characteristic functions $\PROPER_a$ and $\VALUE_a$ for
each $a \in \mathcal{A}$. In the current version of \OPAL, these functions are
implemented with simple GATE (\url{gate.ac.uk}) gazetteers and transducers, that
are either provided by human domain experts or realised by access to external
annotators and knowledge bases such as DBPedia and Freebase. Together they
provide annotators for common domain types such as price, location, or
date. Additional entity recognisers or annotators can be added easily, as
described in Section~\ref{sec:domain:methodology-example}.

Annotations are used in annotation queries to select fields based on annotations
on their labels and the labels of their segments:
\begin{compactdef}
  For a form labeling $F$ on a DOM $P$ and an annotation schema
  $\Lambda$ with annotation types $\mathcal{A}$, an
  \textbf{{\OPALlang} annotation query} is an expression of the form
  $X@A\{d,p, e,m\}$ where $X$ is a first-order variable, $A \in
  \mathcal{A}$, and $d$, $p$, $e$, and $m$ are annotation
  \emph{modifiers}.  An annotation query $X@A\mu$ with $\mu \subseteq
  \{d,p,e,m\}$ holds for $X \in \sem{A\mu}$ with
\begin{align*}
  \sem{A\mu} & = 
  \{n \in \FIELDS: \MATCHINGLABELS_\mu(A,n) \neq \emptyset\} \setminus \BLOCKEDLABELS_\mu(A) 
\end{align*}
\vspace*{-1.5\baselineskip}
{\small
\begin{align*}
  \FIELDS & = \{n \in P: \exists \text{ leaf } f \in F: n \in \FREP(f)\} \\
  \MATCHINGLABELS_\mu(A,n) & = 
  \begin{cases}
    %
    \ALLOWEDNODES_\mu(n) \cap \displaystyle\bigcup_{A'\SUBCLASS^* A} \PROPER_{A'} & \text{if $p \in \mu$} \\
    \ALLOWEDNODES_\mu(n) \cap \displaystyle\bigcup_{A'\SUBCLASS^* A} (\PROPER_{A'} \cup \VALUE_{A'}) & \text{otherwise} \\
  \end{cases} \\
  \BLOCKEDLABELS_\mu(A) & = 
  \begin{cases}
    \{n: \exists A' \neq A: |\MATCHINGLABELS_\mu(A,n)| < |\MATCHINGLABELS_\mu(A',n)|\}
    & \text{if $m \in \mu$} \\
    \{n: \exists A' \PRECEDENCE A: |\MATCHINGLABELS_\mu(A,n)| < |\MATCHINGLABELS_\mu(A',n)|\}
    & \text{if $e \in \mu$} \\
    \emptyset & \text{otherwise} \\
  \end{cases} \\
  \ALLOWEDNODES_\mu(n) & = 
  \begin{cases}
    \FLAB(n) & \text{if $d
      \in \mu$}\\ 
    \FLAB(n) \cup \FLAB(\text{parent of $n$})  & \text{otherwise}\\
  \end{cases} \\
\end{align*}
}
\end{compactdef}

Intuitively, an annotation query $X@A$ returns all fields labeled with a label
that is annotated with $A$. If the modifier $d$ (direct) is \emph{not} present,
we also consider the (direct) segment parents, otherwise only \emph{direct}
labels are considered. If the modifier $p$ (\emph{proper}) is present, only $\PROPER_A$
is used, otherwise also $\VALUE_A$. If the modifier $e$ (\emph{exclusive}) is present, a
node that fullfils all other conditions is still not returned, if there are more
labels with annotations of a type that has precedence over $A$. If the modifier
$m$ (\emph{maximal}) is present, no other type, regardless of precedence, may
have more labels with annotations at the node. Since $m$ excludes strictly more
nodes than $e$, a query with both $m$ and $e$ returns the same nodes as that
query without $e$.

\begin{figure}[tbp]
  \centering 
  \includegraphics[width=.6\columnwidth]{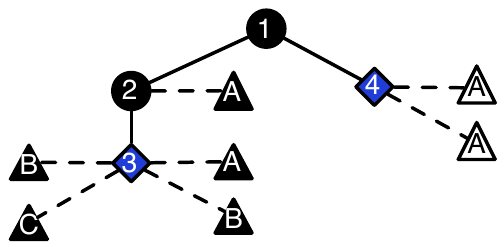}
  \caption{Example Form Labeling}
  \label{fig:form-labeling-ex}
\end{figure}

Consider the form labeling of Figure~\ref{fig:form-labeling-ex} under a
schema with $\type{B} \PRECEDENCE
\type{A}$. Labels are denoted with triangles, fields with diamonds, segments
with circles. Labels are further annotated with matching annotation types (here
always only one), with value labels  drawn as outlines only.  Then, $X@\type{A}\{\}$
matches $3,4$; $X@\type{A}\{e,d\}$ matches $4$, but not $3$ as $3$ has more
labels of $\type{B}$ than of $\type{A}$ and the
exclusive modifier $e$ is present; $X@\type{A}\{e,p\}$ matches $3$, but not
$4$ as the proper modifier $p$ prevents the value labels in white to be
considered. The latter matches $3$ despite the presence of $e$, as we consider
also the labels of the parent of $3$ (since the direct modifier $d$ is absent)
and thus there are two $\type{A}$ labels. 

\begin{figure}[tbp]
  \centering
	\begin{tabular}{cc} 
	  \subfloat[]{\label{fig:annotation-example-d}\includegraphics[width=.75\columnwidth]{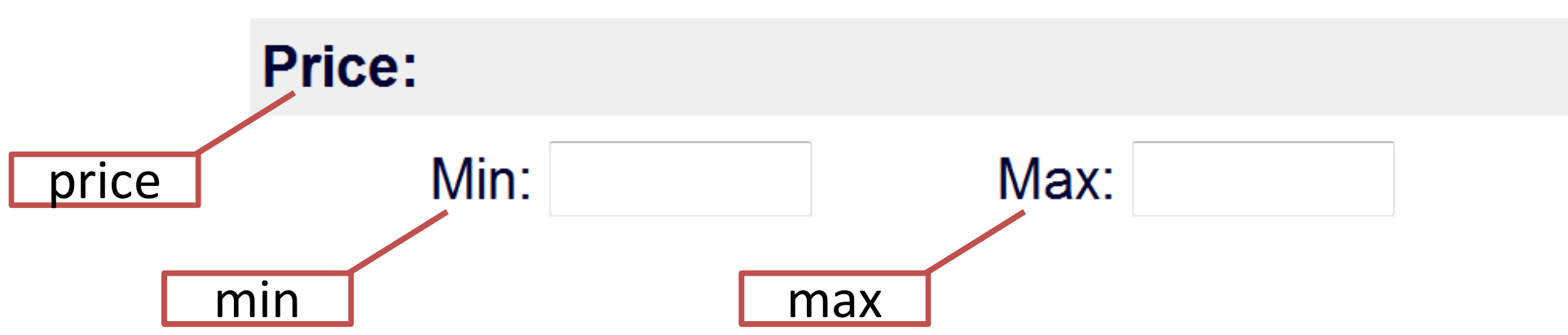}}\\
	  \subfloat[]{\label{fig:annotation-example-e}\includegraphics[width=.75\columnwidth]{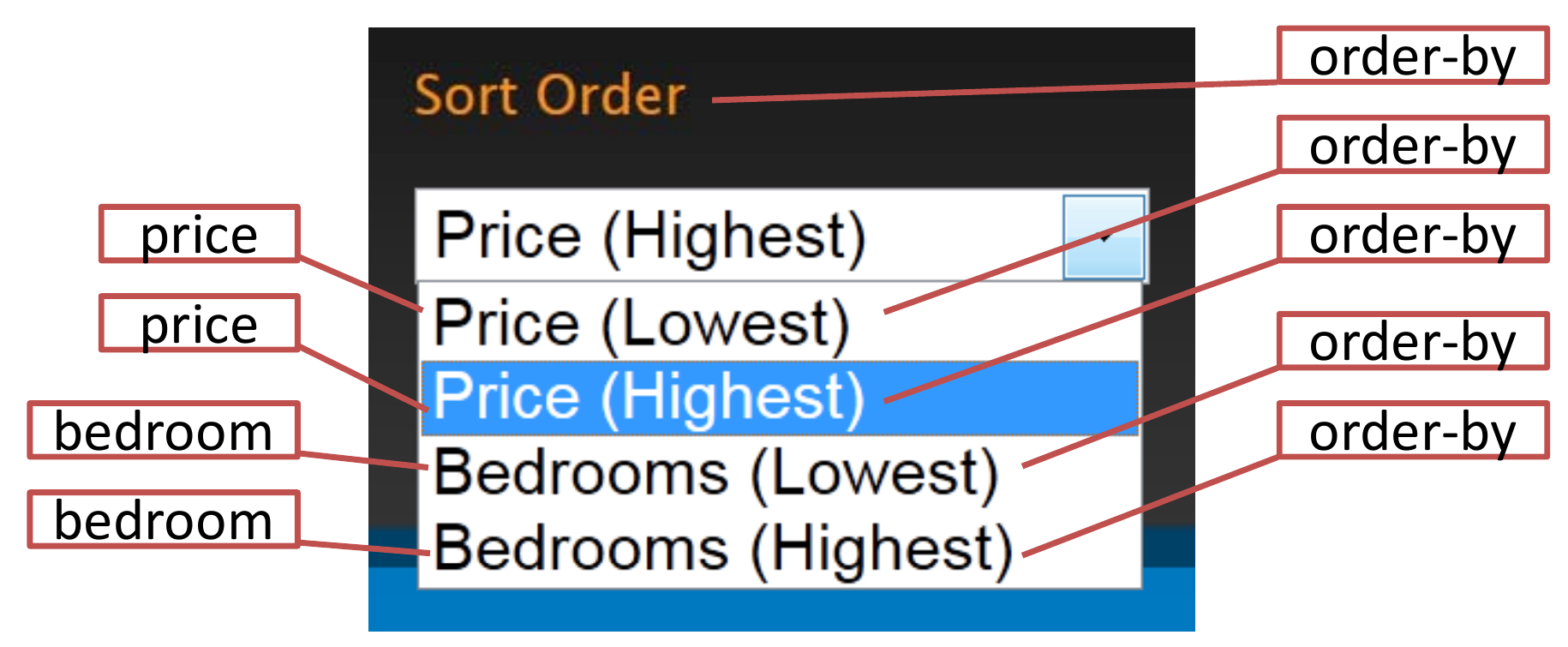}}
	\end{tabular}
  \caption{Label Annotation Examples}
  \label{fig:annotation-example-de}
\end{figure}
 
Figure~\ref{fig:annotation-example-de} shows a real-life example with the
annotations produced by a typical set of
annotators. In~\ref{fig:annotation-example-d}, there are two text inputs for min
and max price. However, the two labels ``min'' and ``max'' are the only directly
associated text boxes and do not carry any information that indicates that these
fields are about prices. This is available only when considering the segment
(and thus indirect) label ``Price:''. Thus, $X@\type{price}\{d\}$ returns the
emptyset, but $X@\type{price}\{\}$ returns the two
fields. In~\ref{fig:annotation-example-e}, the drop-down menu for result
ordering receives two price annotations, two bedroom annotations, and five
order-by annotations. With $\type{order-by} \PRECEDENCE \type{price}$,
$X@\type{price}\{e\}$ returns the emptyset, as the price annotations are
``blocked'' by the order-by annotations.

\paragraph{OPAL-TL templates.}
\OPALlang is a Datalog-based language for the definition of 
reusable templates of domain concepts. Examples of such templates are
basic classification rules deriving a domain type from a conjunction of annotation
types or min-max range templates where we look for multiple fields with related
annotations in a group and some clue that they represent a range. In general,
there are two types of such templates, one for classification constraints, one
for structural constraints. The former specify relationships between domain and
annotation types, the latter the abstract structure of domain concepts.

\begin{compactdef}
  An \textbf{\OPALlang template} is an expression of the form: #TEMPLATE $N$<$T_1, \ldots,
  T_k$> { $p_1$ :- $\mathit{expr}_1$. $\ldots$ }# where $N$ is the template name, $T_1,
  \ldots, T_k$ are template variables, $p_1$ is a template atom, 
  $\mathit{expr}_1$ a boolean formula over template atoms and annotation queries.
  A \emph{template atom} $p$#<#${\textbf{t}}$#>#$({\bf s})$ consists of a first-order
  predicate $p$, a sequence of terms ${\bf }=t_1, \ldots, t_n\}$ (where $t_i$ is
  either a constant or a template variable), and a sequence of terms ${\bf s}=s_1,
  \ldots, s_n$ where each $s_i$ is either a constant or a first-order
  variable. Template and first-order variables constitute two disjoint
  sets. Note that, if ${\bf t}$ is empty, then a template atom is a normal
  first-order atom. Moreover, when all terms ${\bf t}$ are constants, we say
  that the template atom is \emph{template-ground}.
\end{compactdef}
Multiple rules with the same head can be used to express disjunction of their
bodies. For convenience, we use $\lor$ and $\neg$ over conjunctions, which are
translated to Datalog$^\neg$ as usual.

As an example, the following template defines a family of constraints
that associate the concept (domain type) $C$ to a node $N$ whenever $N$ is
labeled by an exclusive direct and proper annotation of type $A$.
\begin{lstlisting}
TEMPLATE basic_concept<C,A>{ concept<C>(N):-N@A{d,e,p} }
\end{lstlisting}

An \emph{instantiation} of a template $\mathit{tpl}$ produces a set of rules
where the template variables $C_1, \ldots, C_k$ are assigned to
values $v_1^{i}, \ldots, v_k^{i}$ defined by a \emph{template instantiation}
expression of the form:
\begin{lstlisting}
INSTANTIATE $\mathit{tpl}$<$T_1,\ldots,T_k$> using {<$v_1^{1},\ldots,v_k^{1}$> $\ldots$ <$v_1^{n},\ldots,v_k^{n}$>}
\end{lstlisting}

For example, the following expression instantiates
#basic_concept# replacing $C$ with type #radius# and $A$ with
annotation type \type{radius}
\begin{lstlisting}
INSTANTIATE basic_concept<C,A> using {<radius, $\type{radius}$>}
\end{lstlisting}
and produces the following instantiated rule:
\begin{lstlisting}
concept<radius>(N):-N@$\type{radius}${d,e,p}
\end{lstlisting}


The full syntax of \OPALlang is given in Figure~\ref{fig:opal:syntax} (with
\synt{string}, \synt{id}, and \synt{var} as in Datalog and \synt{tvar},
\synt{type-id}, \synt{annot-id}, \synt{tag} template variables, domain types,
annotation types, and HTML tags, respectively). 

\setlength{\grammarindent}{5.5em}
\setlength{\grammarparsep}{.25\baselineskip}
\renewcommand{\litleft}{\textquoteleft\color{blue}\ttfamily}
\renewcommand{\litright}{\normalfont\normalcolor\textquoteright}
\begin{figure}[tbp]
  \centering
\begin{grammar}\small

  <program> ::= (<template> | <inst> | <trule> )+

  <template> ::= `TEMPLATE' <id> `<' <tvar>+ `>' `{' <trule>+ `}'

  <inst> ::= `INSTANTIATE' <id> `<' <tvar>+ `>' \\
             \hspace*{6em} `using' `{' (`<' <const>+ `>')+ `}'

  <trule> ::= <tatom> \mathlit{\leftarrow} <tbody> | <inst>

  <tbody> ::= <texpr> (`,' <texpr>)*

  <texpr> ::= <atom> | <annot> | <tatom> | <neg> | <disj>

  <annot> ::= <var>`@' `{' (`d' | `e' | `p' | `m')* `}'

  <tatom> ::= <id> `<' (<tvar> | <const>)+ `>' `(' <par>* `)' 
             \alt `<' <tvar> `>' `(' <par>* `)'
  
  <par>   ::= <var> | <tvar> | <const>

  <const> ::= <type-id> | <annot-id> | <tag> | <string> | <id>

  <neg>   ::= \mathlit{\neg} `(' <tbody> `)'

  <disj>  ::= `(' <tbody> \mathlit{\lor} <tbody> `)'
  
\end{grammar}
  
  \caption{\OPALlang syntax}
  \label{fig:opal:syntax}
\end{figure}



The semantics of \OPALlang is given by rewriting any set of templates $\Sigma_T$
into Datalog$^{\neg}$ programs, using assignments of template variables to
constants specified by the instantiation rules, and by considering every
template-ground predicate name as a new first-order predicate. Due to
possible occurrences of #INSTANTIATE# within templates, the instantiation must be repeated
until there are no applicable #INSTANTIATE# rules. To ensure termination of the
instantiation procedure, we do not allow recursive template
instantiations. Properties such as safety can be easily extended from
Datalog$^\neg$ to \OPALlang:

\begin{compactdef}
  A \OPALlang template is \textbf{safe}, if every template variable that occurs in
  the body also occurs in the head of the template and every rule is safe, i.e.,
  all first-order variables that occur in the head or in a negative atom in the
  body, also occur in a positive atom in the body. 
\end{compactdef}

\begin{proposition}
  Let $\Sigma_T$ be a set of safe \OPALlang templates, and let $\mathcal{S}$ be
  an assignment specified by \OPALlang instantiation rules, then any
  instantiation $\tau(\Sigma_T,\mathcal{S})$ is a safe Datalog$^\neg$ program.
  \label{prop:encoding}
\end{proposition}
 
In contrast to safety, stratification depends also on the instantiation and is
therefore defined over the expanded program as
usual.

A natural question is now the complexity of computing the form model using
\OPALlang. This is related to the complexity of \emph{fact inference} in
\OPALlang.

\begin{proposition}
  Fact inference in \OPALlang is \textsc{PTIME}-complete in data complexity
  (when $\Sigma_T$ and $\mathcal{S}$ are fixed) and \textsc{EXPTIME}-complete in
  combined complexity.
\end{proposition}

\begin{proof}
  Consider a set of template atoms $D$, a set of \OPALlang templates $\Sigma_T$
  over a set of template predicates $\mathcal{R}_T$ of at most arity $k$, and an
  assignment $\mathcal{S}$ of template variables to constants in a set
  $\Gamma_T$ specified by \OPALlang instantiation rules. The \emph{fact
    inference} problem for $D$, $\Sigma_T$, and $\mathcal{S}$ is to decide
  whether $D \cup \langle \Sigma_T, \mathcal{S} \rangle \models \underline{a}$,
  where $\underline{a}$ is a template atom. According to
  Proposition~\ref{prop:encoding}, the problem can be reduced to fact inference
  in Datalog$^{\neg}$, i.e., deciding whether $D \cup \Sigma_D \models
  \underline{a}$ where $\Sigma_D = \tau(\Sigma_T, \mathcal{S})$ is the rewritten
  program. Clearly, the data complexity is \textsc{PTIME}-complete as for
  Datalog$^{\neg}$. Regarding the combined complexity, recall that fact
  inference for a Datalog$^{\neg}$ program $\Sigma_D$ and a set of atoms $D$ is
  \textsc{EXPTIME}-complete since the maximum number of atoms that can be
  inferred is $\mid \mathcal{R}_D \mid \cdot (dom(D))^{w}$ where $\mathcal{R}_D$
  is the set of predicates of $\Sigma_D$, $dom(D)$ is the domain of $D$ and $w$
  is the maximum arity of predicates in $\mathcal{R}_D$. The rewriting
  $\tau(\Sigma_T, \mathcal{S})$ can generate at most $\mid \mathcal{R}_T \mid
  \cdot (\Gamma_T)^{k}$ template-ground atoms that contribute to the signature
  of $\Sigma_D$. Therefore, the number of atoms that can be generated is
  $\mathcal{O}(2^{k} \cdot 2^{w})$ that is still exponential. The claim follows.
\end{proof}

\subsection{Classification}
\label{subsec:classification-constraints}

Classification is based on the classification constraints of the domain
schema. In $\OPAL$ these constraints are specified using \OPALlang to enable
reuse of domain concepts and templates. For instance, in the real estate and used car
domains, we identify four templates that suffice to describe nearly all
classification constraints. These templates effectively capture very common
semantic entities in forms and are parametrized using domain knowledge. The
building blocks are a domain type (or concept) $C$ and an annotation type $A$
that is used to define a classification constraint for $C$. None of these
templates uses more than one annotation type as template parameter, though many
query additional (but fixed) annotation types in their bodies.

\begin{figure}
\begin{lstlisting}[numbers=left]
TEMPLATE concept_by_proper<C,A> {concept<C>(N):-N@A{d,e,p}}

TEMPLATE concept_by_segment<C,A>{concept<C>(N):-N@A{e,p}}

TEMPLATE concept_by_value<C,A>  {concept<C>(N):-N@A{m}, 
       !(A$_1$ $\neq$ A, N@A$_1${d,e,p} or N@A$_1${e,p}) }

TEMPLATE concept_minmax<C,C$_M$,A> {
 concept<C$_M$>(N$_1$):-child(N$_1$,G),child(N$_2$,G),adjacent(N$_1$,N$_2$),
   N$_1$@A{e,d},(concept<C>(N$_2$) or N$_2$@A{e,d})
 concept<C$_M$>(N$_2$):-child(N$_1$,G),child(N$_2$,G),follows(N$_2$,N$_1$),
   concept<C>(N$_1$),N$_2$@range_connector{e,d},!(A$_1$<<$A$,N$_2$@A$_1${d})
 concept<C$_M$>(N$_1$):-child(N$_1$,G),child(N$_2$,G),adjacent(N$_1$,N$_2$),
   N$_1$@A{e,p},N$_2$@A{e,p},$\bigl($(N$_1$@min{e,p},N$_2$@max{e,p}) 
     or (N$_1$@max{e,p},N$_2$@min{e,p})$\bigr)$
\end{lstlisting}
\caption{{\OPALlang} classification templates}
\label{fig:classification-constraints}
\label{tab:class-c}
\end{figure}

Figure~\ref{tab:class-c} shows the classification templates 
for real-estate and used car:
\begin{inparaenum}[\bfseries(1)]
\item \emph{Concept by proper label.} The first template captures direct classification of
  a node $N$ with type $C$, if $N$ matches #X@A{d,e,p}#, i.e., has more proper
  labels of type $A$ than of any other type $A'$ with $A' \PRECEDENCE A$. This
  template is used by far most frequently, primarily for concepts with unambiguous
  proper labels. 
\item \emph{Concept by segment label.} The second template relaxes the requirement by
  considering also indirect labels (i.e., labels of the parent segment). In the
  real estate and used car domains, this template is instantiated primarily for
  control fields such as #order_by# or #display_method# (grid, list, map)
  where the possible values of the field are often misleading (e.g., an
  #order_by# field may contain \emph{``price''}, \emph{``location''}, etc. as
  values).
\item \emph{Concept by value label.} The third template also considers value
  labels, but only if neither the first nor the second template can match. In
  that case, we infer that a field has type $C$, if the majority of its direct
  or indirect, value or proper labels are annotated with $A$.
\item \emph{Min-max concept.} Web forms often show pairs of fields representing
  min-max values for a feature (e.g., the number of bedrooms of a property). We
  specify this template with three simple rules (line 5--12), that describe three
  configurations of segments with fields associated with value labels only
  (proper labels are captured by the first two templates). It is the only
  template with two concept template parameters, $C$ and $C_M$ where $C_M
  \SUBCLASS C$ is the ``minmax'' variant of $C$.  The first locates, adjacent
  pairs of such nodes or a single such node and one that is already classified
  as $C$. The second rule locates nodes where the second follows directly the
  first (already classified with $C$), has a #range_connector# (e.g., ``from''
  or ``to''), and is not annotated with an annotation type with precedence over
  $A$.  The last rule also locates adjacent pairs of such nodes and classifies
  them with $C_M$ if they carry a combination of #min# and #max# annotations.
\end{inparaenum}

In addition to these templates, there is also a small number of specific
rules. In the real estate domain, e.g., we use the following rule to describe
forms that use links (\texttt{a} elements) for submission (rather than submit
buttons). Identifying such a link (without probing and analysis of Javascript
event handlers) is performed based on an annotation type for typical content,
\texttt{title} (i.e., tooltip), or \texttt{alt} attribute of contained
images. This is mostly, but not entirely domain independent (e.g., in real
estate a ``rent'' link).

\begin{lstlisting}
concept<link_button>(N$_1$):-form(F),descendant(N$_1$,F),link(N$_1$),
  N$_1$@link_button{d},!$\bigr($descendant(N$_2$,F),
    (concept<button>(N$_2$) or follows(N$_1$,N$_2$))$\bigr)$ 
\end{lstlisting}

\subsection{Model Repair}
\label{para:completion-constraints}

With fields and segments classified, \OPAL verifies and repairs the structure of
the form according to structural constraints on the segments, such that it fits
to the domain schema.
As for classification constraints, we use \OPALlang to specify the structural
constraints. The actual verification and repair is also implemented in
\OPALlang, but since it is not domain independent, it is not exposed to the user
for modification. Here, we first introduce typical structural constraints and
their templates and then outline the model repair algorithm, but omit the
\OPALlang rules.

\begin{figure}
\begin{lstlisting}[numbers=left]
TEMPLATE segment<C>{
  segment<C>(G):-outlier<C>(G),child(N$_1$,G),!$\bigl($child(N$_2$,G),
    !(C$_1$ $\PARTOF$ C, concept<C$_1$>(N$_2$) or segment<C$_1$>(N$_2$))$\bigr)$ }

TEMPLATE segment_range<C,C$_M$> {
  segment<C>(G):-outlier<C>(G),concept<C$_M$>(N$_1$),
    concept<C$_M$>(N$_2$),N$_1$ != N$_2$,child(N$_1$,G),child(N$_2$,G) }
 
TEMPLATE segment_with_unique<C,U> { 
  segment<C>(G):-outlier<C>(G),child(N$_1$,G),concept<U>(N$_1$,G),
    !$\bigl($C$_1$ $\PARTOF$ C, child(N$_2$,G), N$_1$ != N$_2$,
                      !(concept<C$_1$>(N$_2$)$\lor$segment<C$_1$>(N$_2$))$\bigr)$.}

TEMPLATE outlier<C>{
  outlier<C>(G):-child(G,P),child(G$'$,P),!(segment<C>(G$'$)) }
\end{lstlisting}

\caption{{\OPALlang} structural constraints}
\label{tab:compl-c}
\end{figure}

\begin{figure*}[tbp]
  \centering
  \includegraphics[width=.7\textwidth]{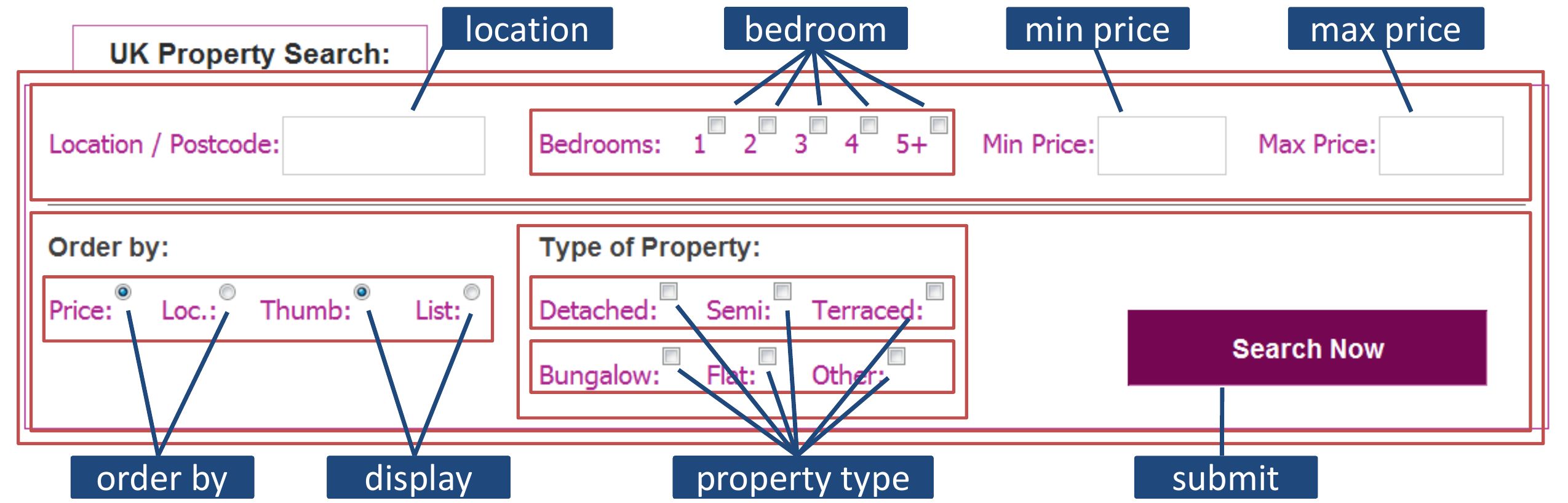} 
  \caption{{\em Farlowestates} before model repair}
  \label{fig:farlowestates-segment-origin}
\end{figure*}

\paragraph{Structural constraints.}
The structural constraints and templates in the real estate and used car
domains are shown in Figure~\ref{tab:compl-c} (omitting only the instantiation as
in the classification case). All segment templates require that there is an
outlier among the siblings of the segment: #outlier<C>(G)# holds if at least one
of $G$'s siblings is not a $C$ segment. 
\begin{inparaenum}[\bfseries(1)]
\item \emph{Basic segment.} A segment is a $C$ segment, if its children are only
  other segments or concepts typed with $C$. This is the dominant segmentation rules,
  used, e.g., for #room#, #price#, or #property_type# in the real estate
  domain. 
\item \emph{Minmax segment.} A segment is a $C$ segment, if it has at least two
  field children typed with $C_M$ where $C_M \SUBCLASS C$ is the minmax type for
  $C$. This is used, e.g., for #price# and #bedroom# range segments.
\item \emph{Segment with mandatory unique.} A segment is a $C$ segment, if its
  children are only segments or concepts typed with $C$ except for one
  (mandatory) field child typed with $U$ where $U \not\SUBCLASS C$. This is
  used, e.g., for #geography# segments where only one #radius# may occur. 
\end{inparaenum}
\paragraph{Repairing form interpretations.}
The classification yields a form interpretation $F$, that is, however,
not necessarily a model under $\Sigma$, and may contain violations of
structural constraints. We adapt the types of fields and segments and
the segment hierarchy of $F$ with the rewriting rules described below
to construct a form model compliant with $\Sigma$. \OPAL performs the
rewriting in a stratified manner to guarantee termination and
introduces at most $n$ new segments where $n$ is the number of fields
in the form.
\begin{asparaenum}[\bfseries(1)]
\item \emph{Under Segmentation:} If there is a segment $n$ with type $t$ such
  that $\mathcal{C}_\mathcal{T}(t)$ requires additional child segments of type
  $t_1, \ldots, t_k \not\in\childTypes(n)$, we try to partition the children of
  $n$ into $k+1$ partitions $P_1, \ldots, P_k, P_n$ such that $P_i \models
  \mathcal{C}_{\mathcal{T}}(t_i)$ and $P_n \cup \{t_1, \ldots, t_k\} \models
  \mathcal{C}_\mathcal{T}(t)$. For each $P_i$ we add a new segment node as child
  of $n$, classify it with $t_i$, and move all nodes assigned to $P_i$ from $n$
  to that segment.  If there is a segment $n$ without type or with type $t$, but
  for which $\childTypes(n) \not\models \mathcal{C}_\mathcal{T}(t)$ and the
  above case can not be applied, then that segment may be split: If there are
  non-overlapping subsequences $c_i$ of children of $n$, such that all children
  of $n$ are covered and, for each $c_i$, there is a type $t_i$ such that the
  types of $c_i$ satisfy the constraints for $t_i$, then we replace $n$ with a
  sequence of segments, one for each $c_i$ typed with $t_i$. 
  
  In practice, few cases of multiple under segmentations occur at the same node
  and we can limit the search space using a total order on $\mathcal{T}$. We
  observe that the number of segments is bounded by the number of fields in the
  form and provide a pool of unused segments in the segmentation. This avoids
  the need for value invention in the model repair.
\item \emph{Over Segmentation:} If there is a segment $n$ of type $t$ with
  children $c_1, \ldots, c_k$ such that $\bigcup \childTypes(c_i) \cup
  \bigcup_{n' \in C} \tau(n') \models \mathcal{C}_{\mathcal{T}}(t)$ where $C$ is
  the set of children of $n$ without $c_1\dots c_k$, then we move the
  children of each $c_i$ to $n$ and delete all $c_i$.
\item \emph{Under Classification:} If there is a segment $n$ of type $t$
  with untyped children $c_1, \ldots, c_k$ and corresponding types $t_1, \ldots,
  t_k$ such that $\childTypes(n) \cup \{t_1, \ldots, t_k\} \models
  \mathcal{C}_{\mathcal{T}}(t)$ and, for each $c_i$, $\childTypes(c_i)
  \models \mathcal{C}_{\mathcal{T}}(t_i)$ holds, then we type $c_i$ with $t_i$.
\item \emph{Over Classification:} If there is a segment node $n$ of type $t$
  with child $c$ typed $t_1$ and $t_2$ such that $\{t_1 \} \cup \bigcup_{c' \in
    C} \tau(c') \models \mathcal{C}_{\mathcal{T}}(t)$ where $C$ is the set of
  children of $n$ without $c$, we drop $t_2$ from $\tau(c)$.
\item \emph{Miss Classification:} If there is a node $n$ of type $t$ where
  $\childTypes(n) \not\models \mathcal{C}_{\mathcal{T}}(t)$, then we delete the
  classification of $n$ as $t$.
\end{asparaenum}


\begin{figure}[tbp]
	\centering 
	  \subfloat[price range]{\label{fig:farlowestates-range}\includegraphics[width=.36\textwidth]{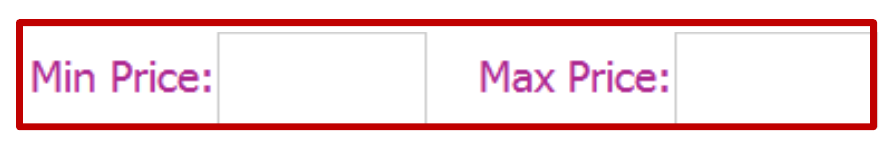}}\\
	  \subfloat[order-by and display-method]{\label{fig:farlowestates-segment-under}\includegraphics[width=.3\textwidth]{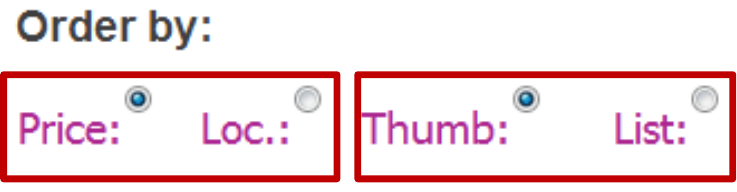}}\\
	  \subfloat[property type]{\label{fig:farlowestates-segment-over}\includegraphics[width=.3\textwidth]{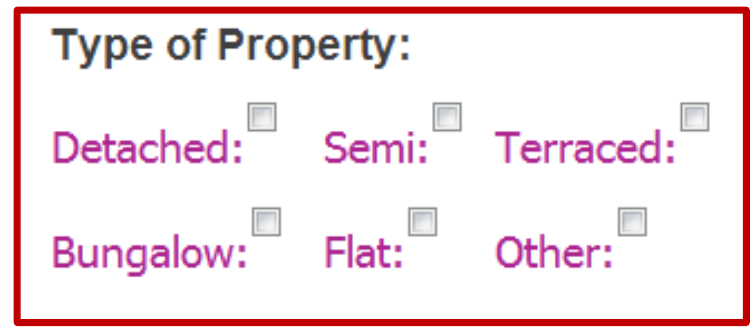}}
	\caption{Model Repair on {\em Farlowestates} Real Estate Form}
	\label{fig:farlowestates-segment}
\end{figure}


Figure~\ref{fig:farlowestates-segment-origin} shows the segmentation and
classification \OPAL obtains for this form before model repair. There are
several problems with this segmentation:
\begin{asparaenum}[\bfseries(1)]
\item The #min_price# and #max_price# fields are not arranged into a range
  segment as no such node is present in the DOM. This is a case of under
  segmentation. Following the #segment_range# constraint, \OPAL introduces a
  price range segment to include both fields as in
  Figure~\ref{fig:farlowestates-range}.
\item The four radio buttons under ``order by'' are of two different domain
  types, i.e., #order_by# for the first two and #display# for the last two. Due
  to #concept_by_segment# from Figure~\ref{fig:classification-constraints} and
  the segment label ``order by'', the last two would also get classified as
  #order_by#, if not for #display# $\PRECEDENCE$ #order_by#.  This is an example
  of the second case of under segmentation, where \OPAL needs to split the
  existing segment as it is not supported by a structural constraint, but there
  are subsequences of children that can form valid segments
  (Figure~\ref{fig:farlowestates-segment-under}).
\item As a result of the original segment with four radio buttons grouped
  together, the last two radio buttons in the four are also typed as #order_by#
  in addition to their #display# type. \OPAL resolves this over classification
  by removing the #order_by# following the restructuring of the segment.
\item The #property_type# segment is subdivided into two segments in the
  original segmentation, since \OPAL identifies no style-equivalence among the
  six check boxes due to lack of similarity. However, two segments of
  #property_type# can not be contained in a single parent segment (due to
  #outlier#). Thus, the two segments are removed with all their children
  directly contained in the larger segment
  (Figure~\ref{fig:farlowestates-segment-over}). This is an example of over
  segmentation.
\item The segmentation obtained at segment scope preserves the two DOM nodes
  representing two form rows. However, in the domain schema, these nodes do not
  carry meaning, and thus are treated as over segmentation and removed.
\end{asparaenum}

\subsection{Domain Instantiation: Methodology and Example}
\label{sec:domain:methodology-example}

In this section, we demonstrate how to derive an \OPAL domain schema, which
includes form specific concepts, from a given standard ontology of a domain. This is
the typical way to instantiate a domain for use with \OPAL.

\begin{figure*}[tbp]
  \centering
  \includegraphics[width=.8\linewidth]{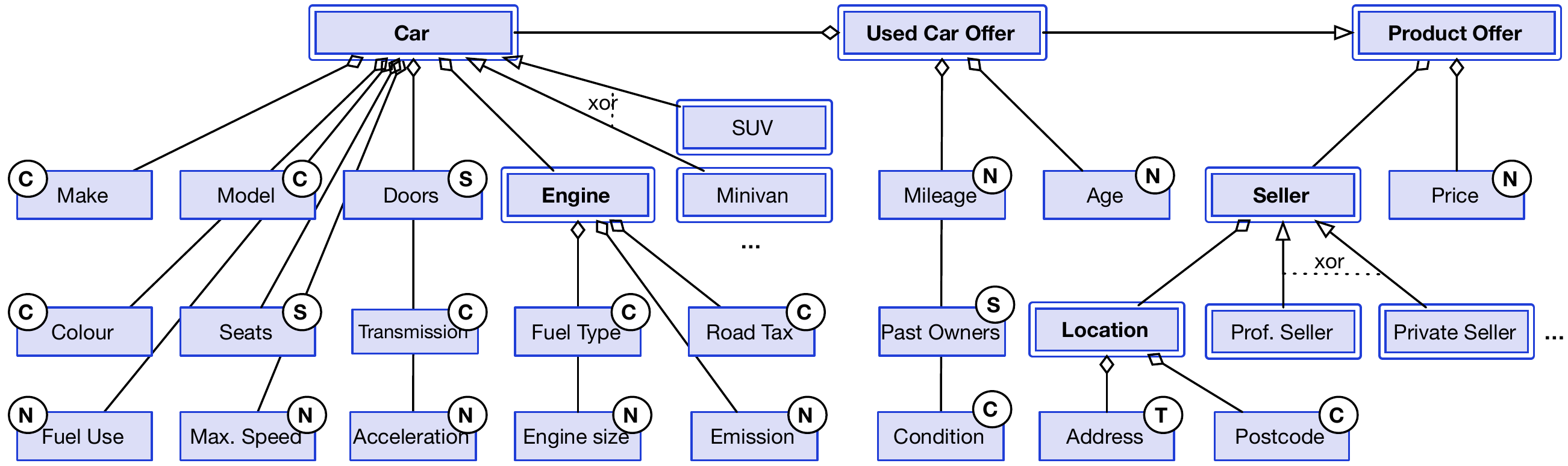}
  \caption{Used car ontology}
  \label{fig:used-car:ontology}
\end{figure*}

Figure~\ref{fig:used-car:ontology} shows a simple ontology for the used car
domain (in the UK). Note, that most search forms are about searching for
entities (double border in Figure~\ref{fig:used-car:ontology}) by their
properties (single border) such as price or mileage of a car. Therefore, most of
the types in an \OPAL domain schema correspond to such properties of entities in
the domain.

We observe that properties can be roughly distinguished into numerical,
categorical, and free text according to their range and that these distinctions
dictate to a large extent the expected form fields for searching by those
properties. For a numerical property we expect, e.g., either a single text
input or slider, two min-max fields for entering a range, or a set of
checkboxes to select common values or ranges. Categorical properties, on the
other hand, never exhibit range inputs. 

These observations are codified in the derivation templates of
Figure~\ref{fig:meta-model:templates}. These templates group typical
instantiations for the above kinds of properties as well as for
compound object types
such as \TYPE{location} in Figure~\ref{fig:used-car:ontology}: 
\begin{asparaenum}[\bfseries(1)]
\item For an \textbf{object type} (#engine#), we instantiate only the
  #segment<C># template, i.e., we allow segments, but not fields of this
  type. Such segments typically collect multiple properties of the object type,
  e.g., #engine_size# and #fuel_type#.
\item For a \textbf{free text type} (e.g., #address#), we instantiate
  only the #concept_by_proper<C,A># and #concept_by_value<C,A># templates that
  allows fields, but not segments of that type. There is usually no need for a
  segment in this case, as there are rarely multiple occurrences of fields for
  such a type. In the rare case where that is nevertheless possible, we
  instantiate #segment<C># separately. 
\item For a \textbf{categorial type} (#make# or #colour#), we instantiate in
  addition to #concept_by_proper<C,A># also #segment<C># and the
  #concept_by_segment<C,A>#. Categorical types are often represented as single
  select boxes or lists of radio buttons or check boxes. For the latter, an
  enclosing segment is desirable and #concept_by_segment<C,A># allows us to
  propagate the segment labels to the fields.
\item For a \textbf{numerical type} (#price# or #seats#), we also instantiate
  the #segment_range# and #concept_minmax# templates, enabling the 
  classification of range segments and fields. 
\end{asparaenum}

\begin{figure}
\begin{lstlisting}[numbers=left]
TEMPLATE object_type<C> {
  INSTANTIATE segment<C> using { <C> } }
 
TEMPLATE free_text_type<C,A> {
  INSTANTIATE concept_by_proper<C,A> using { <C,A> } 
  INSTANTIATE concept_by_value<C,A> using { <C,A> } }
 
TEMPLATE categorical_type<C,A> {
  INSTANTIATE concept_by_proper<C,A> using { <C,A> }
  INSTANTIATE concept_by_segment<C,A> using { <C,A> }
  INSTANTIATE concept_by_value<C,A> using { <C,A> }
  INSTANTIATE segment<C> using { <C> } }
 
TEMPLATE numeric_type<C,C$_M$, A> {
  INSTANTIATE concept_by_proper<C,A> using { <C,A> }
  INSTANTIATE concept_by_segment<C,A> using { <C,A> }
  INSTANTIATE concept_by_value<C,A> using { <C,A> }
  INSTANTIATE concept_minmax<C,C$_M$,A> using { <C,C$_M$,A> }
  INSTANTIATE segment<C> using { <C> }
  INSTANTIATE segment_range<C,C$_M$> using { <C,C$_M$> } }
\end{lstlisting}
\caption{Template for different property kinds}
\label{fig:meta-model:templates}
\end{figure}

With these templates, we can derive an \OPAL annotation and domain schema very
quickly from a given domain schema such as
Figure~\ref{fig:used-car:ontology}. 

First, we normalize the ontology: If a class $C$ has sub-classes without
additional properties (type classes), we generate a new categorical property
$C\TYPE{\_type}$, add all labels from the sub-classes to that property, and
remove the sub-classes. 

Second, we derive the annotation schema and, in particular, the necessary
annotators as follows:
\begin{asparaenum}[\bfseries(1)]
\item For each concept or property $c$ of the ontology, we create an annotation
  type $\type{c}$. All \emph{labels} of $c$, possibly enriched with synonyms
  from an external knowledge base such as Wordnet, form an annotator for
  the proper labels of the concept ($\PROPER_{\type{c}}$).
\item For categorical concepts or properties, we require a given list of
  instances, an existing annotator, or another entity recogniser, again
  possibly provided by an external knowledge base such as DBPedia or
  LinkedGeoData. Numerical values are treated  similarly, though these often take simply
  the form of number in a certain range. This provides $\VALUE_{\type{c}}$.
\end{asparaenum}

\begin{figure*}[tbp]
  \centering
  \includegraphics[width=1\linewidth]{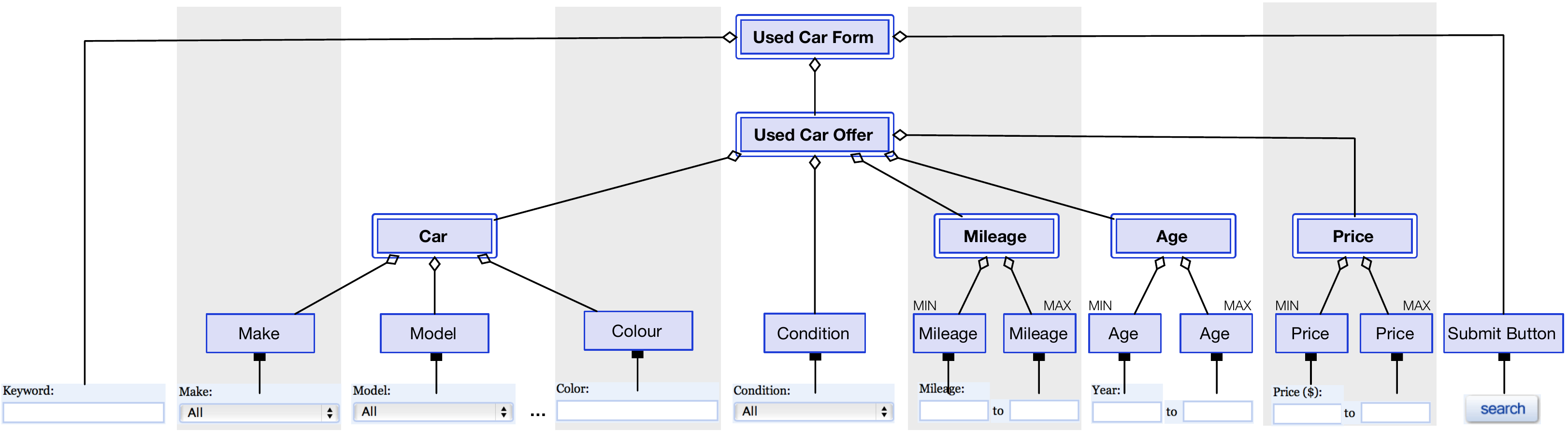}
  \caption{Used car: classified form}
  \label{fig:used-car:classified-form}
\end{figure*}

Third, we derive the domain schema in four steps:
\begin{asparaenum}[\bfseries(1)]
\item For each \textbf{class} $C$, add an instantiation rule for
  #object_type<C>#. In our example, this yields 6 instantiations
  (recall, that type classes are normalised to properties above). 
\item For each \textbf{property}, add an instantiation rule of corresponding
  type, e.g., 
  \begin{lstlisting}
INSTANTIATE numeric_type<C,C$_M$,A> using {<price,price$_M$,$\type{price}$>}
  \end{lstlisting}
  In our example, this yields 22 instantiations (20 properties from
  Figure~\ref{fig:used-car:ontology} and two $\ldots\_\type{type}$ properties). 
\item Determine which ``presentational'' fields and segments occur in
  the given domain and
  add them to the domain schema. A field or segment is presentational, if it
  determines the way the results are represented. In the used car and real
  estate domains, we  identify two types of presentational fields:
  ``order-by'' and ``pagination'' which control the order in which the results
  are presented as well as the number of results per page. These
  presentational types are mostly shared between domains and can be easily
  reused thanks to \OPALlang templates:
  \begin{lstlisting}
INSTANTIATE categorical_type<C, A> using 
  { <order_by, $\type{order\_by}$> <pagination, $\type{pagination}$> }
  \end{lstlisting}
  In this step, we also add generic rules that are independent of the domain,
  e.g., for the form itself and domain-independent form facilities such as
  submit buttons or generic keyword search fields.
\item Sometimes small manual adjustments are necessary. For example, 
  numerical types may occur with multiple units of measure or other
  modifiers, e.g., prices with different currencies or locations with a search
  radius. Such modifier fields are usually unique in their corresponding segment
  and thus added using the #segment_with_unique<C,U># template. In the used car
  domain, we can observe this for #currency# and #radius#:
  \begin{lstlisting}
INSTANTIATE TEMPLATE segment_with_unique<C,U> using
  { <price, currency> <location, radius>  }
INSTANTIATE TEMPLATE concept_by_proper<C,A> using
  { <currency,$\type{currency}$>, <radius,$\type{radius}$> }
INSTANTIATE TEMPLATE concept_by_value<C,A> using
  { <currency,$\type{currency}$>, <radius,$\type{radius}$> }
  \end{lstlisting}
  Some object types, in particular #location#, may also be entered as a whole
  through free text fields and accordingly instantiate the #free_text_type#
  template for them:
  \begin{lstlisting}
INSTANTIATE TEMPLATE free_text_type<C,A> using
  { <location,$\type{location}$ }
  \end{lstlisting}
\end{asparaenum}

Finally, we need to determine part-of and precedence between types. The part-of
relation is derived from the associations of the domain schema, e.g., #address#
$\PARTOF$ #location#, #postcode# $\PARTOF$ #location#, #fuel_type# $\PARTOF$
#engine# for our case. Precedence requires some observation of cases where
annotations for different types overlap. Typically, we want to give
presentational types precedence over all domain types (as they often contain
values such as ``sort by price''). For the used car domain, we observe that
#pagination# \PRECEDENCE #order_by# and that both have precedence over all domain
types. We also observe that #mileage# and #radius# (in locations) can have
overlapping values. Though radius is only used in #segment_with_unique<C,U>#, for
#location# segments which disallow #mileage# elements, we add #mileage#
\PRECEDENCE #radius# to express a bias for #mileage#.

Figure~\ref{fig:used-car:classified-form} shows a form from the used car domain
fully classified according to this domain schema.

\section{Light-weight Form Integration}
\label{sec:filling}

\OPAL's form models allow the easy implementation of many types of applications
that require automatic understanding and interaction with forms, such as form
integration and filling, data extraction, or web automation. As discussed in
Section~\ref{sec:approach}, we focus here on \emph{form integration} (or
filling), i.e., the part of a web integration system
\cite{He:2005:TBM:1053724.1054143} that translates a query on the global schema
(\OPAL's domain schema) to a query against concrete forms. In this section, we
introduce a light-weight form integration system that performs this task fully
automatically for thousands of forms in a domain, given only an \OPAL domain
schema. We have instantiated this system for the real estate and used car
domain, but \OPAL is  as easily  applied to other domains, since only
a very limited amount of additional customisation is needed (on type variations
and, possibly, similarities).

Recall, that we focus on the optimistic, single-query variant of the form
integration problem: We aim for a single-query that returns all results
matching the global (or \emph{master}) query, but allow to return also
non-matching results, if there is no more specific query that returns all
matching ones.

\OPAL's form integration translates the master query into concrete queries
through a small set of translation rules supported by a notion of similarity on
property values.  \OPAL can perform form integration without any other
information than what is provided by an \OPAL domain schema and corresponding
form model. However, it can be further improved by providing additional
domain-specific information.

\textbf{Similarity on values} is represented as a real-valued function on pairs
of values and is based on the property type: For free-text and categorical
properties, \OPAL uses a mix of Levenshtein and longest common substring
distance, for numeric properties a difference-based similarity.  A domain schema
can be enhanced by property-specific similarity function, e.g., to deal with
different units of measure. A small set of such functions is provided with
\OPAL: for price, for distance properties, and for dates.

\textbf{Translation rules} use these similarity functions to translate the
constraints of the master query $Q$ into queries on the concrete forms. For each
form $F$ with form model $M$ and constraint $C \in Q$ on type $T$, we retrieve
the fields $f_1, \ldots, f_n$ classified with $T$. Let $\VALUES(C)$ be the
(possibly infinite) set of values for which $C$ holds.
\begin{compactenum}[\bfseries(1)]
\item \emph{Single field, single value:} If $n=1$, $\VALUES(C) = \{v\}$, and
  \begin{compactenum}[\sffamily(i)]
  \item $f_1$ is a free text input, return $f_1 = v$.
  \item $f_1$ is a select box, return $f_1 = v'$ where $v'$ is the
    option of $f_1$ most similar to $v$.
  \end{compactenum}
\item \emph{Multi field:} If $n \geq 1$, 
  \begin{compactenum}[\sffamily(i)]
  \item $\VALUES(C) = \{v\}$, and all $f_i$ are radio buttons (exclusive options), return $f_k =
    \TRUE$ for the $f_k$ that is most similar to $v$.
  \item $\VALUES(C) = \{v_1, \ldots, v_k\}$ and all $f_i$ are check boxes
    (non-exclusive options), return $f_k = \TRUE$ for each $f_k$ where a $v_i$
    exists such that the similarity of $f_k$ and $v_i$ is minimal among all such
    pairs.
  \item and all $f_i$ are free-text range input fields (i.e., of type $T_M$,
    where $T_M$ is the minmax type to $T$), then return $f_s = v_1$ for each
    $f_s$ that is a minimum input and $f_e = v_k$ for each $f_k$ that is a
    maximum input.
  \item and all $f_i$ are select-box range input fields, then return $f_s =
    v_1'$ for each $f_s$ that is a minimum input where $v_1'$ is the most
    similar option of $f_s$ to $v_i$ that is smaller or equal to $v_1$. Analog
    for $f_e$. 
  \end{compactenum}
\end{compactenum}
In all other cases (e.g., a select box for a set inclusion constraint), we
return no constraints to avoid false negatives.

In many domains, we can observe that the same information is represented in
alternative ways on different sites. E.g., the age of a car is represented by
the manufacturing year on same sites. Similarly, the location of property may be
given as a street address, a postcode, or even just a town, in particular for
rural agencies. To treat this cases, we need to be able to translate a
constraint such as ``$\TYPE{age} = 6$'' to a constraint ``$\TYPE{year} = 2006$''
or ``$\TYPE{postcode} =$ OX1'' to ``$\TYPE{town} =$ Oxford''. We call \TYPE{age}
and \TYPE{year} \emph{type variants} and amend the domain schema with a value
mapping for each pair of type variants. Value mappings for numerical properties
are typically simple conversion functions, e.g., from different units of
measure. Value mappings for categorical properties are typically realised by a
query to an external database or service such as DBPedia. In our example
domains, we use value mappings for conversions of metric and imperial distances
as well as of postcodes to towns and other locations. To treat type variants we
perform the following test and translation before the aforementioned translation
rules:
\begin{compactenum}[\bfseries(1)]\setcounter{enumi}{-1}
\item \emph{Type variants.} If $n = 0$ and there is a field $f'$ with type $T'$ such that $T'$ is a
  variant type of $T$, we translate the values in $C$ to $T'$ and continue with
  that constraint.
\end{compactenum}

\begin{figure}[tbp]
  \centering
  \includegraphics[width=\columnwidth]{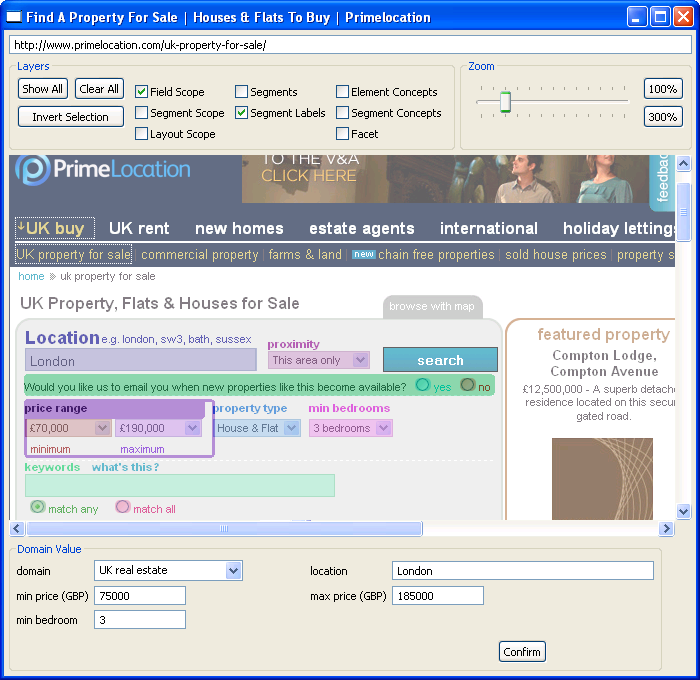} 
  \caption{\OPAL Testing Tool}
  \label{fig:opal-gui}
\end{figure}

With those simple rules, \OPAL's form integration manages to translate most
constraints as shown in Section~\ref{sec:evaluation}. There are, of course,
still cases where the translation fails, e.g., if categorical values are mapped
to ranges by some ordering such as road tax brackets or iPhone models (ordered
according to year of introduction). But as demonstrated in
Section~\ref{sec:evaluation}, this light-weight simple form integration already
provides us with a successful translation of a master query in the vast majority
of cases.

To illustrate \OPAL's form integration, we consider the form of
\url{primelocation.com} as shown in the middle of Figure~\ref{fig:opal-gui}. The
figure shows the \OPAL testing tool that we use to test and verify the accuracy
of \OPAL domain schemas. It allows the user to visualize the form labels, form
segments, and classifications derived by \OPAL and to track down, where, e.g.,
there are problems with the classification constraints or the annotations. It
also provides a master query in the lower third. The concrete form is
automatically filled according to the values provided in the master form. This
allows the user to visually verify that the query has been translated
correctly. The master form is automatically generated from the domain schema,
but the user can provide additional information on which fields to include. For
space reasons, we have focused in Figure~\ref{fig:opal-gui} on the types 
most commonly used in constraints in the UK real estate domain. 

For the concrete form from \url{primelocation.com}, we highlight form fields and
labels by colouring them with the same color (here, e.g., the ``minimum'' and
the first price field). Form segments are shown as boxes with no filling except
for their labels (a price segment with ``price range'' label).  The figure shows
the form \emph{after} \OPAL has filled it according to the values from the
master query. Notice, how for the three select boxes for minimum and maximum
price, as well as bedroom number, \OPAL picks the closest value to the one
specified in the master form. 

\section{Evaluation}
\label{sec:evaluation}


\begin{figure}[tbp]
  \centering
  \includegraphics[width=.9\columnwidth]{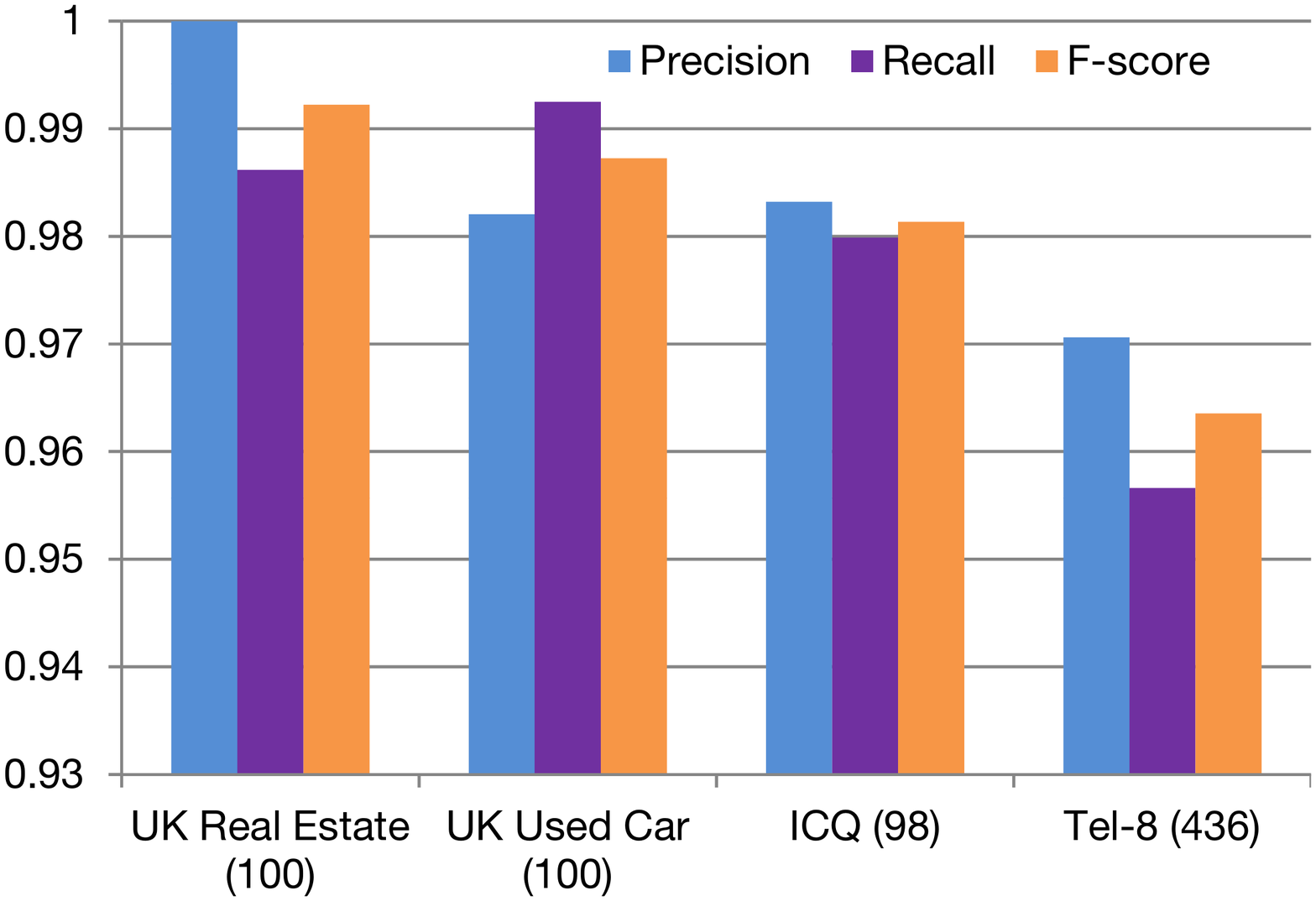} 
  \caption{\OPAL on 734 forms}
  \label{fig:comparison-result}
\end{figure}


We perform experiments on several domains across four different
datasets. Two datasets are randomly sampled from the UK real estate
and UK used-car domains, respectively.
%
%
We compare with existing approaches via ICQ and TEL-8, two public benchmark
sets, on which we only evaluate \OPAL's form labeling. This limitation necessary
to allow a comparison that is fair to existing approaches, that only label forms
and do not use domain knowledge.
Even with this limitation, however, \OPAL outperforms previous approaches in
most domains by at least 5\%.
We also perform an introspective analysis of \OPAL to show
\begin{inparaenum}[\bfseries(1)]
\item the impact of field, segment, layout, and repair in the form interpretation, 
\item \OPAL's performance and scalability with increasing page size, and
\item the effectiveness of the form integration in \OPAL. 
\end{inparaenum}

We evaluate the proper assignment of text nodes to form fields using standard
notions of precision, recall and F-score (harmonic mean $F = F_1=2PR/(P+R)$ of
precision and recall). For form labeling (classification), precision $P$ is
measured as the proportion of correctly labeled (classified) fields over total
labeled fields, while recall $R$ is the fraction of correctly labeled fields
over total number of fields. For form filling precision and recall do not apply
and we therefore report the error rate as portion of total fields that are not
correctly filled (i.e., either filled but with a wrong value or not filled at
all, despite a corresponding constraint in the master query).  For all
considered datasets, we compare the extracted result to a manually constructed
gold standard. We evaluate segmentation through their impact on classification,
see Figure~\ref{fig:scopes-contribution}, and the improved performance on the
two datasets where we perform form interpretation (UK real estate and used car)
versus the ICQ and TEL-8 datasets.

\paragraph{Datasets.}
For the UK real estate domain, we build a dataset randomly selecting $100$ real
estate agents from the UK yellow pages (\url{yell.com}).  Similarly, we randomly
pick $100$ used-car dealers from the UK largest aggregator website
\url{autotrader.co.uk}.  The forms in these two domains have significantly
different characteristics than the ones in ICQ and TEL-8, mainly due to changes
in web technology and web design practices.  The usage of CSS stylesheets for
layout and AJAX features are among the most relevant.
 
The ICQ and TEL-8 datasets cover several domains.  ICQ presents forms from five
domains: air traveling, (used) cars, books, jobs, (U.S.) real estate.  There are
20 web pages for each of the domains, but two of them are no longer accessible
and thus excluded from this evaluation.  TEL-8, on the other hand, contains
forms from eight domains: books, car rental, jobs, hotels, airlines, auto, movies
and music records.  The dataset amounts to $477$ forms, but only $436$ of them
are accessible (even in the cached version).

\begin{figure*}[tbp]
  \centering
  \subfloat[ICQ results]{\label{fig:icq-result}\includegraphics[width=.45\textwidth]{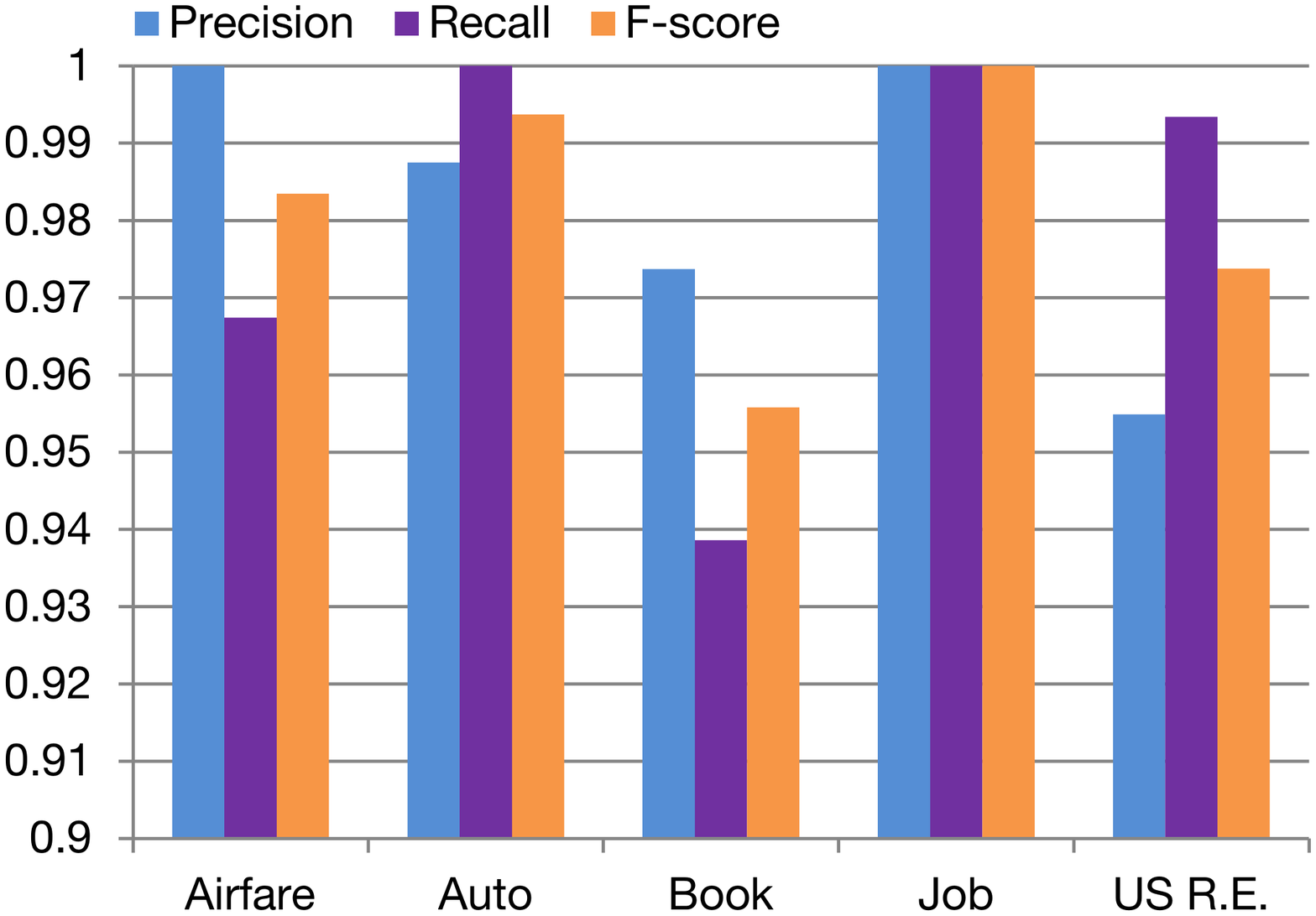}}
  \subfloat[TEL-8 results]{\label{fig:tel8-result}\includegraphics[width=.5\textwidth]{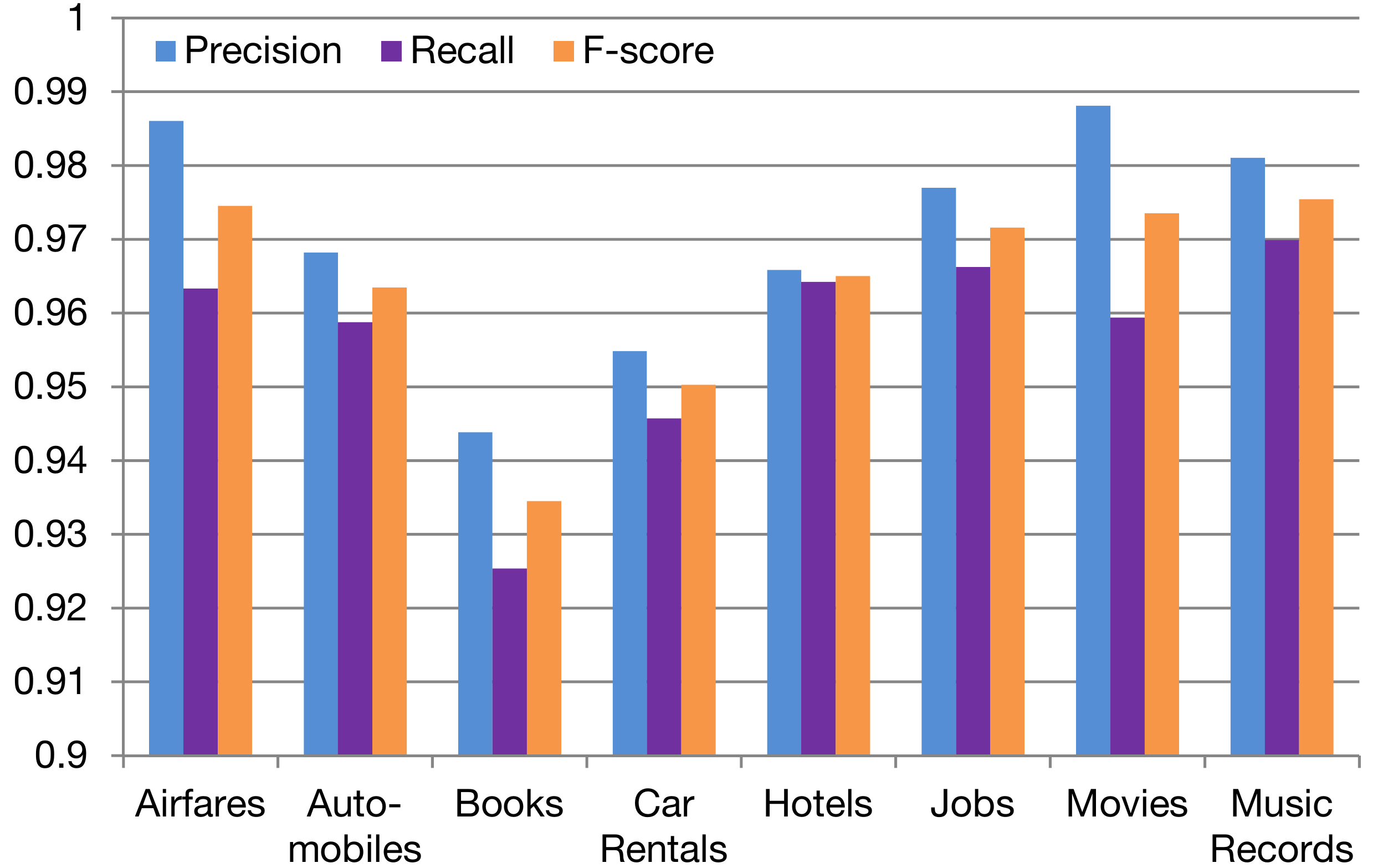}}
  \caption{\OPAL on ICQ and TEL-8 benchmark}
\end{figure*}

\subsection{Field Labeling}
\label{sec:case-study}
In our first experiment we evaluate the accuracy of \OPAL's field labeling on
all four datasets, but only in the UK real estate and used car domain we employ
the form interpretation to further improve the field labeling.
Figure~\ref{fig:comparison-result} shows the results.  The first two bars are
for the random sample datasets. For the real estate domain, \OPAL classifies
fields with perfect precision and $98.6\%$ recall. Overall we obtain a
remarkable $99.2\%$ F-score.  The result is similar for the used car domain,
where \OPAL obtain $98.2\%$ precision and $99.2\%$ recall, that amount to
$98.7\%$ F-score. \OPAL achieves lower precision than recall in the used car
domain due to the fact that web forms in this domain are more interactive:
certain fields are enabled only when some other field is filled properly, yet
non-field placeholders are present in the HTML to indicate that a field will
appear when the other field is filled.  This introduces noise to field
labeling and thus classification.  

For the real estate domain, our domain schema consists of a few dozen field
and segment types and about $40$ annotation types.  Similarly, in the used car
domain, there are about $30$ annotation types. Creating an
initial domain schema (including gazetteers and testing) 
takes a single person familiar with a domain and $\OPALlang$ roughly $1$
week.

The other two entries in Figure~\ref{fig:comparison-result} regard field
labeling on ICQ and TEL-8 datasets.  On these, \OPAL applies only its
domain-independent scopes (field, segment, scope) as no domain schema is
available for these domains.  Nonetheless, \OPAL reports very high accuracy also
on these forms, confirming the effectiveness of our domain-independent analysis.
Not unexpected, \OPAL performs better in the UK real estate and used
car domain where domain knowledge is present, even though the forms in those
datasets are on average more modern and contain more fields ($10.4$ and $9.2$
fields per form in the real-estate and used-car dataset versus $6.5$ and $7.9$
fields per form for ICQ and Tel-8).

\paragraph{Cross Domain Comparison.}
\label{sec:cross-domain-eval}
We use ICQ and TEL-8  to compare field labeling in \OPAL against
existing approaches, on a wide set of domains.
%
%
%
Figure~\ref{fig:icq-result} details the result of \OPAL on each domain of the
ICQ dataset.  It shows perfect F-score values for the jobs domain ($100\%$) as
well as auto and air travelling ($99.3\%$ and $98.3\%$). For comparison,
\cite{dragut09:_hierar_approac_to_model_web} reports $92\%$ F-score for labeling
on ICQ on average, which we outperform even in the domain most difficult for \OPAL
(books). \cite{wu09:_model_and_extrac_deep_web_query_inter} reports slightly
better precision and recall than \cite{dragut09:_hierar_approac_to_model_web},
but \OPAL still outperforms it by several percents. 



The results for the TEL-8 dataset are depicted in
Figure~\ref{fig:tel8-result}. Here, the overall F-score is $96.3\%$, again
mostly affected by the performance in the books domain. Note that, especially on
TEL-8, \OPAL obtains very high precision compared to recall.  Indeed, lower
recall means \OPAL is not able to assign labels to all fields, missing some of
them. 
For comparison, \cite{dragut09:_hierar_approac_to_model_web} reports $88-90\%$
overall F-score, which we outperform by a wide
margin. \cite{nguyen08:_learn_to_extrac_from_label} reports F-scores between
$89\%$ and $95\%$ for four domains in the TEL-8 dataset.  Though they perform
slightly better on books, we significantly outperform them on the three other
domains included in their results, as well as on average.

In Section~\ref{sec:domain-dependent}, we discuss that \OPAL prioritises field
over segment over layout scope and we claim that this is due to the decreasing
precision. Table~\ref{tab:scope-fp} shows the total number of fields labeled in
each scope, as well as the number and percentage of false positives among those
labels. It illustrates that, indeed, the field scope produces almost no false
positives ($2$ out of $762$ fields labeled in this scope, i.e., $0.3\%$),
the segment scope also produces very few ($3$ out of $154$ labeled fields), and
the layout scope produces most ($8$ out of $72$ labeled fields). 
\begin{table}[tbp]
\centering
\begin{tabular}{ r r@{}l r@{}l r@{}l}
  \toprule
  & \multicolumn{6}{c}{\textbf{labeling}} \\
  & \multicolumn{2}{c}{\textbf{field}} & 
  \multicolumn{2}{c}{\textbf{segment}} & 
  \multicolumn{2}{c}{\textbf{layout}}\\
  \midrule
  total & 761& & 154& & 72& \\
  false positives & 2& & 3& & 8& \\
  \midrule
  \quad\quad $= \%$ & 0&.3$\%$ & 1&.9$\%$ & 11&.1$\%$\\
  \bottomrule
\end{tabular}
\caption{False positives}
\label{tab:scope-fp}
\end{table}
%


\lstset{emph=[1]{},emph=[2]{},emph=[1]{nbsp},alsoletter={&},emphstyle=[1]{\normalfont\ttfamily\small\color{lightgray}}}

\newsavebox{\listingBookbeat} 
\begin{lrbox}{\listingBookbeat}
\end{lrbox} 
\begin{figure}[tbp]
	\centering 
\includegraphics[width=.8\columnwidth]{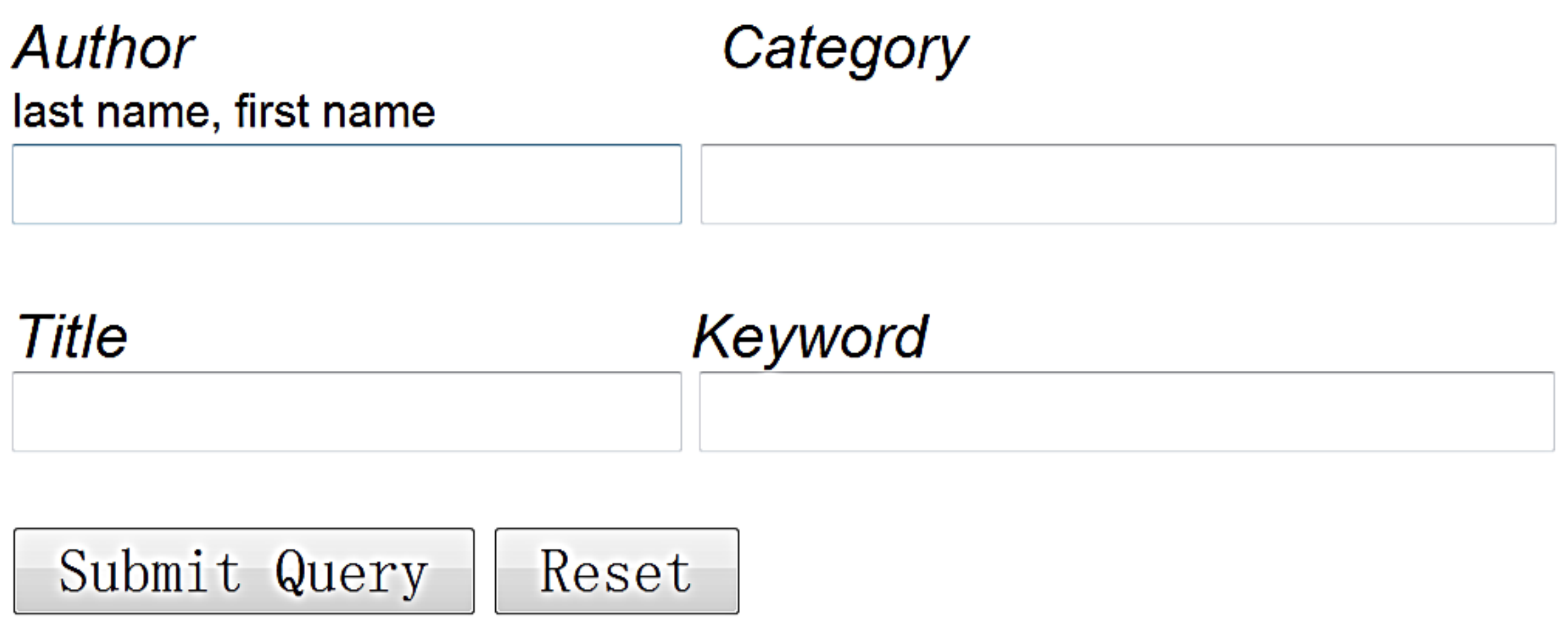}
\begin{lstlisting}[language=Pascal,basicstyle=\normalfont\ttfamily\footnotesize]
<font face="Arial,Helvetica,sans-serif" size="2">
  <i>Title&nbsp;&nbsp;&nbsp;&nbsp;&nbsp;&nbsp;&nbsp;...
     &nbsp;&nbsp;&nbsp;Keyword
    <br />
    <input name="title" > 
    <input size="26" name="keyword" /> 
  </i>
</font>
\end{lstlisting}
	\caption{{\em The Bookbeat} form with source}
	\label{fig:bookbeat}
\end{figure}

\newsavebox{\listingAbbey} 
\begin{lrbox}{\listingAbbey}
\end{lrbox} 
%

What keeps \OPAL from achieving $100\%$ accuracy? Most of the cases are due to
\OPAL's assumption that form labels are separate text nodes. This is evidently
the case in most forms, as demonstrated by near perfect accuracy, but there are
some outliers that use image only labels or merge multiple labels into one node
and use whitespace to achieve the desired result.
Figure~\ref{fig:bookbeat}, e.g., shows a form where ``Title'' and ``Keyword'' are a
single HTML node with \texttt{\&nbsp;} spaces in between. While both cases are
easy enough to address, they do require specific treatment and we omitted them
from the version of \OPAL presented here to illustrate that even without any
such specifically tailored heuristics, we can achieve nearly perfect form
labeling and interpretation.


\subsection{Form Interpretation}
The quality of \OPAL's form interpretation depends on the quality of the form
labeling and that of the annotators. As discussed above, for this evaluation we
use annotators that have been created in about 1 week for the UK real estate and
used car domain. The location related annotators are based on standard sources
(GeoNames and LinkedGeoData) and thus have reasonable recall, but precision is
fairly low, due to the high number of locations in the UK that are homonyms to
common English words (e.g., the town of ``Selling''). Such noise in the value
annotators, however, affects \OPAL very little, as the values of form fields are
only used if the labels are inconclusive and we only use the most frequent
annotation type. Noise in the label values is far more likely to lead to
classification errors. However, typical annotators are small lists of $5-10$
typical labels which are easy to create and have very low noise. E.g., for
$\type{bedroom}$ labels we use just ``bedroom'', ``bed'', and their plural
forms, for \type{make}, \type{model}, \type{mileage} and many more just
``make'', ``model'', ``mileage'', and their plural form, resp.

With this, we achieve near perfect classification, correctly classifying most of
the fields, see Table~\ref{tab:scope-fp}: Precision is $97.3\%$ over all fields
in the real estate data set (with just $24$ out of $931$ classified
fields incorrectly classified) and recall $97.4\%$. This excludes $56$ (or
$5.5\%$) fields for which our domain schema does not contain a concept (usually
as they appear only very rarely).

\begin{figure}[tbp]
  \centering
  \includegraphics[width=.5\linewidth]{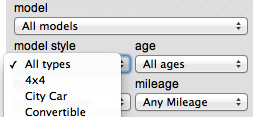}
  \caption{Classification error example}
  \label{fig:classification:error}
\end{figure}

Classification errors are mostly caused by ambiguity in the used form
labels. For example, Figure~\ref{fig:classification:error} shows a form, where
the ``model style'' field is erroneously classified as a \TYPE{model} field by
\OPAL. The field has a proper label ``model style'' which is correctly assigned
to the field in the field labeling, as are the field values ``4x4'', ``City
Car'', etc. In the classification, we prioritise proper labels over values (as
value annotators are more noisy). In most cases, this is indeed preferable, but
here the proper label ``model style'' is annotated with \type{model} and we
classify the field as \type{model} rather than \type{car\_type}, as ``model
style'' is not recognised as a label for \type{car\_type}. A probabilistic
classifications that combines classifications from labels and values (with a
lower weight) would allows us to choose the most likely global form classification
and thus to address such outliers. However, this would also increase the effort
in creating a domain schema.


\subsection{Contributions of Scopes}
\label{sec:scopes-contribution}
We demonstrate the effectiveness of combining different types of analysis by
measuring to what extent each of our four scopes contributes to the overall
quality of form understanding. We use again the two
domain datasets from the previous experiment.
%
%
For both we show the results for recall, though the picture is similar for
precision and F-score, cf.\@ Figure~\ref{fig:comparison-result}. As illustrated
in Figure~\ref{fig:scopes-contribution}, for the field labeling in the
real-estate dataset, the field scope already contributes significantly (67\%).
The Segment scope increases recall by 18\%, layout scope and the repair in the
form interpretation add together another 13\%. Note that, the contribution of
the repair in the form interpretation is more significant than that of the
layout scope, indicating the importance of domain knowledge to achieve very high
accuracy form understanding. In the used car domain, field scope alone is even
more significant 85\% (as many of the websites use modern web technologies and
frameworks with reasonable structure). 


\begin{figure}[tbp]
  \centering
  \includegraphics[width=.75\columnwidth]{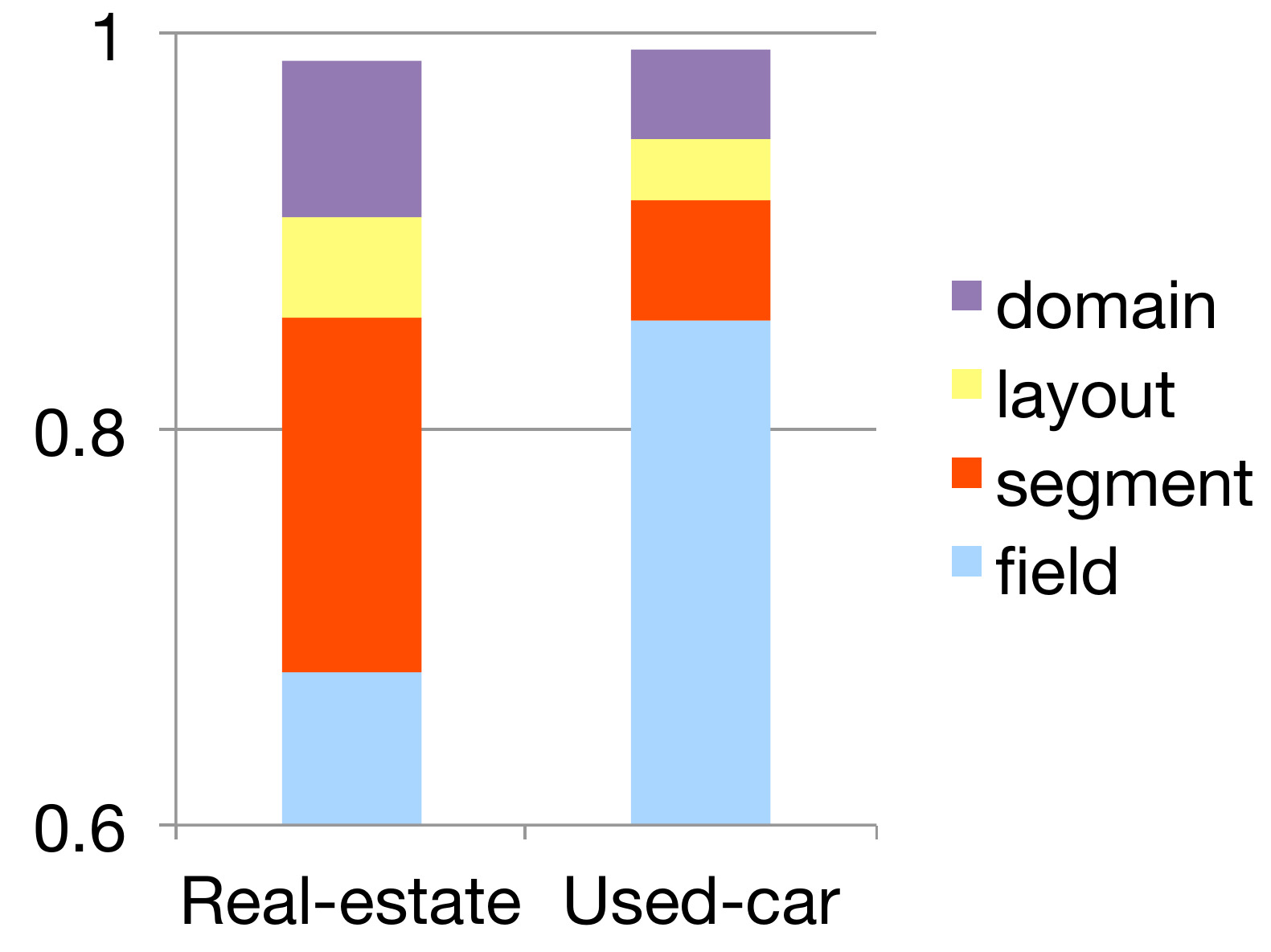} 
  \caption{Scopes}
  \label{fig:scopes-contribution}
\end{figure}





\begin{figure}[tbp]
  \centering
  \includegraphics[width=.75\columnwidth]{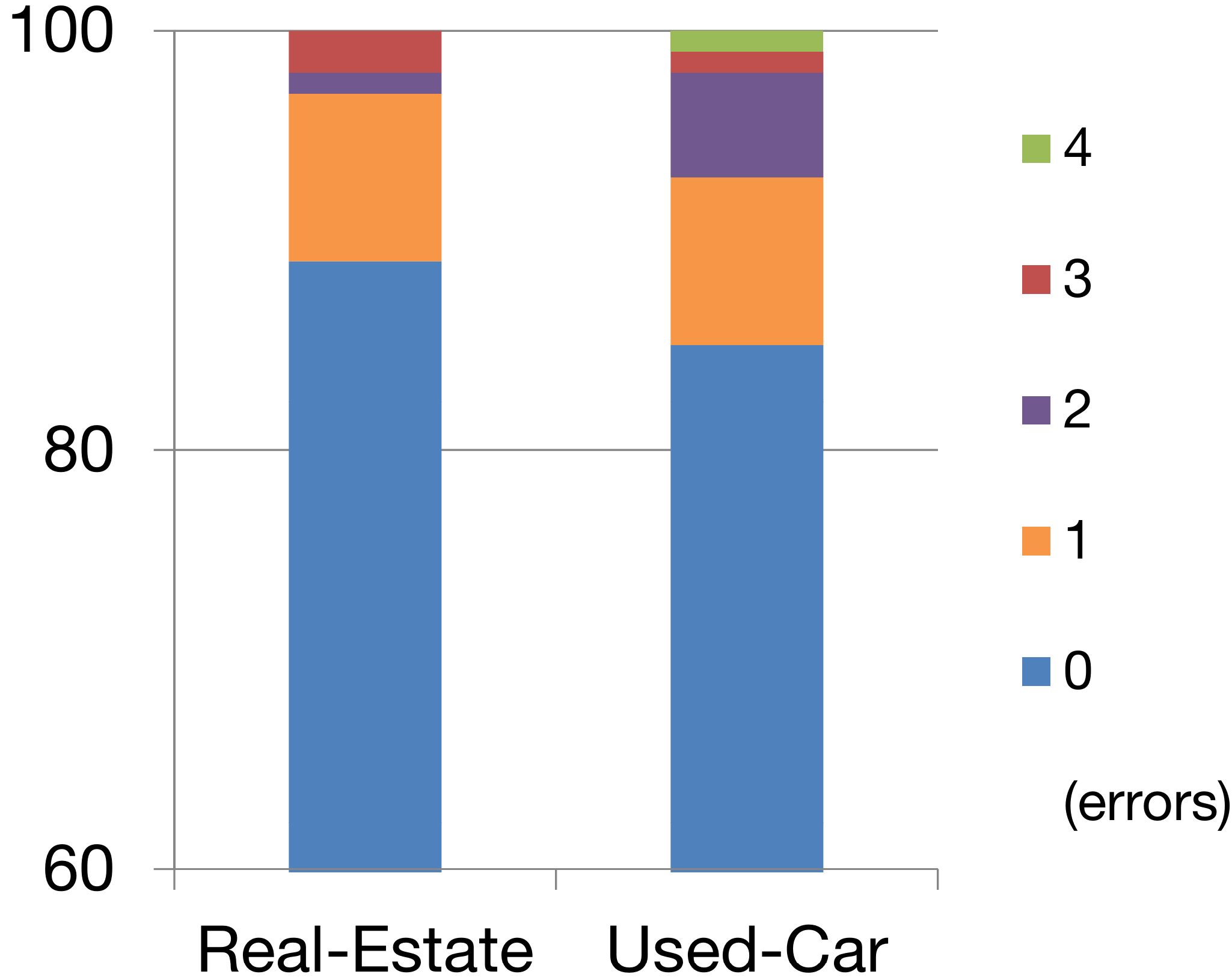} 
  \caption{Form integration errors (per form)}
  \label{fig:filling-failure}
\end{figure}

\begin{figure}[tbp]
  \centering
  \includegraphics[width=.95\columnwidth]{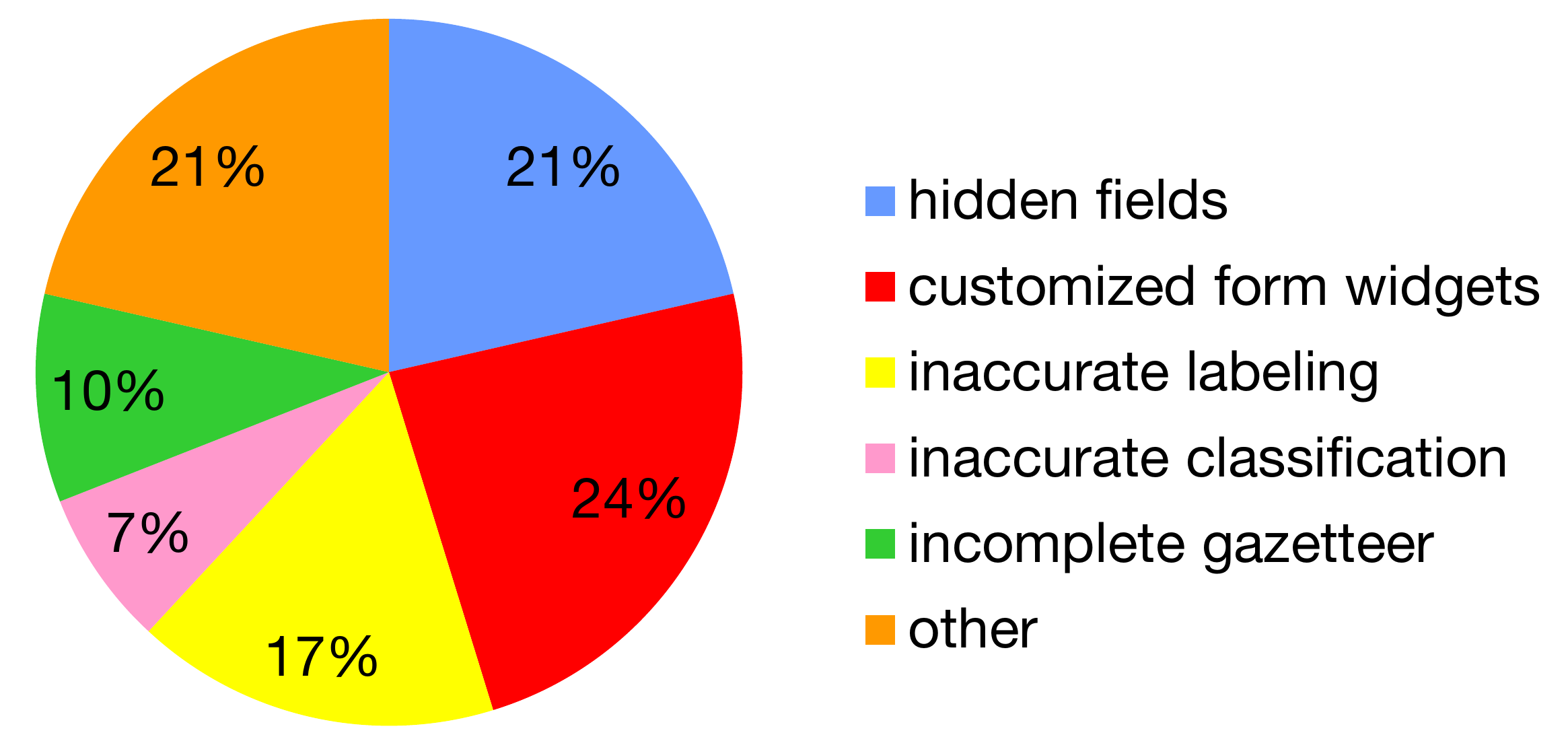} 
  \caption{Types and distribution of form integration errors}
  \label{fig:filling-error}
\end{figure}

\subsection{Form Integration}
For the evaluation of the form integration, we determine the error rate in the
query translation for all forms in the used car and real estate datasets. We use
multiple master queries in both cases, using for the real estate domain
combinations of location, min price, max price, and min bedroom. For the used
car domain, we use combinations of location, make, model, min price, and max
price. We evaluate the constraints separately and consider a constraint
correctly translated, if it involves the right field on the concrete form and
uses the best matching value. Overall, \OPAL generates 95.6\% and 93.8\%
correctly translated constraints. 

Figure~\ref{fig:filling-failure} presents the number of web forms where \OPAL
fails to translate one or more constraints incorrectly.  Overall, 87\% of the
forms were filled perfectly, and 95\% of the forms have no more than one
failure. Figure~\ref{fig:filling-error} presents the major causes for \OPAL's
failure in translating constraints: Most of the errors are caused by scripted
forms with hidden ($21\%$) or heavily customised form controls ($24\%$). The
remaining cases divide rather evenly between errors in the form labeling
($17\%$), in the classification or annotation (incomplete gazetteer), and an
assortment of other issues, mostly browser related (e.g., scripted popovers that
block access to the form fields or fields that can only be filled in a certain
order). 


\begin{figure}[tbp]
  \centering
  \includegraphics[width=.8\columnwidth]{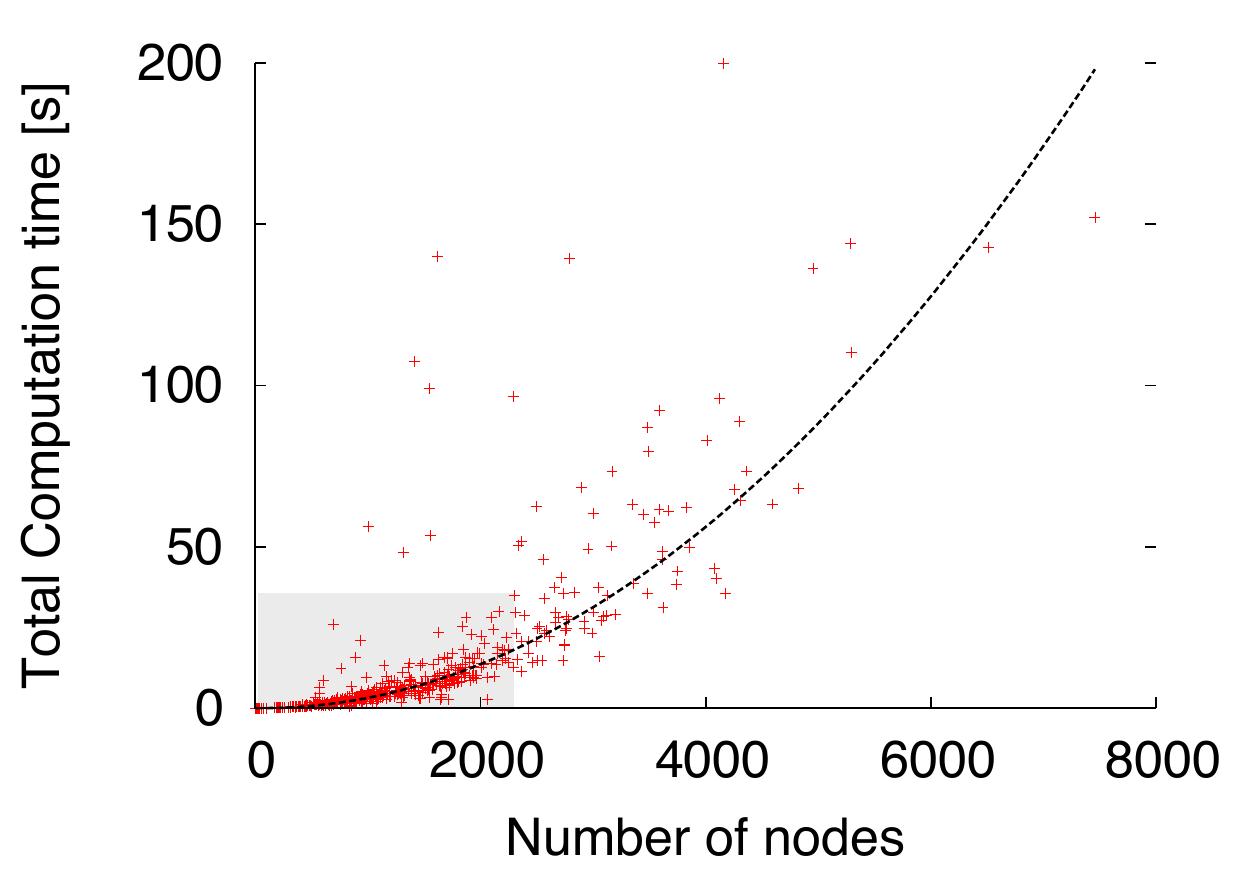} 
  \caption{Time}
  \label{fig:scaling-opal-nodes}
\end{figure}

\subsection{Scalability}
As discussed in Sections~\ref{sec:domain-independent} and
\ref{sec:domain-dependent}, overall the analysis of \OPAL is bounded by $O(n^2)$
due to the layout scope. As expected actual performance follows a quadratic
curve, but with very low constants. There is a significant amount of outliers,
partially due to long page rendering time and partially due to variance in the
depth and sophistication of the HTML
structure. Figure~\ref{fig:scaling-opal-nodes} reports \OPAL performance on all
534 forms in the combined TEL-8 and ICQ datasets. The highlight area covers
$80\%$ of the forms with 2200 nodes. \OPAL requires at most $30s$ for the
analysis (including page rendering) of these forms. Further analysis on the
effect of increasing field or form numbers confirms that these have little
effect and page size is the dominant factor.

\section{Related Work}
\label{sec:related-work}


Form understanding has attracted a number of approaches motivated by deep web
search~\cite{%
  madhavan08:_googl_deep_web_crawl,%
  raghavan01:_crawl_hidden_web,%
  shestakov2005deque}, meta-search engines and web form integration~\cite{%
  he07:_towar_deeper_unter_of_searc,%
  dragut09:_hierar_approac_to_model_web,%
  Wang:2004:ISM:1316689.1316726,%
  wu09:_model_and_extrac_deep_web_query_inter,%
  wu04:_inter_clust_based_approac_to,%
  zhen04:_under_web_query_inter} and web
extraction~\cite{su09:_ode,Wang:2003:DEL:775152.775179}. We focus here on
differences to \OPAL, for a complete survey see
\cite{khare10:_under_deep_web_searc_inter,DBLP:series/synthesis/2012Dragut}. We
present related work for form understanding and form integration separately, as
not all approaches consider both aspects.

\subsection{Form Understanding}
Form understanding approaches can be roughly categorised by the fundamental
approach to the problem: 

\begin{inparaenum}[\bfseries(1)]
\item The most common type encodes (mostly domain independent) observations on
  typical forms into implicit heuristics or explicit rules
  \tool{MetaQuerier}~\cite{zhen04:_under_web_query_inter},
  \tool{ExQ}~\cite{wu09:_model_and_extrac_deep_web_query_inter},
  \tool{SchemaTree}~\cite{dragut09:_hierar_approac_to_model_web},
  \tool{LITE}~\cite{raghavan01:_crawl_hidden_web},
  \tool{Wise-iExtractor}~\cite{he07:_towar_deeper_unter_of_searc},
  \tool{DEQUE}~\cite{shestakov2005deque}, and
  \tool{CombMatch}~\cite{kalijuvee01:_effic_web_form_entry_pdas}.
\item Alternatively, some approaches
  \tool{LabelEx}~\cite{nguyen08:_learn_to_extrac_from_label} and
  \tool{HMM}~\cite{DBLP:conf/cikm/KhareA09} use machine learning from a set of
  example forms (possibly of a specific domain).
\item Form understanding is often done to surface the results hidden behind the
  form and approaches such as
  \cite{madhavan08:_googl_deep_web_crawl,Wang:2004:ISM:1316689.1316726,raghavan01:_crawl_hidden_web}
  exploit the extracted results for form understanding. 
\end{inparaenum}

Aside of system design, \OPAL primarily differs from these approaches in two
aspects:
\begin{inparaenum}[\bfseries(1)]
\item They mostly incorporate only one or two of \OPAL's scopes (and
  feature classes): \tool{MetaQuerier}, \tool{ExQ}, and
  \tool{SchemaTree} mostly ignore the HTML structure (and thus field
  and segmentation scope) and rely on visual heuristics only;
  \tool{CombMatch}, \tool{LITE}, \tool{DEQUE}, and \tool{LabelEx} 
  ignore field grouping. \tool{HMM} ignores visual
  information. \cite{madhavan08:_googl_deep_web_crawl,Wang:2004:ISM:1316689.1316726,raghavan01:_crawl_hidden_web}
  use only the HTML structure, but exploit probing information, i.e., whether a
  submission is successful or not. 
\item None of the approaches provides a proper form model classifying the form
  fields according to a given schema. Furthermore, no approach uses domain
  knowledge is used to improve the labeling or verify the classification, though
  \tool{LabelEx} analyses domain specific term frequencies of label texts and
  \tool{HMM} checks for generic terms, such as ``min''.
\end{inparaenum}
As evident in our evaluation, each of the scopes in \OPAL considerably affects
the quality of the form labeling and classification. The fact, that each of
these approaches omits at least one of the domain-independent scopes, explains
the significant advantage in accuracy \OPAL exhibits on Tel-8 and ICQ. Notice
also that not using domain knowledge keeps these approaches out of reach of the
nearly perfect field classification achieved by \OPAL.

\paragraph{Form understanding by observation and heuristics.}
Most closely related in spirit to \OPAL, though very different in realisation
and accuracy, is \tool{MetaQuerier}~\cite{zhen04:_under_web_query_inter}. It is
built upon the assumption that web forms follow a ``hidden syntax'' which is
implicitly codified in common web design rules. To uncover this hidden syntax,
\tool{MetaQuerier} treats form understanding as a parsing problem, interpreting
the page a sequence of ``atomic visual elements'', each coming with a number of
attributes, in particular with its bounding box.  In a study covering $150$ forms, the
authors of \tool{MetaQuerier} identified $21$ common design patterns. These
patterns are captured by production rules in grammar with preferences. 
\tool{Metaquerier} is not parameterisable for a specific domain. In contrast,
the domain independent part of \OPAL achieves nearly perfect accuracy with only
$6$ generic patterns by combining visual, structural, and textual features, and
a simple prioritisation of these patterns by scope. \OPAL's domain dependent
part allows us to adjust it for patterns specific to a domain. 

\tool{ExQ}~\cite{wu09:_model_and_extrac_deep_web_query_inter} is similarly based
primarily on visual features such as a bias for the top-left located labels
comparable to \OPAL, but disregards most structural clues, such as explicit
\texttt{for} attributes of \texttt{label} tags and does not allow for any domain
specific patterns.

Also \tool{SchemaTree}~\cite{dragut09:_hierar_approac_to_model_web} uses only 
visual features (and the \texttt{tabindex} and \texttt{for} attributes for 
fields and labels). It exploits nine observations on form design, e.g., 
that query interfaces are organised top-down and left-to-right or that fields 
form semantic groups. It uses a hierarchical alignment between fields and text 
nodes similar to \OPAL's segment scope and a ``schema tree'' where the nine 
observations are observed. Again, no adaptation to a specific domain is possible.
%
%
%

\tool{Wise-iExtractor}~\cite{he07:_towar_deeper_unter_of_searc} firstly
tokenizes the form to obtain a high-level visual layout description (an
\emph{interface expressions (IEXP)}), distinguishing text fragments, form
fields, and delimiters, such as line breaks.  It then associates texts and
fields by computing the \emph{association weight} between any given field and
the texts in the same line and the two preceding lines, exploiting ending
colons, similarities between the text and the field's HTML name attribute, and
the text-field distance.  In addition, \tool{Wise} also identifies generic
relationships between fields: range (e.g.~from, to), part (e.g.~first and last
name), group (e.g.~radio buttons), or constraint (e.g.~exact match required).
However, in contrast to \OPAL their form labeling only explores limited visual
and textual information relying mainly on weight computation. Moreover, their
domain-independent typing shares some similarities with \OPAL's templates but
lacks the flexibility provided by \OPAL's domain schemata that allow us to
adjust these generic types to a given domain. Though these adjustments are often
small, their impact is significant, as shown in Section~\ref{sec:evaluation}.



In \cite{Yuan:2009:USI:1605395.1605889}, a (manually derived) domain schema is
used to guide understanding. In contrast to \OPAL, it segments a form purely
based on the domain schema (called schema tree). They evaluate on a fragment
(around 100-150 forms) of TEL-8 using domain schemata derived from the rest of
TEL-8 (about 250 forms). This yields on the considered fragment similar accuracy
as \OPAL achieves on the full TEL-8, yet \OPAL does not use any domain schema in
this case, let alone domain schemata specifically trained on TEL-8.

\paragraph{Form understanding by learning from example forms.}
%
Where the above approaches rely on humans to derive heuristics and rules for
form understanding, the following approaches use machine learning on a set of
example forms. Therefore, they can also be trivially adapted to a specific
domain by using domain-specific training data. The evaluation in
\cite{DBLP:conf/cikm/KhareA09}, however, shows little effect of domain-specific
training data: a training set from the biological domain outperforms
domain-specific training set in four out of five other domains.

\tool{LabelEx}~\cite{nguyen08:_learn_to_extrac_from_label} uses limited domain
knowledge when considering the occurrence frequencies of label terms. Domain
relevance of the terms occurring in a label, measured as the occurrence
frequency in previous forms, is one feature used to score field-label
candidates. Field-label candidates are otherwise created primarily using
neighbourhood and other visual features, as well as their HTML markup.
%
However, \tool{LabelEx} does not consider field groups and thus is unable to
describe segments of semantically related fields or to align fields and labels
based on the group structure and does not use any domain knowledge aside of term
frequency. 


\tool{HMM}~\cite{DBLP:conf/cikm/KhareA09} uses predefined knowledge on
typical terms in forms, such as ``between'', ``min'', or ``max'', but does
not adapt these for a specific domain.
\tool{HMM} employs two hidden Markov models to model an ``artificial
web designer''. 
During form analysis, the HMMs are used to explain the phenomena observed on
the page: The state sequences, that are most likely to produce the given web
form, are considered explanations of the form. Compared to \OPAL, \tool{HMM}
uses no visual features and no domain knowledge.


\paragraph{Form understanding by probing.}
All the above approaches conduct their analysis based purely on information
available on the web forms. Alternatively, there is also an indirect route for
form understanding by submitting the forms and analysing the query results,
which often are much easier to classify (as there are many instances compared to
a single form). The price is, however, that a certain amount of analysis of
those result pages is necessary. Therefore, this is primarily used in a context
where such analysis is anyway required, e.g., in crawlers or data extraction
systems. Typically, these approaches use an incremental approach, identifying
inputs for some fields, submitting the form, analysing the result page, and then
possibly restarting the whole process, now with, e.g., an increased set of input
values for the form. For example, \cite{madhavan08:_googl_deep_web_crawl}
determines whether a field must be filled or is a ``free'' input by iterating
over possible templates and selecting those that return sufficiently distinct
result pages. This is driven by the desire to surface some representative, but
not necessarily complete set of results from the web
form. 
None of these approaches produces a sophisticated form model, but at best rough
classifications of the fields and whether they are mandatory.

\subsection{Form Filling and Integration}

Form integration has been considered in many shapes, either as ``meta-search''
where a master query on a given global schema is translated to concrete forms as
in \OPAL, as ``interface matching'' where many concrete forms are integrated
without a global schema (often using schema matching), or as ``query
generation'' in the context of data extraction or crawling where the aim is to
generate a set of queries to extract all or most of the data, but often not even
full form understanding is performed.

Though some query generation and most interface matching approaches use form
understanding, they are focused on different issues than \OPAL's form
integration which is a type of ``meta-search'': How to find an optimal query set
that uncovers as much deep content as
possible~\cite{Barbosa04siphoninghidden-web}, how to determine if a query will
produce relevant data even if only partial information about the data is
available~\cite{Benedikt:2011:DRA:1989284.1989309}, how to maximize the
diversity of the extracted content~\cite{madhavan08:_googl_deep_web_crawl}, or
how to identify semantic equivalences between fields from different
forms~\cite{Nguyen:2010:PPS:1871437.1871627}.

Similar to \OPAL, \cite{Araujo:2010:CDA:1884110.1884134} fills web forms by
connecting fields at the conceptual level, but with
WordNet~\cite{Pedersen:2004:WMR:1614025.1614037} instead of proper
annotations. Furthermore, \OPAL produces more structured form model that is
verified against a domain
schema. Metaquerier~\cite{Chang:2004:MSL:1046456.1046465}, targets the
integration of web sources and tackles query translation for form filling in
that context. As \OPAL, Metaquerier selects values closest to the constraint in
the source query (similar to our master query). They also perform type-based
query translation to map a source query to a target query considering numeric
and text types, but achieve only $87\%$ accuracy.  \OPAL performs form filling
in a similar fashion, but also considers the number of fields for each domain
type in the master query and performs significantly better ($93\%$).

\section{Conclusion and Future Work}
\label{sec:conclusion}

To the best of our knowledge, \OPAL is the first comprehensive approach to form
understanding and integration. Previous form understanding approaches has been
limited mainly by overly generic, domain independent, monolithic algorithms
relying on narrow feature sets.
\OPAL pushes the state of the art significantly, addressing these limitations
through a very accurate domain independent form labeling, exploiting visual,
textual, and structural features, by itself already outperforming existing
approaches.
This domain independent part is complemented with a domain dependent form field
classification that significantly improves the overall quality of the form
understanding and produces nearly perfect form interpretations.
Accurate form interpretations enables form integration: \OPAL successfully
realizes a light-weight form integration system, able to translate  master queries to forms of a domain with nearly no errors.
  
Nevertheless, there remain open issues in \OPAL and form understanding in
general that need to be addressed for form understanding to become a reliable
tool to access web data through forms with little more effort than through APIs:

\begin{asparaenum}[\bfseries(1)]
\item \textbf{Dynamic, scripted forms:} \OPAL is able to understand most static
  forms with near perfect accuracy, but performs much worse on dynamic forms. We
  are already working on an extension of \OPAL for dealing with dynamic, heavily
  scripted interfaces that combines ideas from state exploration and crawling
  with form understanding. 
\item \textbf{Customised form widgets:} More and more forms use customised
  widgets such as tree views or sliders. Though most of these cases use hidden
  form fields that can be analysed by \OPAL, the use of fully scripted cases
  increases. We plan to extend \OPAL to allow the customisation of the form
  widgets that it can recognise. However, if these cases become more common, it
  may become necessary to automatically explore and learn such new widget
  types. 
\item \textbf{Probing-based understanding:} One of \OPAL's virtues is that it
  achieves its near perfect accuracy without any probing, but thus from the form
  page alone. However, this also limits the information that \OPAL can provide, 
  and prevents the verification and repair of the form model through the results
  returned by a form submission. For applications that need to access the result
  pages (e.g., data extraction and surfacing), we plan to integrate \OPAL
  with the result page analysis system AMBER
  \cite{furche11:_littl_knowl_rules_web} to further improve accuracy and to
  address integrity and access constraints.
\item \textbf{Integrity and access constraints.} \OPAL produces some integrity
  constraints through the domain schema and it's form segmentation, e.g.,
  dependencies between min and max fields in a range segment. We see an increase
  in the use of integrity constraints in forms thanks to the availability of
  easy-to-use client-side validation libraries.  Light-weight methods for
  analysing and exploiting such client side validation would allow us to extend our form models
  with more detailed integrity constraints. This is in addition to integrity and
  access constraints derived from probing. 
\end{asparaenum}

\section{Acknowledgements}
\ACKNOWLEDGEMENTS




\end{document}